\documentclass[twoside,11pt]{article}

\PassOptionsToPackage{dvipsnames}{xcolor}

\usepackage{multirow,tikz,amsmath,comment}
\usepackage{algorithm}
\usepackage{algorithmic}
\usepackage{subfig}
\usepackage{enumitem}
\usepackage{epstopdf}
\usepackage{tabularx,booktabs}
\usepackage{lastpage}
\usepackage{caption}
\usepackage{microtype}
\usepackage{blindtext}
\usepackage{array} 
\usepackage{multirow} 
\usepackage{makecell}

\usepackage[preprint]{jmlr2e}

\usepackage{tikz,pgfplots}
\usepackage{amsmath}
\usepackage{mathtools}
\usepackage{amssymb}

\usepackage{threeparttable}

\usepackage{tikz,pgfplots}
\usepackage{amsmath}
\usepackage{mathtools}
\usepackage{amssymb}

\usepackage{threeparttable}

\usepackage{float}
\usepackage{wrapfig}
\usepackage{bbm}

\usepackage{graphicx}   
\usepackage{subfig}     
\usepackage{caption}    

\newtheorem{Theorem}{Theorem}
\newtheorem{Proposition}{Proposition}

\newtheorem{Definition}{Definition}
\newtheorem{Corollary}{Corollary}
\newtheorem{Assumption}{Assumption}
\newtheorem{Example-set}{Example}

\newtheorem{Remark}{Remark}

\usepackage{bbding}

\usepackage{centernot}

\newcommand{\CI}{\mathrel{\perp\mspace{-10mu}\perp}}
\newcommand{\nCI}{\centernot{\CI}}

\newcommand{\ie}{i.e.}

\usetikzlibrary{fit,matrix,positioning, 
decorations.pathreplacing,calc,decorations,arrows,shadows,patterns}

\usetikzlibrary{arrows,fit,positioning,shapes,matrix}

\usetikzlibrary{backgrounds,shapes.geometric} 

\makeatletter
\newcommand{\srcsize}{\@setfontsize{\srcsize}{4.5pt}{4.5pt}}

\makeatother

\mathchardef\mhyphen="2D

\tikzset{
mycirclestyle/.style={
  circle,
  thin,
  inner sep=0pt,
  text width=6mm,
  minimum size=4mm,
  inner sep=0pt,
  align=center,
  draw=gray,
  fill=gray
  }
}

\tikzstyle{main node}=[circle,draw,minimum size=0.5mm]

\newdimen\arrowsize
\pgfarrowsdeclare{arcsq}{arcsq}
{
	\arrowsize=0.2pt
	\advance\arrowsize by .5\pgflinewidth
	\pgfarrowsleftextend{-4\arrowsize-.5\pgflinewidth}
	\pgfarrowsrightextend{.5\pgflinewidth}
}
{
	\arrowsize=0.8pt
	\advance\arrowsize by .5\pgflinewidth
	\pgfsetdash{}{0pt} 
	\pgfsetroundjoin   
	\pgfsetroundcap    
	\pgfpathmoveto{\pgfpoint{0\arrowsize}{0\arrowsize}}
	\pgfpatharc{-90}{-140}{4\arrowsize}
	\pgfusepathqstroke
	\pgfpathmoveto{\pgfpointorigin}
	\pgfpatharc{90}{140}{4\arrowsize}
	\pgfusepathqstroke
}

\newcommand*\samethanks[1][\value{footnote}]{\footnotemark[#1]}

\newcommand{\revise}[1]{{\color{black}#1}}

\newcommand{\tworevise}[1]{{\color{black}#1}}

\DeclareMathOperator{\Cov}{Cov}
\DeclareMathOperator{\Var}{Var}

\usepackage{lastpage}
\jmlrheading{23}{2022}{1-\pageref{LastPage}}{1/21; Revised 5/22}{9/22}{21-0000}{Author One and Author Two}

\ShortHeadings{Testability of IV in  Additive Nonlinear, Non-Constant Effects Models}{Guo, Li, Huang, Zeng, Geng, and Xie}
\firstpageno{1}

\begin{document}

\title{Testability of Instrumental Variables in Additive Nonlinear, Non-Constant Effects Models}

\author{\name Xichen Guo\thanks{Equal contribution} \email   
guoxichen0@gmail.com\\
\addr Department of Applied Statistics, Beijing Technology and Business University\\Beijing, 102488, China
            \AND
 \name Zheng Li\samethanks \email   zhengli0060@gmail.com\\ 
 \addr Department of Applied Statistics, Beijing Technology and Business University\\Beijing, 102488, China
 \AND
 \name Biwei Huang \email   bih007@ucsd.edu\\ 
 \addr Halicioglu Data Science Institute (HDSI), University of California San Diego\\
 La Jolla, San Diego, California, 92093, USA
        \AND
        \name Yan Zeng  \email yanazeng013@btbu.edu.cn\\
 \addr Department of Applied Statistics, Beijing Technology and Business University\\Beijing, 102488, China
            \AND
            \name Zhi Geng \email zhigeng@pku.edu.cn\\
 \addr Department of Applied Statistics, Beijing Technology and Business University\\ Beijing, 102488, China\\
 \addr School of Mathematical Sciences, Peking University\\ Beijing, 100871, China
		\AND 
            \name Feng Xie \thanks{Corresponding author} \email     fengxie@btbu.edu.cn\\
 \addr Department of Applied Statistics, Beijing Technology and Business University\\Beijing, 102488, China}
\editor{My editor}

\maketitle

\begin{abstract}
We address the issue of the testability of instrumental variables derived from observational data. Most existing testable implications are centered on scenarios where the treatment is a discrete variable, e.g., instrumental inequality \citep{pearl1995testability}, or where the effect is assumed to be constant, e.g., instrumental variables condition based on the principle of independent mechanisms \citep{burauel2023evaluating}. However, treatments can often be continuous variables, such as drug dosages or nutritional content levels, and non-constant effects may occur in many real-world scenarios. In this paper, we consider an additive nonlinear, non-constant effects model with unmeasured confounders, in which treatments can be either discrete or continuous, and propose an Auxiliary-based Independence Test (AIT) condition to test whether a variable is a valid instrument. We first show that, under the completeness condition, if the candidate instrument is valid, then the AIT condition holds. Moreover, we illustrate the implications of the AIT condition and demonstrate that, under certain additional conditions, the AIT condition is necessary and sufficient to detect all invalid IVs. We also extend the AIT condition to include covariates and introduce a practical testing algorithm. Experimental results on both synthetic and three different real-world datasets show the effectiveness of our proposed condition. 
\end{abstract}


\section{Introduction}
Estimating causal effects from observational data is a fundamental task in understanding the underlying relationships between variables. The instrumental variables (IV) model is a well-established method for estimating the causal effect of a treatment (exposure) $X$ on an outcome $Y$ in the presence of unmeasured confounders and has been used in a range of fields, such as economics~\citep{Imbens2014IV,imbens2015causal}, sociology \citep{pearl2009causality,spirtes2000causation}, epidemiology~\citep{hernan2006instruments,baiocchi2014instrumental}, and artificial intelligence \citep{chen2022instrumental,wu2022instrumental}. Generally speaking, given a causal relationship $X \to Y$, the valid IV $Z$ is required to satisfy the following three conditions: $\mathcal{C}1$. $Z$ is related to the treatment (\emph{relevance}), $\mathcal{C}2$. $Z$ is independent of the unmeasured confounders that affect the treatment and outcome (\emph{exogeneity}), and $\mathcal{C}3$. $Z$ has no direct path to the outcome (\emph{exclusion restriction}). Figure \ref{Fig:main-example-graph} illustrates the graphical criteria of the IV model, where $Z$ is a valid IV relative to $X \to Y$ in the subgraph (a).

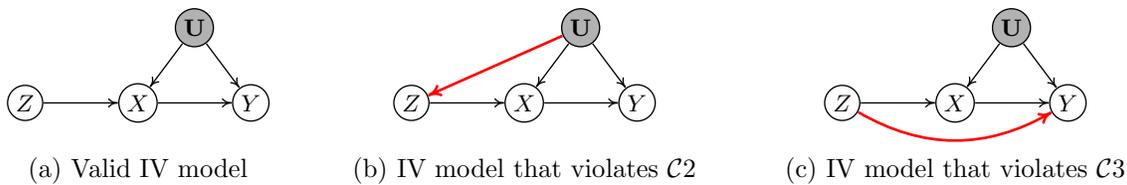
\begin{figure*}[htp]
    \centering
    \begin{tikzpicture}[scale=1.0, line width=0.75pt, inner sep=0.2mm, shorten >=.1pt, shorten <=.1pt]
	\draw (2.75, 1.0) node(U) [circle, dash pattern=on 2.5pt off 2pt, fill=gray!30, draw] {{\footnotesize\,$\mathbf{U}$\,}};
	\draw (0.5, 0) node(Z) [circle, draw] {{\footnotesize\,$Z$\,}};
	\draw (2, 0) node(X) [circle, draw] {{\footnotesize\,{$X$}\,}};
    \draw (3.5, 0) node(Y) [circle, draw] {{\footnotesize\,{$Y$}\,}};
	\draw[-arcsq] (U) -- (X) node[pos=0.5,sloped,above] {};
	\draw[-arcsq] (U) -- (Y) node[pos=0.5,sloped,above] {};
	\draw[-arcsq] (Z) -- (X) node[pos=0.5,sloped,above] {{\scriptsize\,{}\,}};
	\draw[-arcsq] (X) -- (Y) node[pos=0.5,sloped,above] {};
    \draw (2.0, -0.9) node(con3) [] {{\small\,(a) Valid IV model\,}};
    \end{tikzpicture}~~~~~~~
    \begin{tikzpicture}[scale=1.0, line width=0.75pt, inner sep=0.2mm, shorten >=.1pt, shorten <=.1pt]
	\draw (2.75, 1.0) node(U) [circle, dash pattern=on 2.5pt off 2pt, fill=gray!30, draw] {{\footnotesize\,$\mathbf{U}$\,}};
	\draw (0.5, 0) node(Z) [circle, draw] {{\footnotesize\,$Z$\,}};
	\draw (2, 0) node(X) [circle, draw] {{\footnotesize\,{$X$}\,}};
	\draw (3.5, 0) node(Y) [circle, draw] {{\footnotesize\,{$Y$}\,}};
	\draw[-arcsq] (U) -- (X) node[pos=0.5,sloped,above] {};
	\draw[-arcsq] (U) -- (Y) node[pos=0.5,sloped,above] {};
        \draw[-arcsq, line width =1.0pt, color=red] (U) -- (Z) node[pos=0.5,sloped,above,text=black] {{\scriptsize\,{}\,}};
	\draw[-arcsq] (Z) -- (X) node[pos=0.5,sloped,above] {{\scriptsize\,{}\,}};
	\draw[-arcsq] (X) -- (Y) node[pos=0.5,sloped,above] {};
    \draw (2.0, -0.9) node(con3) [] {{\small\,(b) IV model that violates {$\mathcal{C}2$}\,}};
    \end{tikzpicture}~~~~~~~
    \begin{tikzpicture}[scale=1.0, line width=0.75pt, inner sep=0.2mm, shorten >=.1pt, shorten <=.1pt]
	\draw (2.75, 1.0) node(U) [circle, dash pattern=on 2.5pt off 2pt, fill=gray!30, draw] {{\footnotesize\,$\mathbf{U}$\,}};
	\draw (0.5, 0) node(Z) [circle, draw] {{\footnotesize\,$Z$\,}};
	\draw (2, 0) node(X) [circle, draw] {{\footnotesize\,{$X$}\,}};
	\draw (3.5, 0) node(Y) [circle, draw] {{\footnotesize\,{$Y$}\,}};
	\draw[-arcsq] (U) -- (X) node[pos=0.5,sloped,above] {};
	\draw[-arcsq] (U) -- (Y) node[pos=0.5,sloped,above] {};
	\draw[-arcsq] (Z) -- (X) node[pos=0.5,sloped,above] {{\scriptsize\,{}\,}};
	\draw[-arcsq] (X) -- (Y) node[pos=0.5,sloped,above] {};
        \draw[-arcsq, line width =1.0pt, color=red]
    (Z) edge[bend right=30] (Y);
        
    \draw (2.0, -0.9) node(con3) [] {{\small\,(c) IV model that violates {$\mathcal{C}3$}\,}};
    \end{tikzpicture}  
    \caption{Graphical illustration of IV models, where $\mathbf{U}$ is the set of unmeasured confounders. (a) $Z$ is a valid IV. (b) $Z$ is an invalid IV due to the edge $\mathbf{U} \to Z$ (Violate {$\mathcal{C}2$}). (c) $Z$ is an invalid IV due to the edge $Z \to Y$ (Violate {$\mathcal{C}3$}).}
    \label{Fig:main-example-graph}
\end{figure*}

Due to the presence of unmeasured confounders $\mathbf{U}$, determining which variable serves as a valid IV is not always straightforward based solely on observational data, and often requires domain knowledge. A classic test for the IV model is the Durbin-Wu-Hausman test \citep{nakamura1981relationships}. Given a subset of valid IVs, it can identify whether other potential candidates are also valid IVs. However, it does not guide how to find the initial set of valid IVs. 
Meanwhile, given an invalid IV, it may not consistently identify the correct causal effect~\citep{bound1995problems,chu2001semi}. Thus, it is vital to develop statistical methods for selecting IVs solely from observational data.

It is not feasible to ascertain the validity of IVs solely based on the joint distribution of observed variables, without incorporating additional assumptions \citep{pearl2009causality}. \citet{pearl1995testability} introduced a seminal necessary criterion known as the \emph{instrumental inequality}, which acts as a critical test for identifying potential IVs in models featuring discrete variables. Building on this groundwork, subsequent research by ~\citet{manski2003partial,palmer2011nonparametric,kitagawa2015test,wang2017falsification} broadened the scope, exploring the applicability and limitations of IV validity tests across diverse scenarios. A notable advancement was made by \citet{kedagni2020generalized}, who formulated a more encompassing set of criteria, the \emph{generalized instrumental inequalities}. These criteria cater to scenarios with discrete treatment variables, removing constraints of the type of data on the IV and outcome variables and offering a robust framework to challenge the exogeneity condition. Intuitively, these methods mentioned above use the idea that if the IV is independent of the unmeasured confounders and the \emph{exclusion restriction} ($\mathcal{C}3$) holds, then changes in the IV should not have a significant impact on the outcome variable without altering the treatment variable, because the treatment variable mediates the influence of the instrumental variable on the outcome variable. However, these methods fail to work when treatment is a continuous variable. In reality, one may often be concerned about the causal effect of the continuous treatment on the outcome; see~\citet{skaaby2013vitamin, martinussen2019instrumental}.

Several contributions have been made to address continuous treatment settings under certain assumptions. In an additive linear, constant effects (ALICE) model, \citet{kang2016instrumental,bowden2016consistent}, and \cite{windmeijer2019use} have shown that if we assume more than half of the variables are valid IVs in the potential IVs (known as the majority rule), one may identify the valid IVs solely from observed data. Later, \cite{hartwig2017robust,guo2018confidence,windmeijer2021confidence} relaxed the majority rule and assumed that the number of valid IVs is larger than any number of invalid IVs with the same ratio estimator limit (known as the plurality rule). They demonstrated that it is still possible to identify valid IVs under the plurality rule. Another interesting work by \cite{silva2017learning} proposed the IV-TETRAD algorithm, which uses the so-called Trek conditions~\citep{Sullivant-T-separation, Spirtes:2013:Extended-Trek-Theorem} for selecting a valid IV set. This method requires at least two or more valid IVs in the system. However, although these methods have been used in a range of fields, they may fail to test whether a single IV is valid. 

Recently, \citet{xie2022testability} have demonstrated that a single IV imposes specific constraints within the linear non-Gaussian acyclic causal model. However, their method assumes that all noise terms are non-Gaussian and that the effects remain constant. More recently, \citet{burauel2023evaluating} have introduced a novel validity condition for Instrumental Variables based on the Principle of Independent Mechanisms, termed IV-PIM, within the linear IV framework. This condition is particularly notable as it applies to both continuous and discrete treatment variables. Nevertheless, its applicability is constrained by the presence of covariates, making it unsuitable when no covariates are available. Additionally, the condition is limited to scenarios with constant treatment effects.

\citet{pearl1995testability} conjectured that the validity of an instrument cannot be tested when dealing with continuous treatment variables without additional assumptions, a theory recently confirmed by \citet{gunsilius2021nontestability}. Unlike the existing work that focuses on the parametric linear constant effects model, we consider IV validity in a more challenging additive non-parametric model, the Additive NonlInear, Non-Constant Effects (ANINCE) Model. Rather surprisingly, although a single IV is in general not fully testable within the ANINCE model, we will show that a single variable $Z$, being a valid IV relative to $X \rightarrow Y$, imposes specific constraints in the ANINCE model under the completeness condition (\ie, for all functions $\psi(X)$ with finite expectation, $\mathbb{E}[\psi(X)|Z]=0$ implies $\psi(X)=0$.) \footnote{See Section \ref{Section-IV-Identification-Estimator} for more detailed discussion of the completeness condition. For related context, we refer interested readers to \cite{newey2003instrumental,d2011completeness,hu2018nonparametric}.}. Specifically, we make the following contributions:

\begin{itemize}
     \item [1.] We introduce a necessary condition, termed the Auxiliary-based Independence Test (AIT) condition, for detecting a single variable that cannot serve as an IV within the ANINCE model, under the completeness condition (Assumption \ref{Ass-completeness}). This condition is applicable to scenarios involving non-constant causal effects and both discrete and continuous treatment.
    \item [2.] We provide the necessary and sufficient conditions for detecting all invalid IVs using the AIT condition under the ANINCE model. Specifically, we show that, under the partial non-Gaussianity assumption (Assumption \ref{Ass-non-Gaussianity-exogeneity}), all observable violations of the IV exogeneity condition can be identified in the linear, constant effects model. Additionally, under \revise{the completeness condition (Assumption \ref{Ass-completeness}) and} \tworevise{the distributional non-degeneracy condition assumption} (Assumption \ref{Ass-algebraic-condition}), we can detect invalid IVs resulting from violations of either exogeneity or the exclusion restriction in the ANINCE model. We also present two notable types of non-identifiable invalid IVs (\revise{see Proposition \ref{Pro-violate-algebraic}}), along with intuitive explanations for each. 
    
    \item [3.] We present a practical implementation of the AIT condition test that accounts for the presence of covariates with finite data. \tworevise{We establish its asymptotic validity, including control of Type I error and consistency (power tending to one) against alternatives.}
    
    \item [4.] We demonstrate the efficacy and applicability of the proposed approach on both synthetic and three real-world datasets with different scenarios.

\end{itemize}

The rest of this paper is organized as follows. In Section \ref{Preliminaries}, we introduce notations, the additive non-parametric IV model, and the ANINCE model. In Section \ref{Section-AIT-condition-and-implication}, we formulate the AIT condition for the single IV. We show the AIT condition is a necessary condition for IV validity in the ANINCE model \revise{under the completeness condition}. We discuss the implications of AIT condition in the linear, constant effects model and the nonlinear, non-constant effects model, respectively. We show that, under additional assumptions, the AIT condition is a necessary and sufficient condition for IV validity. \tworevise{In Section \ref{Section-Test-Procedure}, we address the practical scenario with covariates and provide the \emph{AIT Condition} algorithm for implementing the test. We further provide a theoretical analysis establishing the asymptotic level and  power of the AIT test.} In Section \ref{Section:experiments}, we present the efficacy and applicability of our method on both synthetic and three real-world datasets which contain continuous and discrete data. Conclusions are given in Section \ref{Section-conclusions}. 

\section{Preliminaries}\label{Preliminaries}
\subsection{Notations}
This work is conducted within the framework of causal graphical models as elaborated by \citet{pearl2009causality} and \citet{spirtes2000causation}. Specifically, we represent causal relationships using the directed acyclic graph (DAG), denoted as $\mathcal{G}$, where nodes represent variables and directed edges (arrows) indicate causal links between those variables. Sets of variables are represented in bold, and individual variables and symbols for graphs are in italics. We use ``instrumental variable (IV)'' and ``instrument'' interchangeably. The main symbols used in this paper are summarized in Table \ref{Table-symbols}.

\begin{table*}[t!]
    \centering
    \caption{List of main symbols used in this paper}
    \label{Table-symbols}
    \small
    \begin{tabular}{p{2.4cm}p{11.6cm}}
    \toprule
    \textbf{Symbol} & \textbf{Description}\\ 
    \midrule
    $\mathcal{G}$   & A directed acyclic graph  \\
    $X$    & Treatment (exposure)  \\
    $Y$    & Outcome   \\
    $Z$    & A candidate (potential) instrument   \\
    $\mathbf{U}$ & The latent (unmeasured) confounders \\
    $\mathbf{W}$    & Covariates   \\
    $\mathcal{Z}$ & The residual of $Z$ after regressing on covariates $\mathbf{W}$ \\
    IV    & Instrumental Variable \\
    ${A} \CI {B} | {C}$  & $A$ is statistically independent of $B$ given $C$ \\
    ${A} \nCI {B} | {C}$ & $A$ is statistically dependent on $B$ given $C$ \\
    \revise{$\{X, Y||Z\}$} & \revise{The candidate instrumental variable $Z$ used to assess the independence of the auxiliary variable $\mathcal{A}_{X \to Y||Z}=Y-h(X)$} \\
    $|\mathbf{W}|$  & The number of variables in set $\mathbf{W}$\\
    $f(X, Z)$   & The causal effect of $X$ and $Z$ on $Y$  \\
    $\widetilde{f}_{bias}(X, Z)$ & The bias between estimated causal effect of $X$ on $Y$ and ground-truth causal effect of $(X, Z)$ on $Y$ \\
    $\varepsilon_{*}$ & The noise term of a variable \\
    $\varphi_{*}(\mathbf{U})$ & The effect of the latent variables $\mathbf{U}$ on the observed variables \\
    $g_{*}(Z)$ & The effect of the instrument variable $Z$ on other observed variables \\
    $\mathbb{R}$ &  The field of real numbers\\
    $\mathbb{R}\rightarrow\mathbb{R}$ & A mapping from the real numbers to the real numbers \\
    $\mathcal{I}(*)$  & The indicator function \\
    $\mathbb{E}(X)$ & The expected value of random variable $X$\\
    $\frac{\partial^2 Y}{\partial X \partial Z} $ & The second-order partial derivative of $Y$ with respect to $X$ and $Z$\\
    $\mathcal{A}_{X \to Y||Z}$   & The auxiliary variable of causal relationship $X \to Y$ relative to $Z$. We often drop the subscript ${X \to Y||Z}$ when there is no ambiguity (\ie, $\mathcal{A}$)\\
    $\hat{h}(X, \mathbf{W})$ & The empirical estimate of the function $h(X, \mathbf{W})$ \\
    $\hat{\mathcal{A}}_{X \to Y \| Z}$ & The estimated auxiliary variable computed as $Y - \hat{h}(X, \mathbf{W})$. We often use $\hat{\mathcal{A}}$ as a shorthand when there is no ambiguity \\ 
    $k, l$ & The kernels $k: \mathcal{U} \times \mathcal{U} \to \mathbb{R}$ and $l: \mathcal{V} \times \mathcal{V} \to \mathbb{R}$ used for variables in $\mathcal{U}$ and $\mathcal{V}$, respectively \\
    K-test method  & Kitagawa's method from \cite{kitagawa2015test} \\
    IV-PIM method &  Burauel's method from \cite{burauel2023evaluating} \\
    \bottomrule
    \end{tabular}
\end{table*} 

\subsection{Additive Nonparametric Instrumental Variable Model}\label{Section-IV-Identification-Estimator}
The instrumental variable approach offers a strategy for inferring the causal effect of interest in the presence of unmeasured confounders \citep{bowden1990instrumental,angrist1996identification,pearl2009causality,imbens2015causal}. 
Given a causal relationship $X \to Y$, a valid IV $Z$ is required to satisfy the following three conditions: 
\begin{description}
    \item [$\mathcal{C}1$. (Relevance).] $Z$ has \revise{a direct effect on} the treatment $X$; 
    \item [$\mathcal{C}2$. (Exogeneity or Randomness).] $Z$ is independent of the unmeasured confounders $\mathbf{U}$; 
    \item [$\mathcal{C}3$. (Exclusion Restriction).] $Z$ does not directly affect the outcome $Y$.
\end{description}
\begin{Definition}
    A random variable $Z$ is a valid IV for the causal relationship $X \to Y$ if the above three conditions $\mathcal{C}1 \sim \mathcal{C}3$ are satisfied.
\end{Definition}
We here consider the additive non-parametric IV model presented in \cite{newey2003instrumental}, which, for a valid IV $Z$, can be expressed as follows \footnote{Here, we have slightly modified the model from \cite{newey2003instrumental} to explicitly represent the unmeasured confounders for subsequent analysis.}:
\begin{equation}\label{Eq-IV-Model}
    \begin{aligned}
        X &= g(Z) + \underbrace{\varphi_{X}(\mathbf{U}) + \varepsilon_{X}}_{\delta},\\
        Y &= f(X) + \underbrace{\varphi_{Y}(\mathbf{U}) + \varepsilon_{Y}}_{\epsilon},  
    \end{aligned}
\end{equation}
where $\mathbb{E}[\varphi_Y(\mathbf{U}) + \varepsilon_Y|Z]=0$, and the noise terms $\varepsilon_{X}$ and $\varepsilon_{Y}$ are statistically independent. 

\tworevise{
We here would like to mention that in nonparametric IV models, identifying the function $f(\cdot)$ amounts to solving the conditional moment restriction $\mathbb{E}[Y - f(X) \mid Z] = 0$, which can be formulated as a linear inverse problem of the form $Tf = q$, where $Tf := \mathbb{E}[f(X) \mid Z]$ and $q := \mathbb{E}[Y \mid Z]$. Under mild regularity conditions on the joint distribution of $(X,Z)$, the associated conditional expectation operator $T$ is compact. As a consequence, solving $Tf = q$ constitutes an ill-posed inverse problem, and the existence of a solution is not guaranteed for arbitrary $q$. Classical results from inverse problem theory show that existence is characterized by a Picard criterion \citep{kress1989linear,kress2013linear}\footnote{\tworevise{The Picard criterion provides necessary and sufficient conditions for the existence of a solution to $\mathbb{E}[Y-f(X)\mid Z]=0$. In particular, when the associated conditional expectation operator $T$ is compact, existence can be characterized in terms of 
the operator $T$ and the function $q$; see Theorem~15.18 in 
\citet{kress1989linear,kress2013linear} and the discussion in 
\citet{horowitz2012specification}.}}. As noted by \citet{darolles2011nonparametric}, restrictions that confine the function to a compact set are mathematically equivalent to enforcing the Picard criterion and thereby ensuring the existence of a solution to the associated inverse problem. Following common practice in the nonparametric IV literature (e.g., \citet{newey2003instrumental,carrasco2007linear,florens2011identification}), we restrict the true function $f(\cdot)$ to belong to a compact set of functions, which ensures the existence of a solution to the conditional moment restriction $\mathbb{E}[Y - f(X) \mid Z] = 0$, that is, $q$ lies in the range of the operator $T$. This condition is assumed to hold throughout the paper. 
}

\revise{
\begin{Assumption}{\normalfont(\textbf{Completeness Condition} \tworevise{\citep{newey2003instrumental,d2011completeness,hu2018nonparametric}})}\label{Ass-completeness}
Given a valid IV $Z$, if for all measurable functions $\psi(X)$ such that $\mathbb{E}[|\psi(X)|] < +\infty$, 
\begin{equation}
    \begin{aligned}
        \mathbb{E}[\psi(X)|Z]=0 \text{ almost surely} \Rightarrow \psi(X)=0 \text{ almost surely}.
    \end{aligned}
\end{equation}
\end{Assumption}
 
}

\revise{Intuitively, for a valid IV $Z$, completeness imposes a condition on the conditional probability density $k(X|Z)$ in the nonparametric IV model. Note that identifying the nonparametric IV model is generally nontrivial due to the challenges inherent in nonparametric settings. Consequently, the completeness condition has been extensively studied to establish identification in various nonparametric and semiparametric models, such as \citet{newey2003instrumental,ai2003efficient,hall2005nonparametric,chen2006identification,blundell2007semi,chernozhukov2007instrumental,carrasco2007linear,hu2008instrumental,carroll2010identification,d2011completeness,darolles2011nonparametric,an2012well,newey2013nonparametric,canay2013testability,shiu2013identification,feve2014non,andrews2017examples,hu2018nonparametric}. For instance, 
\citet{newey2003instrumental} have shown that, given a valid IV $Z$, the causal effect $f(\cdot)$ of interest in the model specified by Equation \eqref{Eq-IV-Model} can be \revise{uniquely identified} if the above completeness condition—\ie, Assumption \ref{Ass-completeness}—holds.} 

\begin{Remark}
\revise{Regarding the completeness condition, three key aspects are worth highlighting: }
\begin{enumerate}
    \item \revise{The testability of the completeness condition has attracted attention in recent literature. \cite{canay2013testability} have shown that the completeness condition is generally untestable without additional restrictions. Therefore, for more specific models, testable conditions have been proposed. For example, \cite{freyberger2017completeness} provided a test for restricted completeness by linking the outcome of the test to consistency of an estimator. \cite{hu2022simple} provided a useful result for testing the completeness condition in a class of models based on convolution. Other tests include full-rank tests for completeness in discrete settings, as proposed by \citet{robin2000tests}.}
    \item \revise{\cite{newey2003instrumental} have shown that the case of finite support and the exponential family setting constitute sufficient conditions for completeness. Many commonly used distributions—such as Gaussian, Poisson, Binomial, and certain multivariate forms of these—fall within the exponential family framework \citep{hu2018nonparametric}. Building on this, other sufficient conditions for completeness have been further developed in the literature \citep{d2011completeness,chen2014local,andrews2017examples,hu2018nonparametric}. Notably, \cite{d2011completeness} derived sufficient conditions for various forms of completeness of the endogenous variable $X$ given the instrument $Z$, and applied these results to nonparametric IV regression. Furthermore, \cite{hu2018nonparametric} provided sufficient conditions for completeness of the distribution of treatment conditional on the instrument, without relying on a specific functional form. }
    \item \revise{Under the completeness condition, estimation methods for nonparametric IV models have also been extensively studied \citep{newey2003instrumental,ai2003efficient,darolles2011nonparametric,chernozhukov2007instrumental,newey2013nonparametric,singh2019kernel,bennett2019deep}. Among them, \cite{newey2003instrumental} developed a nonparametric equivalent to the two-stage least squares estimator: they used linear-in-parameter series expansions of $\mathbb{E}[Y \mid Z]$ and $\mathbb{E}[g(X) \mid Z]$ in a generalized method of moments framework. \cite{darolles2011nonparametric} analyzed identification and overidentification of the nonparametric IV model, and proposed an estimator based on Tikhonov regularization to address the ill-posed inverse problem inherent in nonparametric IV estimation. A comprehensive review of several estimation techniques has been provided in \cite{carrasco2007linear}.}
\end{enumerate}
\end{Remark}

\subsection{Additive Nonlinear, Non-Constant Effects Model}
Without loss of generality, we assume that all variables have a zero mean (otherwise can be centered) and that no covariates are present for simplicity. In Section \ref{Section-Test-Procedure}, we address the practical scenario where covariates are included. In this paper, we focus our attention on the Additive NonlInear, Non-Constant Effects (ANINCE) Model. Specifically, the generation process satisfies the following structural causal model:
\begin{equation}\label{Eq-Main-Model}
    \begin{aligned}
        X &= g(Z) + \varphi_{X}(\mathbf{U}) + \varepsilon_{X},\\
        Y &= f(X, Z) + \varphi_{Y}(\mathbf{U}) + \varepsilon_{Y}, 
    \end{aligned}
\end{equation}
where $f(\cdot)$ denotes the true, unknown causal effect of interest, and 
$g(\cdot):\mathbb{R} \to\mathbb{R}$, $f(\cdot):\mathbb{R}^{2}\to\mathbb{R}$, and 
$\varphi_{*}(\cdot):\mathbb{R}^{|\mathbf{U}|}\to\mathbb{R}$ are smooth functions. The noise terms $\varepsilon_{X}$, $\varepsilon_{Y}$, $\varepsilon_{Z}$ and $\boldsymbol{\varepsilon_{U}}$ are mutually independent. Note that $Z$ and $\mathbf{U}$ may be dependent, which indicates that the \emph{exogeneity} condition ($\mathcal{C}2$) is violated, and the non-zero $f(\cdot, Z)$ function indicates that $Z$ directly affects the outcome $Y$, implying that the \emph{exclusion restriction} condition ($\mathcal{C}3$) is violated.

A special case of the ANINCE model is the additive linear, constant effects model (ALICE), where functions $g(\cdot)$, $f(\cdot)$, and $\varphi_{*}(\cdot)$ are linear functions, which have been extensively studied in works such as those by \citet{bowden2015mendelian,kang2016instrumental, silva2017learning,windmeijer2021confidence}. Compared to these works, we investigate the testability of IV in a more challenging scenario, where $g(\cdot)$, $f(\cdot)$, and $\varphi_{*}(\cdot)$ may be non-linear functions. Additionally, we focus on the testability of a single valid IV, whereas previous works have focused on the testability of a set of IVs (assuming that including at least two or more valid IVs among the candidate variables). 

\tworevise{
\begin{Remark}
It is worth noting that, unlike the usual additive nonparametric IV model (e.g., \citealp{newey2003instrumental}), the ANINCE model is formulated within the structural causal model (SCM) framework and therefore assumes mutually independent noise components. This SCM-based independence structure is essential for deriving additional distributional constraints that render IV validity testable; without it, valid instruments would not imply observable independence relations and no empirical test could be constructed. Moreover, the usefulness of independent-noise assumptions is well established in additive SCMs for causal discovery, where they enable identifiability of causal directions and have been validated across numerous applications \citep{hoyer2008nonlinear,peters2014causal,peters2014identifiability,buhlmann2014cam,peters2017elements,glymour2019review}. Consequently, our adoption of independent-noise terms follows standard practice in SCM-based causal analysis rather than departing from conventional modeling assumptions.
\end{Remark}

}

\noindent\textbf{Our Goal.} The goal of this paper is to determine, from the observed dataset $\{X, Y, Z\}$ satisfying an ANINCE model, whether $Z$ is related to $X$ (i.e., \emph{relevance} condition), $Z$ is exogenous relative to $(X, Y)$ (i.e., \emph{exogeneity} condition) and $Z$ does not directly affect outcome $Y$ (i.e., \emph{exclusion restriction} condition). Note that the first condition \emph{relevance}, can be easily checked by the independence test because $Z$ and $X$ are observed variables. Therefore, we focus on the last two conditions of IV $Z$. In summary, we aim to provide a new necessary condition to detect whether a variable is a valid IV and investigate the necessary and sufficient conditions under which all invalid IVs can be detected.

\begin{Remark}
\revise{Existing approaches have attempted to detect violations of the exogeneity for a single IV in settings with discrete variables, such as through \emph{instrumental inequality} and its extensions. In contrast to prior work, we focus on continuous-variable settings. Notably, instrument validity is generally untestable in continuous variable settings without additional assumptions. Hence, in this paper, we introduce additive function constraints, allowing us to detect the invalid IVs when the treatment variable is continuous.} 
\end{Remark}

\section{AIT Condition and Its Implications in ANINCE Models}\label{Section-AIT-condition-and-implication}
In this section, we first formulate the Auxiliary-based Independence Test condition (AIT condition) and show that it is a necessary condition for evaluating IV validity \revise{under the completeness condition (Assumption \ref{Ass-completeness})}. We further present theoretical results regarding the implications of the AIT condition in the linear, constant effects model and nonlinear, non-constant effects model. 

\subsection{AIT Condition}\label{section_AIT-condition}

Below, we give the AIT condition, which defines the independent relationship between the ``Auxiliary variable" and candidate IV. Note that the concept of ``auxiliary variable'' has been developed to address different tasks \citep{drton2004iterative,chen2017identification,cai2019triad}, but our formalization is different from theirs (see Equation \eqref{Eq-auxiliary-variable}). To the best of our knowledge, it has not been realized that the independence property involving such an auxiliary variable reflects the validity of the IV in the ANINCE model. 
\begin{Definition}[\textbf{AIT Condition}]\label{Definition-AIT-condition} Suppose the treatment $X$, the outcome $Y$, and a candidate IV $Z$ are nodes in a causal graph $\mathcal{G}$. Define the auxiliary variable of the causal relationship $X \to Y$ relative to $Z$ as
\begin{align}
\label{Eq-auxiliary-variable}
    \mathcal{A}_{X \to Y||Z} \coloneqq Y - h (X),
\end{align}
where $h(\cdot)$ satisfies $\mathbb{E}[ \mathcal{A}_{X \to Y||Z}|Z] = {0}$ and $h(\cdot) \neq {0}$. We say that $\{X, Y||Z\}$ follows the AIT condition if and only if $\mathcal{A}_{X \to Y||Z}$ is independent from $Z$.
\end{Definition}

The AIT condition states that if there exists a function $h(\cdot)$ such that $\mathbb{E}[Y - h(X) \mid Z] = 0$, then $\{X, Y \| Z\}$ satisfies the AIT condition if and only if the auxiliary variable $\mathcal{A}_{X \to Y \| Z}$ is independent of $Z$. If no non-zero function $h(\cdot)$ exists that satisfies the conditional moment restriction—i.e., $\mathbb{E}[Y - h(X) \mid Z] \neq 0$ for all $h(\cdot)$—then $Y - h(X)$ is always correlated with $Z$ (recall that all variables are assumed to be mean-centered). This implies that $Y - h(X)$ is dependent on $Z$, and hence the AIT condition does not hold. It is noteworthy that under the completeness condition, there exists a unique function $h(\cdot)$ satisfying $\mathbb{E}[Y - h(X) \mid Z] = 0$ in the additive nonparametric IV model (Equation \eqref{Eq-IV-Model}); this function coincides with the true causal function $f(\cdot)$ of $X$ on $Y$. For the sake of conciseness, we often drop the subscript ${X \to Y||Z}$ from $\mathcal{A}_{X \to Y||Z}$ when there is no ambiguity. The following theorem shows the testability of an IV in light of the AIT condition in an ANINCE model.

\begin{Theorem}[\textbf{{Necessary Condition for IV}}]\label{Theorem-Necessary-Condition-IV}
Let $X$, $Y$, and $Z$ be the treatment, outcome, and candidate IV in an ANINCE model, respectively. Suppose that $X$, $Y$, and $Z$ are correlated and that \revise{Assumption \ref{Ass-completeness} holds}. If $Z$ is a valid IV relative to $X \to Y$, then $\{X, Y||Z\}$ always satisfies the AIT condition. 
\end{Theorem} 
\begin{proof} 
\revise{
\tworevise{If $Z$ is a valid IV relative to $X \to Y$, then the ANINCE model can be rephrased as additive nonparametric IV models: 
\begin{equation}
\begin{aligned}
    X = g({Z}) + \varphi_{X}(\mathbf{U}) + \varepsilon_{X},  \quad Y = f(X) + \varphi_{Y}(\mathbf{U}) + \varepsilon_{Y}, 
    \end{aligned}
\end{equation}
where $\mathbf{U} = \boldsymbol{\varepsilon_{U}}$, $Z = \varepsilon_{Z}$, and $\mathbb{E}[\varphi_{Y}(\mathbf{U}) + \varepsilon_{Y}|Z]=0$ due to IV validity. 

Taking conditional expectations of $Y$ with respect to $Z$ yields 
\begin{equation}\label{Eq:Integral equation}
\begin{aligned}
    \mathbb{E}[Y|Z] &= \mathbb{E}[f(X) + {\varphi_{Y}(\mathbf{U}) + \varepsilon_{Y}}|Z] = \mathbb{E}[f(X)|Z] = \int\limits f(X)k(X|Z)dX, 
    \end{aligned}
\end{equation}
where $k(X|Z)$ is the conditional probability distribution function of $X$ given $Z$. Thus, $f(\cdot)$ is a solution to the integral equation \eqref{Eq:Integral equation}.

Now let $h(\cdot)$ be any function that satisfies the conditional moment restriction 
\begin{equation}
\begin{aligned}
E[Y - h(X)\mid Z] = 0
\quad\Longleftrightarrow\quad
E[h(X)\mid Z] = E[Y\mid Z].
    \end{aligned}
\end{equation}

Since $f(\cdot)$ also satisfies $E[Y - f(X)\mid Z]=0$, both $f(\cdot)$ and $h(\cdot)$ solve the same integral equation. Subtracting the two displays yields 
\begin{equation}\label{Eq:same func}
\begin{aligned}
    \mathbb{E}[h(X)-f(X)|Z] = \mathbb{E}[h(X)|Z] - \mathbb{E}[f(X)|Z] = 
    \mathbb{E}[Y|Z] - \mathbb{E}[Y|Z] =0. 
    \end{aligned}
\end{equation}
Define $\psi(X) := h(X) - f(X)$. By the completeness condition (Assumption \ref{Ass-completeness}), the condition $\mathbb{E}[\psi(X)\mid Z]=0$ implies $\psi(X)=0$ almost surely, and thus $h(X)=f(X)$ almost surely. Therefore, the function $h(\cdot)$ that solves $\mathbb{E}[Y-h(X)\mid Z]=0$ is uniquely identified and coincides with the true causal effect function $f(\cdot)$ of $X$ on $Y$. This is the standard completeness-based identification argument in the additive non-parametric IV literature (see, e.g., \citet{newey2003instrumental,newey2013nonparametric}, \citet{singh2019kernel}, and \citet{bennett2019deep}).
}
}

Thus, we have the auxiliary variable 
\begin{align}\label{Eq:Theorem-Necessary-Condition-IV-A}
    \mathcal{A}_{X \to Y||Z} = Y - h(X) = Y - f(X) =  \varphi_{Y}(\mathbf{U}) + \varepsilon_{Y}.
\end{align}
By Theorem 2.2.5 and its extension both in \cite{meester2008natural}, if random variables are mutually independent, then any measurable functions applied to disjoint subsets of them yield independent random variables (see Theorem \ref{The_function_indep} and Corollary~\ref{The_sum_indep} in Appendix~\ref{Appendix-proofs} for further details). Based on this result, we next show that the auxiliary variable $\mathcal{A}_{X \to Y||Z}$ and $Z$ are statistically independent. Specifically, since the noise terms $\varepsilon_Z$, $\varepsilon_Y$, and $\boldsymbol{\varepsilon_U}$ are mutually independent, we can obtain that $\varepsilon_Z$ is also independent of $\varphi_Y(\boldsymbol{\varepsilon_U}) + \varepsilon_Y$. Furthermore, combining the equations $\boldsymbol{U}=\boldsymbol{\varepsilon_U}$ and $Z = \varepsilon_Z$, we conclude that $Z$ is independent of $\varphi_Y(\boldsymbol{U}) + \varepsilon_Y$. Therefore, the auxiliary variable $\mathcal{A}_{X \to Y||Z}$ and $Z$ are statistically independent, \ie, $\mathcal{A}_{X \to Y||Z} \CI Z$. 
This implies that $\{X, Y||Z\}$ satisfies the AIT condition. 

\end{proof}

Theorem \ref{Theorem-Necessary-Condition-IV} means that if $\{X, Y||Z\}$ violates the AIT condition, then $Z$ is an invalid IV relative to $X \to Y$. Otherwise, $Z$ may or may not be valid.

\subsection{Implications of AIT Condition in Additive Linear, Constant Effects Models}\label{Subsection-AIT-Linear-Model}

In this section, we focus our attention on a special type of ANINCE model, the linear, constant effects model, which has been widely studied \citep{bowden2015mendelian,kang2016instrumental, silva2017learning,windmeijer2021confidence}. \revise{Specifically, we assume that the underlying causal structure of the system can be represented by a DAG $\mathcal{G}$, and that the data-generating process follows a linear structural equation model associated with $\mathcal{G}$. In particular, each variable $V_i$ satisfies} $V_i = \sum_{V_j \in pa(V_i)} {\alpha_{ij}V_j + \varepsilon_{V_i}}, i = 1,2,..., r$, where the noise terms $\varepsilon_{V_1},..., \varepsilon_{V_r}$ are independent of each other, and $\alpha_{ij}$ is the direct effect of $V_j \to V_i$. Hence, the ANINCE model in Equation \eqref{Eq-Main-Model} can be expressed as follows:
\begin{equation}\label{Eq-Linear-model}
    \begin{aligned}
        X & = \tau Z  + \boldsymbol{\rho}^{T} \boldsymbol{U} + \varepsilon_{X}, &\quad Y &= \beta X + \nu Z + \boldsymbol{\kappa}^{T} \boldsymbol{U} + \varepsilon_{Y},
    \end{aligned}
\end{equation}
where $Z = \boldsymbol{\gamma}^{T} \boldsymbol{U} + \varepsilon_{Z}$. When $\boldsymbol{\gamma}^{T} = \mathbf{0}$ (satisfying the \emph{exogeneity} condition) and $\nu = 0$ (satisfying the \emph{exclusion restriction} condition), $Z$ qualifies as a valid IV relative to $X \to Y$. Below, we show the implications of the AIT condition in this model. 

\textbf{Motivating Examples:}
Firstly, we illustrate with two simple examples that while a valid IV does not impose any restrictions on the joint marginal distribution of the observed variables within the linear Gaussian model, it does impose certain constraints in the linear partial non-Gaussian model, which can be identified using AIT condition. Consider the causal graph in Figure \ref{Fig:main-example-graph}(b), where $Z$ is an invalid IV for $X \to Y$, as it violates the exogeneity condition. Let $\mathcal{N}(0,1)$ denote the standard normal distribution, and $\text{exp}(0.5)$ denote the exponential distribution with a parameter of $0.5$. Suppose the generating mechanisms of these models are as follows: 
\begin{itemize}
    \item \emph{Linear Gaussian model}. $U = \varepsilon_{U}$, $Z = 2U + \varepsilon_{Z}$, $X =  1.5Z + 0.8U + \varepsilon_{X}$, $Y = X + 3.5U + \varepsilon_{Y}$, and $\varepsilon_{U}, \varepsilon_{Z}, \varepsilon_{X}, \varepsilon_{Y} \sim \mathcal{N}(0,1)$.
    \item \emph{Linear partial non-Gaussian model}. $U = \varepsilon_{U}$, $Z = 2U + \varepsilon_{Z}$, $X = 1.5Z + 0.8U + \varepsilon_{X}$, $Y = X + 3.5U + \varepsilon_{Y}$, $\varepsilon_{U} \sim \text{exp}(0.5)$, and $\varepsilon_{Z}, \varepsilon_{X}, \varepsilon_{Y} \sim \mathcal{N}(0,1)$.
\end{itemize}

\begin{figure}[htbp]
        \centering
    \subfloat[Linear Gaussian Model]{%
        \includegraphics[width=0.45\linewidth, height=4cm]{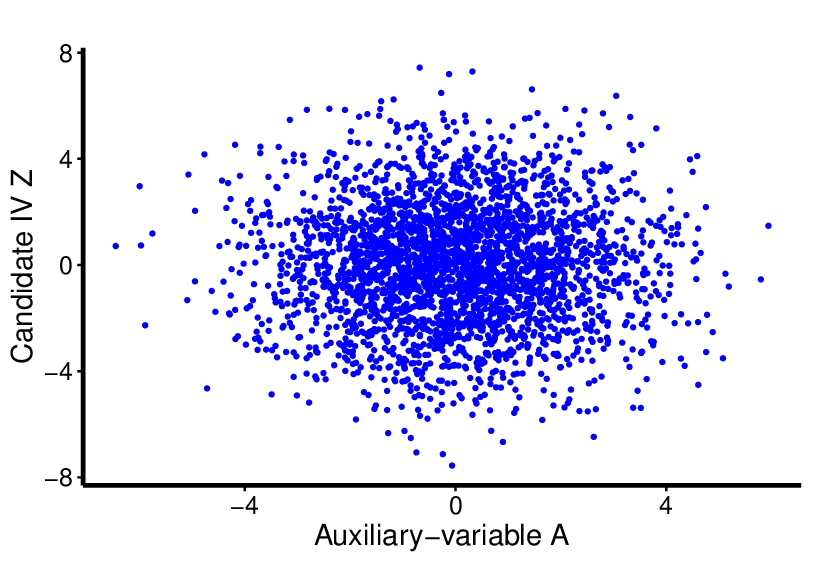}
        \label{fig:Linear_All_Gaussian_Model}
    }
    \hspace{0.05\linewidth}
    \subfloat[Linear Partial Non-Gaussian Model]{%
        \includegraphics[width=0.45\linewidth, height=4cm]{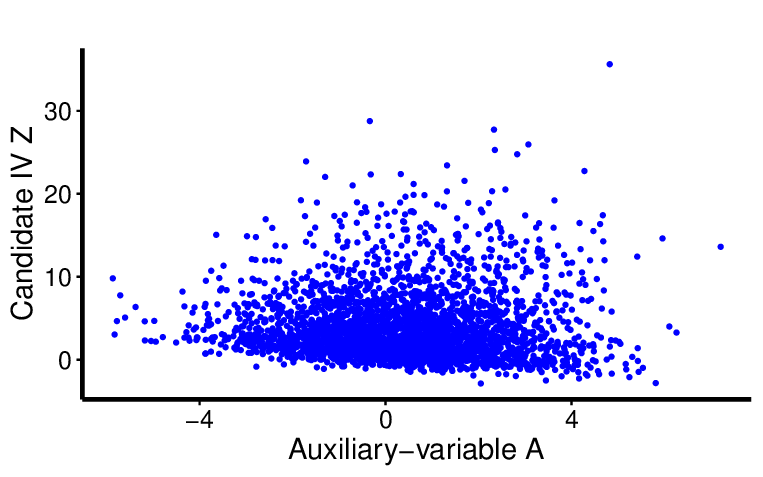}
        \label{fig:Linear_Partial_monGaussian_valid_IV_Model}
    }
    \caption{Scatter plots of Candidate IV $Z$ and Auxiliary-variable $\mathcal{A}$ under the linear models. (a) All noise terms follow Gaussian distributions. (b) Some noise terms follow non-Gaussian distributions.}
    \label{fig:SP_Linear_Model}
\end{figure}

The difference between the above two models lies in the noise term $\varepsilon_{U}$; the first follows a Gaussian distribution, while the second follows an exponential distribution (non-Gaussian distribution). Figure \ref{fig:SP_Linear_Model} shows the scatter plots of $\mathcal{A}_{X \to Y||Z}$ versus the invalid IV $Z$ for two models. Interestingly, we find that, in the linear Gaussian model, $\mathcal{A}_{X \to Y||Z}$ and $Z$ are statistically independent (satisfying AIT condition), while in the linear partial non-Gaussian model, $\mathcal{A}_{X \to Y||Z}$ and $Z$ are statistically dependent (violating AIT condition). These results suggest that non-Gaussianity is beneficial for identifying the invalid IV.

The following Propositions \ref{proposition-linear-gaussian-model} and \ref{Proposition-linear-exogeneity} formalize the phenomena discussed above. 
\begin{Proposition}[Non-Testability in Linear Gaussian Models]\label{proposition-linear-gaussian-model}
Let $X$, $Y$, and $Z$ be the treatment, outcome, and candidate IV in a linear model (Equation \eqref{Eq-Linear-model}), respectively. Suppose that $X$, $Y$, and $Z$ are correlated. 
If all noise terms of variables follow Gaussian distributions, then regardless of whether $Z$ is a valid IV relative to $X \to Y$ or not, $\{X, Y||Z\}$ always satisfies the AIT condition.
\end{Proposition}
\begin{proof}\label{Proof-proposition-linear-gaussian-model}
The proof of Proposition \ref{proposition-linear-gaussian-model} is straightforward. Let \( Z \) represent any candidate IV, which may or may not be valid. \revise{In the linear Gaussian model, by the definition of the AIT condition, there exists a function $h(X) = \frac{\operatorname{Cov}(Y, Z)}{\operatorname{Cov}(X, Z)} \cdot X$ such that $\mathbb{E}[Y - h(X) \mid Z] = 0$. Hence, $\mathbb{E}[ \mathcal{A}_{X \to Y||Z} | Z]= \mathbb{E}[Y - \frac{\operatorname{Cov}(Y, Z)}{\operatorname{Cov}(X, Z)} X \mid Z] = 0$, which implies that $Cov( \mathcal{A}_{X \to Y||Z}, Z)=0$. Since zero correlation implies independence in the linear Gaussian model~\citep{bain1992introduction}, we conclude that \(\mathcal{A}_{X \to Y||Z}\) is independent of the candidate IV \( Z \). Consequently, \(\{X, Y || Z\}\) always satisfies the AIT condition.} 
\end{proof}

Proposition \ref{proposition-linear-gaussian-model} states that checking the AIT condition in a linear Gaussian causal model (second-order statistics) does not provide any useful information for identifying invalid IVs. Below, we show that one can leverage higher-order statistics \footnote{Higher-order statistics mean beyond the second-order moments in statistics of the data, such as skewness and kurtosis.} of noise terms to identify certain types of invalid IVs that violate \emph{exogeneity} condition. Before presenting the result, we give the key assumption.

\begin{Assumption}[Partial Non-Gaussianity]\label{Ass-non-Gaussianity-exogeneity}
At least one of the following conditions holds: 
(i) there exists at least one variable $U_i \in \mathbf{U}$ whose noise term follows a non-Gaussian distribution and cause $Z$; (ii) the noise term of $Z$ follows a non-Gaussian distribution.
\end{Assumption}

Assumption \ref{Ass-non-Gaussianity-exogeneity} states the non-Gaussianity of data, which is expected to be widespread, as suggested by Cram\'{e}r Decomposition Theorem \citep{Cramer62}.  Considerable works have already been built on this assumption~\citep{shimizu2006linear,salehkaleybar2020learning}. For additional references, see \citet{spirtes2016causal,shimizu2022statistical}.

We now show that the AIT condition can access the validity of \emph{exogeneity} condition in linear models under Assumption \ref{Ass-non-Gaussianity-exogeneity}. 
\begin{Proposition}[Testability of Exogeneity in Linear Models]\label{Proposition-linear-exogeneity}
Let $X$, $Y$, and $Z$ be the treatment, outcome, and candidate IV in a linear model (Equation \eqref{Eq-Linear-model}), respectively. Suppose that $X$, $Y$, and $Z$ are correlated, and that Assumption \ref{Ass-non-Gaussianity-exogeneity} holds. If $Z$ violates the exogeneity condition, i.e., at least one variable $U_i \in \mathbf{U}$ causes $Z$, then $\{X, Y||Z\}$ violates the AIT condition.
\end{Proposition}
\begin{proof}
    Roughly speaking, by the Darmois–Skitovich theorem \citep{darmois1953analyse,skitovitch1953property},  if $\mathcal{A}_{X \to Y||Z}$ shares any common non-Gaussian noise terms $\varepsilon_{U_i}$ or $\varepsilon_Z$ with $Z$, $\mathcal{A}_{X \to Y||Z}$ is statistically dependent on $Z$. This implies that $\{X, Y||Z\}$ violates the AIT condition. See Appendix \ref{Proof-proposition-exogeneity-linear-model} for its complete proof.
\end{proof}

\begin{Remark}
When all noise terms follow the non-Gaussian distributions, the linear partial non-Gaussian model becomes the well-known Linear Non-Gaussian Acyclic Model (LiNGAM), which has been extensively studied \citep{shimizu2006linear, salehkaleybar2020learning}. Consequently, according to Proposition \ref{Proposition-linear-exogeneity}, an invalid IV that violates the exogeneity condition in LiNGAM can be detected in light of the AIT condition.
\end{Remark}

Proposition \ref{Proposition-linear-exogeneity} states that we can detect an invalid IV that violates \emph{exogeneity} condition using the AIT condition based on the observational data in the linear model under Assumption \ref{Ass-non-Gaussianity-exogeneity}. A natural question that arises is whether we can detect IVs that violate the \emph{exclusion restriction} condition within the same framework. Unfortunately, in practice, we cannot detect an invalid IV that solely violates the \emph{exclusion restriction} condition in the linear model, as shown in the following example. 

\begin{Example-set}
Let's consider the causal structure illustrated in Figure \ref{Fig:main-example-graph} (a), where $Z$ is a valid IV. We assume the following relationship: $Z = \varepsilon_Z$, $X = \tau Z + \boldsymbol{\rho}^{T} \boldsymbol{U} + \varepsilon_{X}$, $Y = \beta X + \boldsymbol{\kappa}^{T} \boldsymbol{U} + \varepsilon_{Y}$. Next, we demonstrate how to construct another causal structure as shown in Figure \ref{Fig:main-example-graph} (c), where $Z$ becomes an invalid IV (violating solely the exclusion restriction condition). Specifically, let $\beta^{\prime} = \beta - \frac{\nu}{\tau}$, $Z^{\prime} = Z$, $X^{\prime} = X$, and $Y^{\prime} = \beta^{\prime} X^{\prime} + \nu Z^{\prime} + \boldsymbol{\kappa}^{T} \boldsymbol{U} + \varepsilon_{Y} + \frac{\nu}{\tau}( \boldsymbol{\rho}^{T} \boldsymbol{U} + \varepsilon_X) = Y$. Thus, $(Z, X, Y)$ has the same distribution as $(Z^{\prime}, X^{\prime}, Y^{\prime})$. This implies that a variable being an instrument imposes no constraints on the joint marginal distribution of the observed variables. The same result is also discussed in Section 3 of \citet{chu2001semi}.    
\end{Example-set}

The following proposition states the above phenomenon in the linear model. 
\begin{Proposition}[Non-Testability of Exclusion Restriction in Linear Models]\label{proposition-linear-model-exclusion}
Let $\quad$ $X$, $Y$, and $Z$ be the treatment, outcome, and candidate IV in a linear model (Equation \eqref{Eq-Linear-model}), respectively. Suppose that $X$, $Y$, and $Z$ are correlated. If $Z$ satisfies the exogeneity condition, regardless of whether $Z$ violates the exclusion restriction condition or not, then $\{X, Y||Z\}$ always satisfies the AIT condition. 
\end{Proposition}

\begin{proof}
    This proof is straightforward. The auxiliary variable $\mathcal{A}_{X \to Y||Z}$ shares no common noise terms with candidate IV $Z$, whether the noise terms are Gaussian or non-Gaussian. By the Darmois-Skitovich theorem \citep{darmois1953analyse,skitovitch1953property}, $\mathcal{A}_{X \to Y||Z}$ is statistically independent from $Z$. This implies that $\{X, Y||Z\}$ always satisfies the AIT condition. See Appendix \ref{Proof-proposition-exclusion-linear-model} for its complete proof.
\end{proof}

Based on Theorem \ref{Theorem-Necessary-Condition-IV} and Propositions \ref{proposition-linear-gaussian-model} $\sim$ \ref{proposition-linear-model-exclusion}, we derive the following theorem, which provides a necessary and sufficient condition for detecting invalid IVs \revise{that violate the \emph{exogeneity} condition} within a linear causal model. 

\begin{Theorem}[\textbf{{Necessary and Sufficient Conditions in Linear Models}}]\label{Theorem-Necessary-Sufficient-Condition-IV-Linear-Models}
Let $X$, $Y$, and $Z$ be the treatment, outcome, and candidate IV in a linear model (Equation \eqref{Eq-Linear-model}), respectively. Suppose that $X$, $Y$, and $Z$ are correlated and that Assumption \ref{Ass-non-Gaussianity-exogeneity} holds. $\{X, Y||Z\}$ violates the AIT condition if and only if the candidate IV $Z$ is invalid due to a violation of the exogeneity condition. 
\end{Theorem}

\begin{proof}
    See Appendix \ref{proof:Theorem-Necessary-Sufficient-Condition-IV-Linear-Models} for its proof.
\end{proof}

This theorem states that the AIT condition is necessary and sufficient to detect the candidate IV violations of the \emph{exogeneity} condition when Assumption \ref{Ass-non-Gaussianity-exogeneity} holds.

\subsection{Implications of AIT Condition in Additive Nonlinear,
Non-Constant Effects Models}\label{Subsection-AIT-Non-Linear-Model}
In this section, we investigate the implications of the AIT condition on the ANINCE model. 
Before giving our main results, we first show a simple example to show that nonlinearity is beneficial in identifying the invalid IV. 

\textbf{A Motivating Example:} Continue to consider the causal graph in Figure \ref{Fig:main-example-graph} (b), where $Z$ serves as an invalid IV for the causal relationship $X \to Y$, violating the exogeneity condition. Here, we modify the generation mechanism of the linear Gaussian model by introducing a nonlinear function between $U$ and $Z$, specifically as follows:
\begin{itemize}
    \item \emph{Partial Non-linear Gaussian model}. $U  = \varepsilon_{U}$, $Z = \exp(U) + \varepsilon_{Z}$, $X = 1.5 Z + 0.8 U +\varepsilon_{X}$, $Y = X + 3.5 U + \varepsilon_{Y}$, and $\varepsilon_{U}, \varepsilon_{Z}, \varepsilon_{X}, \varepsilon_{Y} \sim \mathcal{N}(0,1)$.
\end{itemize}
\begin{figure}[htp]
\begin{center}
    \includegraphics[width=0.5\linewidth, height=4cm]{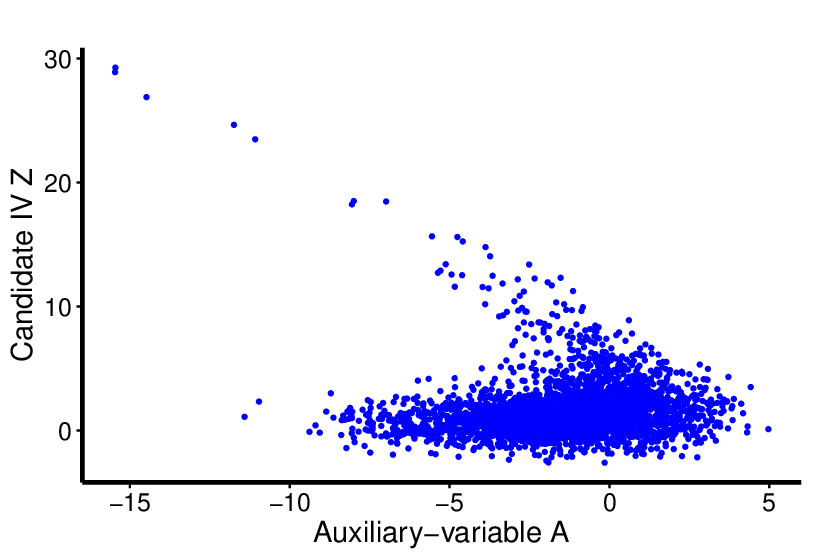}
    \caption{
    Scatter plot of Candidate IV $Z$ and Auxiliary-variable $\mathcal{A}$ when all noise terms follow Gaussian distribution in the partially non-linear invalid IV model. 
    }
    \label{fig:SP_Violate_A2_assumption-nonlinear}
    \end{center}
\end{figure}

Figure \ref{fig:SP_Violate_A2_assumption-nonlinear} presents the scatter plots of $\mathcal{A}_{X \to Y||Z}$ versus the candidate IV $Z$ in the partial non-linear Gaussian model. Compared to the linear Gaussian model, we transformed the functional relationship between $U$ and $Z$ from a linear function ($2U$) to an exponential function ($\exp(U)$). Interestingly, in the partial non-linear Gaussian model, $\mathcal{A}_{X \to Y||Z}$ and $Z$ are statistically dependent (violating AIT condition). 
Note that Proposition \ref{proposition-linear-gaussian-model} shows that the AIT condition is always satisfied in the linear Gaussian model. These findings suggest that nonlinearity is beneficial in assessing the validity of \emph{exogeneity} condition. 

We now investigate the conditions under which the invalid IV can be detected in terms of AIT condition. It is noteworthy that the ANINCE model (Equation \eqref{Eq-Main-Model}) is flexible as functions $g$, $f$, $\varphi_{X}$, and $\varphi_{Y}$ might be any unknown functions. Consequently, without imposing further parametric assumptions, it is impossible to determine the explicit forms of the estimated $f$ and $\mathcal{A}_{X \to Y||Z}$. Hence, let $h(\cdot)$ be a function satisfying $\mathbb{E}[ Y - h(X)|Z] = {0}$ and $h(\cdot) \neq {0}$. According to the definition of the AIT condition, 
the auxiliary variable $\mathcal{A}_{X \to Y||Z}$ is given by:
\begin{equation}
    \begin{aligned}
        \mathcal{A}_{X \to Y||Z} = Y- h(X)= \underbrace{f(X, Z) - h(X)}_{\widetilde{f}_{bias}(X, Z)} + \varphi_{Y}(\mathbf{U}) + \varepsilon_{Y},
    \end{aligned}
\end{equation}
where $\widetilde{f}_{bias}(X, Z) = f(X,Z) - h(X)$. It is important to note that for a valid instrumental variable $Z$, $\widetilde{f}_{bias}(X, Z) =0$. 

Below, we give the key assumption regarding the second-order partial derivative in the nonlinear model.

\begin{Assumption}[Distributional Non-degeneracy Condition]\label{Ass-algebraic-condition}
Assume that the joint density $p(\mathcal A_{X\to Y\|Z},Z)$ is twice continuously differentiable and satisfies
\begin{equation}\label{Eq:second-order-partial-derivative}
\begin{aligned}
\frac{\partial^2}{\partial \mathcal A_{X\to Y\|Z}\,\partial Z}
\log p(\mathcal A_{X\to Y\|Z},Z)\neq 0
\end{aligned}
\end{equation}
on a set with non-zero Lebesgue measure.
\end{Assumption}

Assumption \ref{Ass-algebraic-condition} is a natural condition that one expects to hold for detecting all invalid IVs in the ANINCE model. Intuitively, if the auxiliary variable $\mathcal{A}_{X \to Y \| Z}$ is independent of the instrument $Z$ (i.e., AIT holds), then the joint probability density of $(\mathcal{A}_{X \to Y \| Z}, Z)$ factorizes into the product of the marginal densities: $p(\mathcal{A}_{X \to Y \| Z}, Z) = p(\mathcal{A}_{X \to Y \| Z}) \cdot p(Z)$. According to \cite{lin1997factorizing}, for a set of independent random variables whose joint density is twice differentiable, the Hessian matrix of the logarithm of their joint density is diagonal everywhere (see Theorem \ref{Theorem-lin} in Appendix \ref{Appendix-proofs} for further details). Taking the logarithm of $p(\mathcal{A}_{X \to Y \| Z}, Z)$ and computing its second-order mixed partial derivative yields: $\frac{\partial^2 \log p(\mathcal{A}_{X \to Y \| Z}, Z)}{\partial \mathcal{A}_{X \to Y \| Z} \, \partial Z} = 0$. In contrast, if AIT does not hold—\ie, $\mathcal{A}_{X \to Y \| Z}$ and $Z$ are dependent—then $p(\mathcal{A}_{X \to Y \| Z}, Z) \neq p(\mathcal{A}_{X \to Y \| Z}) \cdot p(Z)$, which implies that $\frac{\partial^2 \log p(\mathcal{A}_{X \to Y \| Z}, Z)}{\partial \mathcal{A}_{X \to Y \| Z} \, \partial Z} \neq 0$ (Equation \eqref{Eq:second-order-partial-derivative}). 

\tworevise{
Below, we provide two nontrivial analytic examples that satisfy the distributional non-degeneracy condition. 
\begin{Example-set}\label{Example:violation of exogeneity}{\normalfont(\textbf{Violation of the Exogeneity Condition})} 
Consider the causal graph in Figure~\ref{Fig:main-example-graph} (b), where $Z$ violates exogeneity. Let the structural model be $U=\varepsilon_U$, $Z=\gamma U+\varepsilon_Z$, $X=\exp(Z)+\rho U+\varepsilon_X$, and $Y=\beta X+\kappa U+\varepsilon_Y$, where $\gamma\neq 0$ and $(\varepsilon_U,\varepsilon_Z,\varepsilon_X,\varepsilon_Y)\stackrel{\mathrm{ind}}{\sim}\mathcal N(0,1)$. According to the definition of auxiliary variables, we obtain $\mathcal{A}_{X \to Y||Z}=(\beta-\hat\beta)\exp(Z)+\big((\beta-\hat\beta)\rho+\kappa\big)U+(\beta-\hat\beta)\varepsilon_X+\varepsilon_Y$, where $\hat\beta
=\frac{\Cov(Y,Z)}{\Cov(X,Z)}=\beta+\frac{\kappa\gamma}{\sigma_Z^2 \exp({\sigma_Z^2/2})+\rho\gamma}$ and $\sigma_Z^2=\gamma^2+1$. Conditioned on $Z=z$, the auxiliary variable $\mathcal{A}_{X \to Y||Z}$ is a linear combination of Gaussian variables $\varepsilon_U, \varepsilon_X$, and $\varepsilon_Y$, implying $\mathcal{A} \mid Z=z \sim \mathcal{N}(m(z), v)$, where $m(z)=(\beta-\hat\beta)\exp({z})+\big((\beta-\hat\beta)\rho+\kappa\big)\frac{\gamma}{\sigma_Z^2}z$ and $v=\big((\beta-\hat\beta)\rho+\kappa\big)^2\frac{1}{\sigma_Z^2}+(\beta-\hat\beta)^2+1$. Using the factorization $p(\mathcal A,Z)=p(\mathcal A\mid Z)p(Z)$, we take the derivative of the joint density of $(\mathcal{A},Z)$ and obtain: 
\begin{equation}\nonumber
\frac{\partial^2 \log p(\mathcal A,Z)}{\partial \mathcal A\,\partial Z}=\frac{\partial^2 \log p(\mathcal A\mid Z)}{\partial \mathcal A\,\partial Z}= \frac{m'(Z)}{v} =\frac{\kappa\gamma-\beta_{bias}[\rho\gamma+\exp(Z)(\gamma^2+1)]}{(\kappa-\beta_{bias}\rho)^2 + (\gamma^2+1)(\beta_{bias}^2+1)}, 
\end{equation}
where $\beta_{bias}=\frac{\kappa\gamma}{\sigma_Z^2 \exp({\sigma_Z^2/2})+\rho\gamma}$. Since the numerator $\kappa\gamma-\beta_{bias}[\rho\gamma+\exp(Z)(\gamma^2+1)]$ is not identically zero, the mixed derivative is nonzero on a set with non-zero Lebesgue measure, and therefore Assumption~\ref{Ass-algebraic-condition} holds. A detailed derivation is provided in Appendix~\ref{App:Ex1-exogeneity}. 
\end{Example-set}

\begin{Example-set}\label{Example:violation of exclusion}{\normalfont(\textbf{Violation of the Exclusion Restriction Condition})} 
Consider the causal graph in Figure~\ref{Fig:main-example-graph} (c), where $Z$ violates the exclusion restriction. Let the structural model be $U=\varepsilon_U$, $Z=\varepsilon_Z$, $X=\exp(Z)+\rho U+\varepsilon_X$, and $Y=\beta X+\nu Z+\kappa U+\varepsilon_Y$, where $\nu\neq 0$ and $(\varepsilon_U,\varepsilon_Z,\varepsilon_X,\varepsilon_Y)\stackrel{\mathrm{ind}}{\sim}\mathcal N(0,1)$. According to the definition of auxiliary variables, we obtain $\mathcal{A}_{X \to Y||Z}=(\beta-\hat\beta)\exp(Z)+\nu Z+\big((\beta-\hat\beta)\rho+\kappa\big)U+(\beta-\hat\beta)\varepsilon_X+\varepsilon_Y$, where $\hat\beta
=\frac{\Cov(Y,Z)}{\Cov(X,Z)}=\beta+\nu \exp({-1/2})$. Conditioned on $Z=z$, the auxiliary variable $\mathcal{A}_{X\to Y||Z}$ is a linear combination of Gaussian variables $U, \varepsilon_X$, and $\varepsilon_Y$, implying $\mathcal{A} \mid Z=z \sim \mathcal{N}(m(z), v)$, where $m(z)=(\beta-\hat\beta)\exp({z})+\nu z$, and 
$v=\big((\beta-\hat\beta)\rho+\kappa\big)^2+(\beta-\hat\beta)^2+1$. Using the factorization $p(\mathcal A,Z)=p(\mathcal A\mid Z)p(Z)$, we take the derivative of the joint density of $(\mathcal{A},Z)$ and obtain: 
\begin{equation}\nonumber
\frac{\partial^2 \log p(\mathcal A,Z)}{\partial \mathcal A\,\partial Z}=\frac{\partial^2 \log p(\mathcal A\mid Z)}{\partial \mathcal A\,\partial Z} = \frac{m'(Z)}{v} =\frac{\nu [1-\exp(Z-1/2)]}{[\kappa-\nu \exp{(\frac{-1}{2})}]^2 + \nu^2\exp{(-1)+1}}.
\end{equation}
Since the numerator $\nu [1-\exp(Z-1/2)]$ is not identically zero, the mixed derivative is nonzero on a set with non-zero Lebesgue measure, and therefore Assumption~\ref{Ass-algebraic-condition} holds. A detailed derivation is provided in Appendix~\ref{App:Ex2-exclusion}. 
\end{Example-set}
}

\begin{Proposition}[Testability of IV in ANINCE Models]\label{Proposition-identifiable-IV-ANINCE}
Let $X$, $Y$, and $Z$ be the treatment, outcome, and candidate IV in an ANINCE model, respectively. Suppose that $X$, $Y$, and $Z$ are correlated and that \revise{Assumptions \ref{Ass-completeness} and \ref{Ass-algebraic-condition} hold}. If the candidate IV $Z$ is invalid, then $\{X, Y||Z\}$ violates the AIT condition. 
\end{Proposition} 
\begin{proof}
See Appendix \ref{Proof-Proposition-identifiable-IV-ANINCE} for its proof.
\end{proof}
The proposition above shows that, \revise{under Assumptions \ref{Ass-completeness} and \ref{Ass-algebraic-condition}}, the AIT condition can be used to detect invalid IVs. Although it is not obvious whether Assumption \ref{Ass-algebraic-condition} (\tworevise{the dstributional non-degeneracy condition}) holds in general, some solutions under which the Assumption \ref{Ass-algebraic-condition} does not hold are worth reporting, \revise{as shown in the following proposition.}

\revise{
\begin{Proposition}
\label{Pro-violate-algebraic}
    Let $X$, $Y$, and $Z$ be the treatment, outcome, and candidate IV in an ANINCE model, respectively. Assumption \ref{Ass-algebraic-condition} does not hold if any of the following conditions is satisfied: 
    \begin{enumerate}
        \item[(i)] the candidate IV $Z$ solely violates the exclusion restriction condition and the direct causal effect of $Z \to Y$ is a linear function of the direct causal effect from $Z \to X$ in the ANINCE model, i.e., $g_Y(Z) = a \cdot g_X(Z) + b$, where $a$ is the non-zero constant; 
        \item[(ii)] all noise terms associated with the variables follow Gaussian distributions within the linear causal model. 
    \end{enumerate}
\end{Proposition}
\begin{proof}
    See Appendix \ref{proof-Pro-violate-algebraic} for its proof.
\end{proof}
}

\revise{Condition $(i)$ implies that the AIT condition is always satisfied when the candidate IV $Z$ satisfies the \emph{exogeneity} condition and the relationship of effect $g_Y(Z) = a\cdot g_X(Z) + b$ holds in the ANINCE model. In other words, if the direct causal effect of $Z \to Y$ is not a linear function of the direct causal effect of $Z \to X$, we can identify invalid IVs that solely violate the exclusion restriction condition using the AIT condition. Intuitively, the distribution of the invalid IV model under condition $(i)$, in which candidate IV $Z$ solely violates the exclusion restriction condition and satisfies the specific linear relationship, can be transformed into the distribution of the valid IV model, as illustrated in Example \ref{Example:valid-invalid-transfromed-ANINCE}. Note that the constraint $g_Y(Z) = a \cdot g_X(Z) + b$ naturally arises in linear causal models. Thus, condition $(i)$ further implies that when the candidate IV $Z$ solely violates the exclusion restriction condition in the linear causal model, Assumption \ref{Ass-algebraic-condition} does not hold, which is consistent with the result of Proposition \ref{proposition-linear-model-exclusion}. In addition, condition $(ii)$ shows that the AIT condition is always satisfied in the linear Gaussian model, which aligns with the conclusion of Proposition \ref{proposition-linear-gaussian-model}. }

\begin{Example-set}\label{Example:valid-invalid-transfromed-ANINCE}
Continue to consider the causal graph in Figure \ref{Fig:main-example-graph} (c), where $Z$ is an invalid IV relative to $X \to Y$. The generating mechanism is as follows:
\begin{equation}\label{Eq-main-model-exclusion}
\begin{aligned}
    Z = \varepsilon_{Z}, \quad 
    X = g_{X}(Z) + \varphi_{X}(\mathbf{U}) + \varepsilon_{X},  \quad Y = f(X) + g_Y(Z) + \varphi_{Y}(\mathbf{U}) + \varepsilon_{Y},
\end{aligned}
\end{equation}
where $g_Y(Z) = a \cdot g_X(Z) + b $.  
We now construct another model based on the causal graph shown in Figure \ref{Fig:main-example-graph} (a), where $Z$ is a valid IV relative to $X \to Y$. Let $Z^{\prime} = Z$, $X^{\prime} = X$, and $f^{\prime}(X^{\prime}) = f(X^{\prime}) + a \cdot 
g_X(Z^{\prime}) + b$. Furthermore, $Y^{\prime}$ is expressed as $Y^{\prime} = f^{\prime}(X^{\prime}) + \varphi_{Y}(\mathbf{U}) + \varepsilon_{Y} = f(X) + g_{Y}(Z) + \varphi_{Y}(\mathbf{U}) + \varepsilon_{Y} = Y$. Hence, we conclude that the distribution of $(Z, X, Y)$ and $(Z^{\prime}, X^{\prime}, Y^{\prime})$ are identical. Because we can solely observe the variables $(Z, X, Y)$, it is impossible to determine from the distribution of $(Z, X, Y)$ whether the data come from Figure \ref{Fig:main-example-graph} (a) or Figure \ref{Fig:main-example-graph} (c). In other words, we cannot ascertain whether $Z$ is a valid IV or not. 
\end{Example-set}

According to Proposition \ref{Pro-violate-algebraic}, not all invalid IVs can be identified solely from the joint distribution of observational data. Below, based on Theorem \ref{Theorem-Necessary-Condition-IV} and Proposition \ref{Proposition-identifiable-IV-ANINCE}, we introduce the necessary and sufficient conditions for invalid IV in the additive nonlinear, non-constant effects model.

\begin{Theorem}[\textbf{{Necessary and Sufficient Conditions for IV in ANINCE Models}}]\label{Theorem-Necessary-Sufficient-Condition-IV-ANINCE-Models}
Let $X$, $Y$, and $Z$ be the treatment, outcome, and candidate IV in an ANINCE model (Equation \eqref{Eq-Main-Model}), respectively. Suppose that $X$, $Y$, and $Z$ are correlated, and that \revise{Assumption \ref{Ass-completeness} holds}. 
\begin{itemize}
    \item \revise{If $Z$ is a valid IV relative to $X \to Y$, then $\{X, Y||Z\}$ always satisfies the AIT condition.} 
    \item \revise{If $Z$ is an invalid IV relative to $X \to Y$ and Assumption \ref{Ass-algebraic-condition} holds, then $\{X, Y||Z\}$ always violates the AIT condition.} 
\end{itemize}
\end{Theorem}
\begin{proof}
    See Appendix \ref{Proof-Theorem-Necessary-Sufficient-Condition-IV-ANINCE-Models} for its proof.
\end{proof}
Theorem \ref{Theorem-Necessary-Sufficient-Condition-IV-ANINCE-Models} outlines two scenarios where a candidate IV would be considered invalid: either the IV doesn't meet the \emph{exogeneity} condition or it violates the \emph{exclusion restriction} condition within a nonlinear model. \revise{
Figure \ref{fig:flowchart} presents how the AIT condition relates to the validity of IVs under different assumptions. The completeness condition (Assumption \ref{Ass-completeness}) alone guarantees necessity, while adding a \tworevise{distributional non-degeneracy condition} (Assumption \ref{Ass-algebraic-condition}) ensures sufficiency.}

\begin{figure}[h]
  \centering
\includegraphics[trim=0cm 3.7cm 2.0cm 0cm, clip, width=1.0\textwidth, page=1]{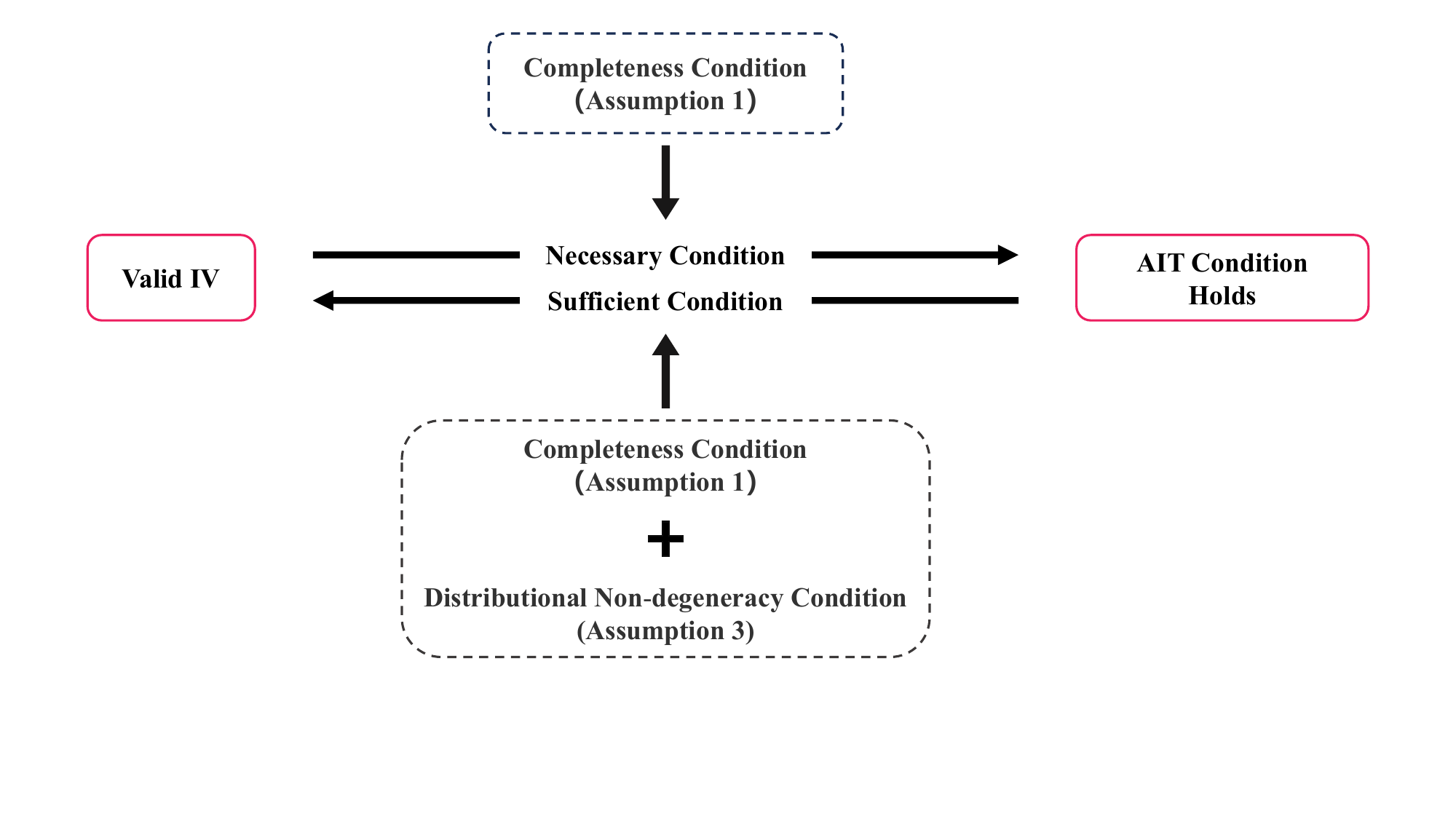}
  \caption{Flowchart showing how assumptions lead to the AIT condition being necessary or sufficient for IV validity.}
  \label{fig:flowchart}
\end{figure}

\section{Practical Implementation of AIT Condition}\label{Section-Test-Procedure}
In this section, we discuss the practical implementation of the AIT condition. In Section \ref{Section-Expanded-AIT-Condition}, we first address how to implement the AIT condition when covariates are present. \tworevise{In Section \ref{Section-algorithm}, we then propose a sample splitting procedure for implementing the AIT condition with finite samples. Finally, in Section \ref{Sec: Asymptotic Validity}, we establish the asymptotic validity of the AIT test.}

\subsection{AIT Condition with Covariates}\label{Section-Expanded-AIT-Condition}
In practice, there are scenarios where covariates $\mathbf{W}$ are present. For instance, age may influence how the treatment method affects patient recovery speed. 
\revise{Specifically, we focus on the generation of the ANINCE model with covariates as follows: 
\begin{equation}\label{Eq:ANINCE-covariates}
\begin{aligned}
    X &= g(\mathbf{W}, Z) + \varphi_{X}(\mathbf{U}) + \varepsilon_{X}, \\
    Y &= f(X, \mathbf{W}, Z) + \varphi_Y(\mathbf{U}) + \varepsilon_{Y},
    \end{aligned}
\end{equation}
where $g(\cdot):\mathbb{R}^{|\mathbf{W}|+1}\to\mathbb{R}$,
$f(\cdot):\mathbb{R}^{|\mathbf{W}|+2}\to\mathbb{R}$, and $\varphi_{*}(\cdot):\mathbb{R}^{|\mathbf{U}|}\to\mathbb{R}$
are smooth functions. The noise terms $\varepsilon_X$, $\varepsilon_Y$, $\varepsilon_Z$, and $\boldsymbol{\varepsilon_U}$ are statistically independent.} 
Below, we first extend the AIT condition from Definition \ref{Definition-AIT-condition} to account for the presence of covariates, as stated in the following definition.\begin{Definition}[\textbf{AIT Condition with Covariates}]\label{Definition-AIT-condition-covariates} Suppose the treatment $X$, the outcome $Y$, the covariates $\mathbf{W}$, and a candidate IV $Z$ are nodes in a causal graph $\mathcal{G}$. Define the auxiliary variable of the causal relationship $X \to Y$ relative to ${(Z,\mathbf{W})}$ as 
\begin{align}
\label{Eq-auxiliary-variable-covariates}
    \mathcal{A}_{X \to Y||(Z,\mathbf{W})} \coloneqq Y - h (X,\mathbf{W}),
\end{align}
where $h(\cdot)$ satisfies $\mathbb{E}[ \mathcal{A}_{X \to Y||(Z, \mathbf{W})}|Z, \mathbf{W}] = {0}$ and $h(\cdot) \neq {0}$. Define the residual of $Z$ after regressing on $\mathbf{W}$ as:
\begin{align}
    \mathcal{Z} \coloneqq Z - \mathbb{E}[Z|\mathbf{W}].
\end{align}
We say that $\{X, Y||(Z, \mathbf{W})\}$ follows the AIT condition if and only if $\mathcal{A}_{X \to Y||(Z,\mathbf{W})}$ is independent from $\mathcal{Z}$.
\end{Definition}

\tworevise{

\begin{Remark}
It is worth noting that covariates $\mathbf{W}$  can provide additional constraints that can be exploited in the IV validity test. Specifically, under the ANINCE model, incorporating covariates generalizes the AIT condition to the independence $\mathcal{A}_{X\to Y\mid\mid(Z,\mathbf{W})} \,\CI\, (Z,\mathbf{W})$, meaning that the auxiliary variable must be independent of the joint vector $(Z,\mathbf{W})$. However, directly testing independence against a high-dimensional vector is statistically challenging: nonparametric tests typically suffer from low power in high dimensions due to the curse of dimensionality and concentration of measure \citep{zhang2011kernel,ramdas2015decreasing}. To mitigate this issue, we use a statistically more tractable surrogate condition: 
$
\mathcal{A}_{X\to Y\mid\mid(Z,\mathbf{W})} \,\CI\, \mathcal{Z},
\text{where } \mathcal{Z} := Z - \mathbb{E}[Z\mid \mathbf{W}],
$
which tests independence between the auxiliary variable and the residual part of $Z$ that is 
conditionally orthogonal to $\mathbf{W}$. This preserves the intended inferential target under the ANINCE model while improving statistical power.
\end{Remark}

}

Based on Definition \ref{Definition-AIT-condition-covariates} and Theorem \ref{Theorem-Necessary-Condition-IV}, we obtain the necessary condition for IV in the presence of covariates $\mathbf{W}$, as described in the following corollary.
\begin{Corollary}[\textbf{{Necessary Condition for IV with Covariates}}]\label{Corollary-Necessary-Condition-IV-covariates}
Let $X$, $Y$, $\mathbf{W}$, and $Z$ be the treatment, outcome, covariates, and candidate IV in an ANINCE model, respectively. Suppose that $X$, $Y$, $\mathbf{W}$, and $Z$ are correlated and that \revise{Assumption \ref{Ass-completeness} holds}. If $Z$ is a valid IV relative to $X \to Y$ given $\mathbf{W}$, then $\{X, Y||(Z, \mathbf{W})\}$ always satisfies the AIT condition.
\end{Corollary}
\begin{proof}
    See Appendix \ref{Proof:Corollary-Necessary-Condition-IV-covariates} for its proof. 
\end{proof}
Corollary \ref{Corollary-Necessary-Condition-IV-covariates} means that if $\{X, Y||(Z, \mathbf{W})\}$ violates the AIT condition, then $Z$ is an invalid IV relative to $X \to Y$ given $\mathbf{W}$. Otherwise, $Z$ may or may not be valid.

\tworevise{
\subsection{AIT Condition with Finite Data}\label{Section-algorithm}

Below, we provide the practical implementation of the AIT condition with finite data. For a candidate IV $Z$, under the ANINCE model, we need to test the following hypothesis:
\begin{equation}
    \begin{aligned}
        & H_0 : Z \text{ is a valid IV}, \quad H_1 : Z \text{ is an invalid IV}.
    \end{aligned}
\end{equation}

Since the AIT test is implemented as a plug-in procedure involving estimated functions, including the function $\hat{h}(X,\mathbf{W})$ and, when covariates are present, the regression function $\hat{\pi}(\mathbf{W})$ of $\mathbb{E}[Z \mid \mathbf{W}]$, we employ sample splitting to separate estimation and testing \citep{doran2014permutation, chernozhukov2018double, polo2023conditional, ren2025regression,scheidegger2025residual}. 
Specifically, we randomly split the observed dataset $\mathcal{D}$ with $|\mathcal{D}| = N$ into two disjoint subsets, $\mathcal{D} = \mathcal{D}_1 \cup \mathcal{D}_2$, where $|\mathcal{D}_1| = n = \lfloor \rho N \rfloor$, $\rho \in (0,1)$ is the splitting ratio, and $|\mathcal{D}_2| = m = N - n$. The functions $\hat{h}(\cdot)$ and $\hat{\pi}(\cdot)$ (if covariates are present) are estimated using dataset $\mathcal{D}_1$, and the statistical independence test between $\hat{\mathcal{A}}_{X \to Y \| (Z, \mathbf{W})} = Y - \hat{h}(X, \mathbf{W})$ 
and $\hat{\mathcal{Z}} = Z - \hat{\pi}(\mathbf{W})$ is then implemented using dataset $\mathcal{D}_2$. 

Next, we describe how to address the two main issues above, including estimating the functions $\hat{h}(\cdot)$ and $\hat{\pi}(\cdot)$ on $\mathcal{D}_1$, followed by conducting the independence test on $\mathcal{D}_2$. 

\textbf{Issue 1: Estimating the Functions on $\mathcal{D}_1$.} Regarding the estimation of the function $\hat{h}(\cdot)$, the intuitive idea is to apply a standard IV estimator to obtain $\hat{h}(X, \mathbf{W})$. Note that solving a nonparametric IV problem is, in general, highly nontrivial and may require regularization to address issues such as slow convergence rates, as thoroughly discussed in \cite{darolles2011nonparametric}. 
\citet{newey2003instrumental} show that, given a valid IV and under the completeness condition, 
the estimated function $\hat{h}(X, \mathbf{W})$ converges to the true function $f(X, \mathbf{W})$ in the population (infinite-sample) limit 
for exponential families or for discrete data with finite support. 
Conversely, if the IV is invalid, the population-limit estimator $\hat{h}(X, \mathbf{W})$ is biased relative to the true model. There are now many IV estimators for Additive Non-Parametric IV Models \citep{guo2016control, singh2019kernel, bennett2019deep}. Rather than proposing a new estimation strategy for nonparametric IV models, our paper focuses on developing a testable condition for identifying valid instruments using existing IV estimation methods, under an oracle scenario where the estimator is assumed to have access to the true structural function. Hence, we here adopt the control function IV estimator proposed by \citet{guo2016control}, which is a two-stage approach. In the first stage, it regresses the treatment variable $X$ on covariates $\mathbf{W}$ and instruments $(Z, \tau_2(Z), \dots, \tau_i(Z))$ (a known vector of linearly independent functions of $Z$), obtaining the predicted value $\hat{X}$ and the residual $e_1 = X - \hat{X}$. In the second stage, the outcome $Y$ is regressed on the treatment $X$, covariates $\mathbf{W}$, and the residual $e_1$ from the first stage regression. The coefficients of the second stage regression are taken as the control function estimates. Note that nonlinear functions can be estimated using basis functions. In our experiments, we apply polynomial basis functions as a predefined set of linearly independent functions, leveraging this control function method. 
\tworevise{Regarding the regression function $\hat{\pi}(\cdot)$, when the covariates $\mathbf{W}$ are present, we use the random forest regression method \citep{pedregosa2011scikit,breiman2001random} to fit the function $\hat{\pi}(\mathbf{W})$.}
}

\tworevise{\textbf{Issue 2: Testing Independence on $\mathcal{D}_2$.} 
As a first step, we compute the estimated auxiliary variable 
$\hat{\mathcal{A}}_{X \to Y \| (Z, \mathbf{W})} = Y - \hat{h}(X, \mathbf{W})$ 
and the residualized instrument 
$\hat{\mathcal{Z}} = Z - \hat{\pi}(\mathbf{W})$ using the dataset $\mathcal{D}_2$, 
where the functions $\hat{h}(\cdot)$ and $\hat{\pi}(\cdot)$ were estimated on $\mathcal{D}_1$ in the previous Issue 1. To check whether \tworevise{$\hat{\mathcal{A}}_{X \to Y \| (Z, \mathbf{W})}$ and $\hat{\mathcal{Z}}$} are statistically independent, we employ the Large-Scale HSIC Test, a Hilbert-Schmidt Independence Criterion (HSIC)-based test proposed by \citet{zhang2018large} for independence testing. If the output $p_{value}$ is less than the preset significance level $\alpha$, we reject the null hypothesis $H_0$, indicating that the candidate IV $Z$ is invalid. Conversely, if we fail to reject the null hypothesis, it suggests that $Z$ is a valid IV.

Based on the above discussions, the complete \emph{AIT Condition} test procedure is given in Algorithm \ref{Alg: AIT-Condition-Test}. To allow for greater flexibility in the algorithm and considering that prior knowledge, such as a constant causal effect, may sometimes be available, we include a two-stage least squares estimator suitable for linear effects \citep{basmann1957generalized,henckel2024graphical} to estimate the causal effect.
}

\begin{algorithm}[h]
\caption{AIT Condition}
\label{Alg: AIT-Condition-Test}
\begin{algorithmic}[1]
\tworevise{
\REQUIRE Observed dataset $\mathcal{D} = \{X, Y, \mathbf{W}, Z\}$, where $X$ is the treatment, $Y$ is the outcome, $\mathbf{W}$ denotes covariates, and $Z$ is a candidate IV; block size $B$; data split ratio $\rho$; significance level $\alpha$ 
\STATE Initialize: Result $\gets$ Do not reject $H_0$
\vspace{0.15cm}

\STATE \textbf{Sample Splitting.} Randomly split the dataset $\mathcal{D}$ with $|\mathcal{D}| = N$ into two disjoint subsets, $\mathcal{D}_1$ and $\mathcal{D}_2$, where $|\mathcal{D}_1| = n = \lfloor \rho N \rfloor$, and $
|\mathcal{D}_2| = m = N - n$.

\vspace{0.15cm}

\textbf{Step 1: Estimate functions $h(\cdot)$ and $\pi(\cdot)$ on $\mathcal{D}_1$} 

\IF{the constant-effect assumption is adopted}
    \STATE $\hat{h}(\cdot) \gets 
        \text{Two-Stage Least Squares Estimator }(X, Y, \mathbf{W}, Z)$ on $\mathcal{D}_1$
\ELSE
    \STATE $\hat{h}(\cdot) \gets 
        \text{Control Function IV Estimator }(X, Y, \mathbf{W}, Z)$ on $\mathcal{D}_1$
\ENDIF

\vspace{0.1cm}
\IF{$\mathbf{W}\neq \emptyset$}
    \STATE Fit the regression function $\hat{\pi}(\cdot)$ using the random forest regression method on $\mathcal{D}_1$
\ENDIF

\vspace{0.25cm}

\textbf{Step 2: Test the AIT Condition on $\mathcal{D}_2$}

\STATE Compute the auxiliary variable $\hat{\mathcal{A}}_{X \to Y \| (Z, \mathbf{W})} \gets Y - \hat{h}(X,\mathbf{W})$ on $\mathcal{D}_2$

\IF{$\mathbf{W} \neq \emptyset$}
    \STATE Compute the residualized instrument $\hat{\mathcal{Z}} \gets Z - \hat{\pi}(\mathbf{W})$ on $\mathcal{D}_2$
\ELSE
    \STATE Set $\hat{\mathcal{Z}} \gets Z$ on ${\mathcal{D}_2}$
\ENDIF

\STATE $p_{value} \gets 
    \text{Large-Scale HSIC Test}(\hat{\mathcal{A}}_{X \to Y\|(Z,\mathbf{W})}, \hat{\mathcal{Z}}, B)$

\IF{$p_{value} < \alpha$}
    \STATE Result $\gets$ Reject $H_0$
\ELSE
    \RETURN Result
\ENDIF

\ENSURE Result
}
\end{algorithmic}
\end{algorithm}

\tworevise{
\subsection{Theoretical Analysis of Type I and Type II Errors in the AIT Test}\label{Sec: Asymptotic Validity}
In this section, we investigate the asymptotic properties of the AIT test, focusing on both Type I and Type II errors. Despite being a plug-in procedure, the AIT test achieves asymptotic control of the Type I error at the nominal level under the null hypothesis and appropriate regularity conditions. At the same time, the AIT test exhibits desirable power properties against fixed alternatives.

In the implementation of the algorithm, we adopt the block-based HSIC statistic introduced in \citet{zhang2018large}. We first briefly review the relevant theoretical results on the Hilbert–Schmidt independence criterion for testing independence between two random variables; see \cite{zhang2018large, gretton2005measuring, gretton2008kernel}. Let $U$ and $V$ be random variables taking values in domains $\mathcal{U}$ and $\mathcal{V}$, respectively. Given kernels $k:\mathcal{U}\times \mathcal{U} \to \mathbb R$, $l:\mathcal{V}\times \mathcal{V} \to \mathbb R$, and the independently and identically distributed samples $S_{i=1}^m= \{u_i,v_i\}_{i=1}^m$, the samples are split into $\frac{m}{B}$ disjoint blocks, each of size $B$, denoted by $\{\{u_i^{(b)},v_i^{(b)}\}_{i=1}^B\}_{b=1}^{m/B}$. For each block $b \in \{1, \ldots, \frac{m}{B}\}$, the unbiased HSIC statistic is \begin{equation}\label{Eq:single-block-HSIC}
    \begin{aligned}
    \hat{\eta}_{b}:= \frac{1}{B(B-3)}\left[\operatorname{tr}\left(\tilde{K}^{(b)} \tilde{L}^{(b)}\right)+\frac{\mathbbm{1}^{T} \tilde{K}^{(b)} \mathbbm{1}\mathbbm{1}^{T} \tilde{L}^{(b)} \mathbbm{1}}{(B-1)(B-2)}\right. \left.-\frac{2}{B-2} \mathbbm{1}^{T} \tilde{K}^{(b)} \tilde{L}^{(b)} \mathbbm{1}\right], 
    \end{aligned}
\end{equation}
where \(\tilde{K} = K - \operatorname{diag}(K)\), with 
\(K_{i,j} := k(u_i,u_j)\), \ie, the kernel matrix with diagonal elements set to zero, and similarly for $\tilde{L}$. Here, $\mathbbm{1}$ denotes a vector of ones of relevant dimension. Throughout, we assume that the block size satisfies $B \ge 4$; when the sample size $m$ is not divisible by $B$, the number of blocks is taken to be $\lfloor m/B \rfloor$, so that the estimator is well defined.

The block-based HSIC estimator is then given by $\widehat{\mathrm{HSIC}} = \frac{B}{m}\sum_{b=1}^{m/B}\hat{\eta}_{b}$. 
According to \cite{zhang2018large}, in the asymptotic regime where $m \to \infty$, $B \to \infty$, and $m/B \to \infty$, we have 
\[
\sqrt{mB} \, \widehat{\mathrm{HSIC}} \xrightarrow{D} \mathcal{N}(0, \sigma^2),
\]
where $\sigma^2$ is the variance of the null distribution. 

For notational convenience, let $k$ and $l$ denote the kernels used to construct $\mathcal{A}$ and $\mathcal{Z}$ in the oracle setting, as well as $\hat{\mathcal{A}}$ and $\hat{\mathcal{Z}}$ in the estimated setting of the AIT test. For simplicity, we show the asymptotic validity of the AIT test without covariates. Before that, we first introduce regularity assumptions on the kernels required for our theoretical results.

\begin{Assumption}[Boundedness and Lipschitz Kernels]
\label{Ass: Bounded Lipschitz Kernels}
Let $\mathcal U \subset \mathbb R^{d_u}$ and 
$\mathcal V \subset \mathbb R^{d_v}$, and let $\|\cdot\|$ denote the Euclidean norm. 
\begin{enumerate}
\item[(1)] The kernel $l: \mathcal V \times \mathcal V \to \mathbb R$ 
is bounded, i.e., there exists a constant $M>0$ such that 
$|l(v,v')|\le M$ for all $v,v'\in\mathcal V$.

\item[(2)] The kernel $k: \mathcal U \times \mathcal U \to \mathbb R$ is Lipschitz continuous, i.e., there exists a constant $Q>0$ such that $|k(u,u')-k(\tilde u,\tilde u')|
\le Q(\|u-\tilde u\|+\|u'-\tilde u'\|)$
for all $u,u',\tilde u,\tilde u' \in \mathcal U$.
\end{enumerate}
\end{Assumption}

These two conditions are standard in the kernel literature. In particular, the Gaussian radial basis function (RBF) kernel satisfies both boundedness and Lipschitz continuity on compact domains. Moreover, similar regularity assumptions have been employed in kernel-based independence testing, including the regression-based conditional independence tests studied by \citet{polo2023conditional, ren2025regression}.

\begin{Assumption}[$\sqrt{n}$-Consistency of $\hat{h}$]\label{Ass:sqrt-n-consistency}
The estimator $\hat h$ satisfies the $\sqrt{n}$-convergence rate:
\[
\|h - \hat h\|_{L_2(P_{X})} = O_p(n^{-1/2}).
\]
where $h$ denotes the oracle function, and $\|\cdot\|_{L_2(P_X)}$ denotes the $L_2$ norm with respect to the distribution of $X$. 
\end{Assumption}

Note that fully nonparametric IV estimators generally fail to achieve the $\sqrt{n}$-rate due to the ill-posedness of the associated inverse problem \citep{horowitz2012specification,chen2011rate,hall2005nonparametric}. 
In the nonlinear case, we employ the control function IV estimator proposed by \citet{guo2016control}, which is semiparametric and attains the $\sqrt{n}$-rate under standard regularity conditions \citep{guo2016control,murphy2002estimation}. 
In the linear case, the two-stage least squares (2SLS) estimator also satisfies Assumption~\ref{Ass:sqrt-n-consistency} \citep{basmann1957generalized,henckel2024graphical,wooldridge2010econometric}.
}

\tworevise{
\begin{Theorem}[Asymptotic Level and Power of the AIT Test]\label{The:HSIC}
Suppose that Assumptions~\ref{Ass: Bounded Lipschitz Kernels} and \ref{Ass:sqrt-n-consistency} hold, and that the sample sizes satisfy $mB = o(n)$ as $n \to \infty$.
\begin{itemize}
\item \textbf{Type I Error:} Under the null hypothesis $H_0$, 
$$P_{H_{0}}(\text {Type I error})=P_{H_{0}}\left(\sqrt{mB}\widehat{\mathrm{HSIC}}\left(\hat{\mathcal{A}}, \mathcal{Z}\right)>c_{m,\alpha}\right) \leq \alpha + o(1),$$
\item \textbf{Type II Error:} Under any fixed alternative hypothesis 
$H_1$ such that $\mathrm{HSIC}(\mathcal A,\mathcal Z)=\eta>0$, 
$$P_{H_1}(\text{Type II error}) = P_{H_1}\left(\sqrt{mB}\widehat{\mathrm{HSIC}}(\hat{\mathcal{A}}, \mathcal{Z}) \leq c_{m,\alpha}\right) = O\left(\frac{1}{\sqrt{mB}} e^{-c mB}\right), \text{ } mB\to\infty,$$
\end{itemize}
where $c_{m,\alpha}$ denotes the critical value, and the constant $c = \frac{\eta^2}{2\sigma_{H_1}^2} > 0$, and $\sigma_{H_1}^2$ is the asymptotic variance. 
\end{Theorem}

\begin{proof}
    See Appendix \ref{App:proof of HSIC} for its proof. 
\end{proof}

Theorem~\ref{The:HSIC} establishes bounds on both Type I and Type II errors of the AIT test while simultaneously accounting for errors in both the estimation and independence testing stages. Specifically, the test controls the Type~I error at the nominal level $\alpha$ asymptotically and achieves vanishing Type~II error under fixed alternatives with $\widehat{\mathrm{HSIC}}(\mathcal A,\mathcal Z)>0$. 

We next consider the theoretical asymptotic validity of the AIT test with covariates $\mathbf{W}$. To this end, we introduce three assumptions that will be used to establish the result.

\begin{Assumption}[Boundedness and Lipschitz Kernels]
\label{Ass: Bounded Lipschitz Kernels with covariates}
Let $\mathcal{U} \subset \mathbb{R}^{d_u}$ and $\mathcal{V} \subset \mathbb{R}^{d_v}$, and let $\|\cdot\|$ denote the Euclidean norm. The kernels $k: \mathcal{U} \times \mathcal{U} \to \mathbb{R}$ and $l: \mathcal{V} \times \mathcal{V} \to \mathbb{R}$ satisfy the following conditions: 
\begin{enumerate}
    \item[(1)] 
    There exists a constant $M > 0$ such that $|k(u, u')| \le M$ and $|l(v, v')| \le M$ for all $u, u' \in \mathcal{U}$ and all $v, v' \in \mathcal{V}$.
    \item[(2)] 
    There exists a constant $Q > 0$ such that $|k(u,u') - k(\tilde{u}, \tilde{u}')| \le Q(\|u - \tilde{u}\| + \|u' - \tilde{u}'\|)$ and $|l(v,v') - l(\tilde{v}, \tilde{v}')| \le Q(\|v - \tilde{v}\| + \|v' - \tilde{v}'\|)$ for all 
    $u, u', \tilde{u}, \tilde{u}' \in \mathcal{U}$ and all $v, v', \tilde{v}, \tilde{v}' \in \mathcal{V}$. 
\end{enumerate}
\end{Assumption}

With covariates, $\mathcal{Z}$ is estimated. Unlike Assumption~\ref{Ass: Bounded Lipschitz Kernels}, which requires boundedness and Lipschitz continuity for only one kernel ($k$ or $l$), Assumption~\ref{Ass: Bounded Lipschitz Kernels with covariates} imposes these conditions on both.

\begin{Assumption}[$\sqrt{n}$-Consistency of $\hat{h}$ with Covariates]\label{Ass:sqrt-n-consistency-covariates}
The estimator $\hat h$ satisfies the $\sqrt{n}$-convergence rate:
\[
\|h - \hat h\|_{L_2(P_{(X,\mathbf{W})})} = O_p(n^{-1/2}).
\]
where $h$ denotes the oracle function, and $\|\cdot\|_{L_2(P_{(X,\mathbf{W})})}$ denotes the $L_2$ norm with respect to the distribution of $(X,\mathbf{W})$. 
\end{Assumption}

Assumption \ref{Ass:sqrt-n-consistency-covariates} is an extension of Assumption \ref{Ass:sqrt-n-consistency} that takes covariates into account. 

\begin{Assumption}[$n^{q}$-Consistency of $\hat{\pi}$]\label{Ass:rate}
The estimator $\hat{\pi}$ satisfies the convergence rate 
\[
\|\pi-\hat{\pi}\|_{L_2(P_\mathbf{W})} = O_p(n^{-q}),
\]
for some $q \in (0, 1/2)$, where $\|\cdot\|_{L_2(P_\mathbf{W})}$ denotes the $L_2$ norm with respect to the distribution of $\mathbf{W}$. 
\end{Assumption}

Assumption~\ref{Ass:rate} is formulated as a high-level rate condition on the regression estimator. 
In our implementation, we employ random forests to estimate $\hat{\pi}$. Random forests are known to be consistent under suitable structural assumptions, such as additive or low-dimensional models \citep{scornet2015consistency}. 
More generally, under standard nonparametric smoothness assumptions (e.g., Hölder smoothness of order $s$ in dimension $d$), regression estimators typically achieve convergence rates of the form $n^{-s/(2s+d)}$, which corresponds to some $q \in (0, 1/2)$ \citep{stone1982optimal,yang2015minimax}. Therefore, Assumption~\ref{Ass:rate} is compatible with commonly used nonparametric estimators, including random forests.

\begin{Corollary}[Asymptotic Level and Power of the AIT Test with Covariates]\label{The-HSIC-covariate}
$\quad$ Suppose that Assumptions~\ref{Ass: Bounded Lipschitz Kernels with covariates} $\sim$ \ref{Ass:rate} hold, and that the sample sizes satisfy $mB = o(n^{2q})$ as $n \to \infty$. 
\begin{itemize}
\item \textbf{Type I Error:}
Under the null hypothesis,  
$${P}(\text{Type I error}) = {P}_{H_0}(\sqrt{mB}\,\widehat{\mathrm{HSIC}}(\hat{\mathcal{A}},\hat{\mathcal{Z}}) > c_{m,\alpha}) \le \alpha + o(1),$$
\item \textbf{Type II Error:} Under any fixed alternative hypothesis 
$H_1$ such that $\mathrm{HSIC}(\mathcal A,\mathcal Z)=\eta>0$, 
$$P_{H_1}(\text{Type II error}) = P_{H_1}\left(\sqrt{mB}\widehat{\mathrm{HSIC}}(\hat{\mathcal{A}}, \hat{\mathcal{Z}}) \leq c_{m,\alpha}\right) = O\left(\frac{1}{\sqrt{mB}} e^{-c mB}\right), \text{ } mB\to\infty,$$
\end{itemize}
where $c_{m,\alpha}$ denote the critical value, and the constant $c = \frac{\eta^2}{2\sigma_{H_1}^2} > 0$, and $\sigma_{H_1}^2$ is the asymptotic variance.  
\end{Corollary}
\begin{proof}
    See Appendix \ref{App:proof of HSIC of covariates} for its proof. 
\end{proof}

Corollary~\ref{The-HSIC-covariate} establishes that the proposed test controls the Type~I error asymptotically and achieves consistency under fixed alternatives, with power converging to one as $mB \to \infty$.

\begin{Remark}
We would like to mention two main tuning parameters in practical implementations: the data split ratio $\rho$ and the block size $B$. Regarding the data split ratio $\rho$: 
Theoretically, our method requires the rate condition $mB = o(n)$ (or $mB = o(n^{2q})$ in the presence of covariates). In practice, we heuristically set the data split ratio to $\rho = 0.7$, following \cite{polo2023conditional,wang2025practical}. Moreover, the experimental results in Section~\ref{Section:experiments} show that this choice leads to favorable performance. Regarding the block size $B$: The sample size of each block $B$ can be guided by \cite{zaremba2013b,zhang2018large}, who showed that the null distribution is well approximated by the central limit theorem when $B$ is small, which thus helps the Type~I error approach the nominal level. However, a smaller $B$ may lead to reduced statistical power for a given sample size. In our experiments (Section~\ref{Section:experiments}), we heuristically set $B = m/15$ in the continous setting and $B = m/6$ in the discrete setting, which demonstrates superior performance. 
\end{Remark}

}

\section{Experiments}\label{Section:experiments}
In this section, we evaluated the performance of the proposed AIT condition for instrument validity across both synthetic data and three real-world datasets. Our source code is available at \url{https://github.com/zhengli0060/AIT_Condition}.
\subsection{Synthetic Data}\label{SubSection-Synthetic}
\tworevise{
We conducted simulation experiments from four perspectives. First, in Section \ref{Subsection-simulation-theory}, we verified the correctness of our theoretical results. Then, in Section \ref{Subsection-simulation-Continous}, we compared the IV-PIM method proposed by \cite{burauel2023evaluating}, which is designed for continuous treatments with covariates. Moreover, in Section \ref{Subsection-simulation-Discrete}, we compared the proposed method with Kitagawa's method, abbreviated as the K-test method \citep{kitagawa2015test}, which is designed for discrete treatments. Finally, in Section~\ref{Section-Supple-simulation}, we validate our theoretical results under an oracle setting, where the auxiliary variable is correctly specified (i.e., no estimation error in Step 1 of Algorithm~\ref{Alg: AIT-Condition-Test}). 
In all experiments, we evaluated the performance of our method using the following metrics: 
\begin{itemize}
[leftmargin=30pt,align=left,itemsep=1pt,topsep=0pt]
    \item \textbf{Valid IVs Misidentification Ratio (abbreviated as Valid MR)}: the ratio of the number of valid IVs incorrectly identified in the output.
    \item \textbf{Invalid IVs Misidentification Ratio (abbreviated as Invalid MR)}: the ratio of the number of invalid IVs incorrectly identified in the output.
\end{itemize}

\subsubsection{Theoretical Validation of Proposed Method}\label{Subsection-simulation-theory}
\vspace{-1mm}
In this section, we conducted simulations to verify our theoretical results for both the linear model and the ANINCE model. {Table~\ref{Summary_of_experiments} summarizes the propositions validated by the experimental results reported in various tables, along with the corresponding model settings.} Specifically, we tested our method's ability to identify invalid IVs that violate the \emph{exogeneity} condition under three scenarios: linear constant-effect setting with varied distributions, partially nonlinear constant conditions (with Gaussian noise terms), and partially nonlinear non-constant conditions with varied functional forms but a consistent Uniform distribution. The performance results are presented in Table \ref{Table-linear-model-A2}, Table \ref{Table-non-linear-constant-model-A2}, and Table \ref{Table-non-linear-non-constant-model-A2}, respectively. Further, we examined the proposed method’s effectiveness in identifying invalid IVs that violate the \emph{exclusion restriction} condition under nonlinear constant and nonlinear non-constant conditions, considering various functional forms as well as Beta and Uniform distributions, respectively. Results for these tests are shown in Table \ref{Table-non-linear-constant-model-A3} and Table \ref{Table-non-linear-non-constant-model-A3}.

\begin{table}[htbp]
\revise{
    \centering
    \caption{\revise{Summary of propositions validated by the experimental results, along with the corresponding model settings.}}
    \label{Summary_of_experiments}
    \resizebox{1.0\textwidth}{!}{
    \begin{tabular}{l l l l}
        \toprule
        \textbf{Table} & \textbf{Proposition Demonstrated} & \textbf{Model} & \textbf{IV Violation} \\
        \hline
        Table \ref{Table-linear-model-A2} & Propositions \ref{proposition-linear-gaussian-model} and \ref{Proposition-linear-exogeneity} & Linear & Exogeneity \\
        Table \ref{Table-non-linear-constant-model-A2} & Proposition \ref{Proposition-identifiable-IV-ANINCE} & Partial Non-Linear Constant Effect& Exogeneity  \\
        Table \ref{Table-non-linear-non-constant-model-A2} & Proposition \ref{Proposition-identifiable-IV-ANINCE} & Partial Non-Linear Non-Constant Effect&Exogeneity \\
        Table \ref{Table-non-linear-constant-model-A3} & Proposition \ref{Proposition-identifiable-IV-ANINCE} & Partial Non-Linear Constant Effect& Exclusion Restriction\\
        Table \ref{Table-non-linear-non-constant-model-A3} & Proposition \ref{Proposition-identifiable-IV-ANINCE} & Non-Linear Non-Constant Effect& Exclusion Restriction\\
        \bottomrule
    \end{tabular}
    }
\begin{tablenotes}
\item \small{Note: ``Proposition Demonstrated'' refers to the propositions that are supported by the experimental results presented in each table. 
}
\end{tablenotes}
}
\end{table}

\textbf{Experimental Design:} We generated data according to the model in Equation \eqref{Eq-Main-Model}, where each candidate IV set contains both a valid IV and an invalid IV. For the nonlinear non-constant effects model, we applied five functions-{\emph{logarithmic, quadratic polynomial, cubic polynomial, logarithmic quadratic polynomial, exponent quadratic polynomial}}-for $g(\cdot)$, $f(\cdot)$ and $\varphi_{*}(\cdot)$. In contrast, these functions were linear under linear, constant settings. For the error term $\varepsilon_{*}$, we selected six distributions\footnote{The ``Mixed distribution'' refers to that obtained by randomly selecting from the mentioned distributions.}: \emph{Gaussian, Uniform, T, Beta, Gamma, Log-normal}. In all experiments, the significance level was set to 0.05 for sample sizes smaller than 1000, and to 0.01 otherwise, with the sample splitting ratio fixed at $\rho = 0.7$. 
Each experiment was repeated 100 times with randomly generated data, and the reported results were averaged. The sample sizes are chosen from 3000 (3k), 5000 (5k), and 7000 (7k). Additional experimental details can be found in Appendix \ref{Appendix-Synthetic}.

\begin{table*}[hpt!]
\vspace{-1.mm}
\tworevise{
\centering
\small
\caption{Performance in Testing Violations of Exogeneity within Linear Models.}
\resizebox{0.98\textwidth}{!}{
\begin{tabular}{c|cc|cc|cc}
\toprule
\multicolumn{1}{c}{} & \multicolumn{2}{|c|}{Size=3K} & \multicolumn{2}{|c|}{Size=5K} & \multicolumn{2}{c}{Size=7K}\\
\hline
Distribution & {\textbf{Valid MR}$\downarrow $} %
&{\textbf{Invalid MR}$\downarrow $}  & {\textbf{Valid MR}}$\downarrow $ %
&{\textbf{Invalid MR}$\downarrow $} & {\textbf{Valid MR}}$\downarrow $ 
&{\textbf{Invalid MR}$\downarrow $}\\
\hline
Uniform 
& 0.00 & 0.04 
& 0.00 & 0.00 
& 0.00 & 0.00 \\
Beta
& 0.01 & 0.12 
& 0.00 & 0.03 
& 0.00 & 0.01 \\
T-distribution
& 0.02 & 0.08 
& 0.00 & 0.04 
& 0.00 & 0.01 \\
Gamma 
& 0.01 & 0.08 
& 0.00 & 0.01 
& 0.00 & 0.01 \\
Log-normal 
& 0.01 & 0.08 
& 0.00 & 0.02 
& 0.00 & 0.00 \\
Gaussian 
& 0.00 & 1.00 
& 0.00 & 1.00 
& 0.00 & 1.00 \\
Mixed 
& 0.01 & 0.20 
& 0.00 & 0.15 
& 0.00 & 0.05 \\
\bottomrule
\end{tabular}
}
\label{Table-linear-model-A2}
\begin{tablenotes}
\item \small{Note: 
$\downarrow $ means a lower value is better, and vice versa.}
\end{tablenotes}
\vspace{-4mm}
}
\end{table*}

\begin{table*}[hpt!]
\tworevise{
\centering
\small
\caption{Performance in Testing Violations of Exogeneity within Partial Non-Linear Constant Effect Models.}
\resizebox{0.98\textwidth}{!}{
\begin{tabular}{c|cc|cc|cc}
\toprule
\multicolumn{1}{c}{} & \multicolumn{2}{|c|}{Size=3K} & \multicolumn{2}{|c|}{Size=5K} & \multicolumn{2}{c}{Size=7K}\\
\hline
Function & {\textbf{Valid MR}$\downarrow $} 
&{\textbf{Invalid MR}$\downarrow $}  & {\textbf{Valid MR}}$\downarrow $ 
&{\textbf{Invalid MR}$\downarrow $} & {\textbf{Valid MR}}$\downarrow $ 
&{\textbf{Invalid MR}$\downarrow $}\\
\hline
Log
& 0.00 & 0.01 
& 0.00 & 0.00 
& 0.00 & 0.00 \\
Quadratic polynomial
& 0.00 & 0.00 
& 0.00 & 0.00 
& 0.00 & 0.00 \\
Cubic polynomial 
& 0.00 & 0.00 
& 0.00 & 0.00 
& 0.00 & 0.00 \\
Log(quadratic) 
& 0.00 & 0.00 
& 0.00 & 0.00 
& 0.00 & 0.00 \\
Exp(quadratic)
& 0.00 & 0.00 
& 0.00 & 0.00 
& 0.00 & 0.00 \\
\bottomrule
\end{tabular}
}
\label{Table-non-linear-constant-model-A2}
\begin{tablenotes}
\item \small{Note: $\downarrow$ means a lower value is better, and vice versa.}
\end{tablenotes}
}
\vspace{-4mm}
\end{table*}

\begin{table*}[hpt!]
\tworevise{
\centering
\small
\caption{Performance in Testing Violations of Exogeneity within Partial Non-Linear Non-Constant Effect Models.}
\resizebox{0.98\textwidth}{!}{
\begin{tabular}{c|cc|cc|cc}
\toprule
\multicolumn{1}{c}{} & \multicolumn{2}{|c|}{Size=3K} & \multicolumn{2}{|c|}{Size=5K} & \multicolumn{2}{c}{Size=7K}\\
\hline
Function & {\textbf{Valid MR}$\downarrow $} 
&{\textbf{Invalid MR}$\downarrow $}  & {\textbf{Valid MR}}$\downarrow$ 
&{\textbf{Invalid MR}$\downarrow $} & {\textbf{Valid MR}}$\downarrow $ 
&{\textbf{Invalid MR}$\downarrow$}\\
\hline
Log
& 0.00 & 0.00 
& 0.00 & 0.00 
& 0.00 & 0.00 \\
Quadratic polynomial
& 0.00 & 0.00 
& 0.00 & 0.00 
& 0.00 & 0.00 \\
Cubic polynomial 
& 0.00 & 0.00 
& 0.00 & 0.00 
& 0.00 & 0.00 \\
Log(quadratic) 
& 0.00 & 0.00 
& 0.00 & 0.00 
& 0.00 & 0.00 \\
Exp(quadratic)
& 0.00 & 0.01 
& 0.00 & 0.00 
& 0.00 & 0.00 \\
\bottomrule
\end{tabular}
}
\label{Table-non-linear-non-constant-model-A2}
\begin{tablenotes}
\item \small{Note: $\downarrow$ means a lower value is better, and vice versa.}
\end{tablenotes}
}
\vspace{-4mm}
\end{table*}

\begin{table*}[hpt!]
\tworevise{
\centering
\small
\caption{Performance in Testing Violations of Exclusion Restriction within Partial Non-Linear Constant Effect Models.}
\resizebox{0.98\textwidth}{!}{
\begin{tabular}{c|cc|cc|cc}
\toprule
\multicolumn{1}{c}{} & \multicolumn{2}{|c|}{Size=3K} & \multicolumn{2}{|c|}{Size=5K} & \multicolumn{2}{c}{Size=7K}\\
\hline
Function & {\textbf{Valid MR}$\downarrow$}  
&{\textbf{Invalid MR}$\downarrow$}  & {\textbf{Valid MR}}$\downarrow$  
&{\textbf{Invalid MR}$\downarrow$} & {\textbf{Valid MR}}$\downarrow$  
&{\textbf{Invalid MR}$\downarrow$}\\
\hline
Log
& 0.02 & 0.11 
& 0.01 & 0.08 
& 0.00 & 0.04 \\
Quadratic polynomial
& 0.00 & 0.09 
& 0.00 & 0.05 
& 0.00 & 0.02 \\
Cubic polynomial
& 0.01 & 0.12
& 0.00 & 0.07 
& 0.00 & 0.01 \\
Log(quadratic) 
& 0.00 & 0.16
& 0.00 & 0.09 
& 0.00 & 0.04 \\
Exp(quadratic)
& 0.01 & 0.09 
& 0.00 & 0.05 
& 0.00 & 0.03 \\
\bottomrule
\end{tabular}
}
\label{Table-non-linear-constant-model-A3}
\vspace{-1mm}
\begin{tablenotes}
\item \small{Note: $\downarrow $ means a lower value is better, and vice versa.}
\end{tablenotes}
}
\end{table*}

\begin{table*}[hpt!]
\tworevise{
\centering
\small
\caption{Performance in Testing Violations of Exclusion Restriction within Non-Linear Non-Constant Effect Models.}
\resizebox{0.98\textwidth}{!}{
\begin{tabular}{c|cc|cc|cc}
\toprule
\multicolumn{1}{c}{} & \multicolumn{2}{|c|}{Size=3K} & \multicolumn{2}{|c|}{Size=5K} & \multicolumn{2}{c}{Size=7K}\\
\hline
Function & {\textbf{Valid MR}$\downarrow$} 
&{\textbf{Invalid MR}$\downarrow$}  & {\textbf{Valid MR}}$\downarrow $ 
&{\textbf{Invalid MR}$\downarrow$} & {\textbf{Valid MR}}$\downarrow$ 
&{\textbf{Invalid MR}$\downarrow$}\\
\hline
Log
& 0.00 & 0.16 
& 0.00 & 0.12 
& 0.00 & 0.04 \\
Quadratic polynomial
& 0.00 & 0.02 
& 0.00 & 0.00 
& 0.00 & 0.00 \\
Cubic polynomial
& 0.00 & 0.02 
& 0.00 & 0.00 
& 0.00 & 0.00 \\
Log(quadratic) 
& 0.01 & 0.06 
& 0.00 & 0.02 
& 0.00 & 0.00 \\
Exp(quadratic)
& 0.00 & 0.00 
& 0.00 & 0.00 
& 0.00 & 0.00 \\
\bottomrule
\end{tabular}
}
\label{Table-non-linear-non-constant-model-A3}
\vspace{-1mm}
\begin{tablenotes}
\item \small{Note: $\downarrow$ means a lower value is better, and vice versa.}
\end{tablenotes}
}
\end{table*}

\textbf{Results:} As shown in Tables \ref{Table-linear-model-A2} $\sim$ \ref{Table-non-linear-non-constant-model-A3}, both metrics generally improve significantly with increasing sample sizes across various distributions and functions. Those facts suggest that our method can correctly identify invalid IVs that violate the \emph{exogeneity} condition or violate the \emph{exclusion restriction} condition. More specifically, Table \ref{Table-linear-model-A2} highlights that non-Gaussianity is beneficial for identifying IV in linear models, as demonstrated in Figure \ref{fig:SP_Linear_Model} and supported by Proposition \ref{Proposition-linear-exogeneity}. 
Notably, the Invalid MR value in the Gaussian distribution of the linear model in Table \ref{Table-linear-model-A2} is 1, indicating that our method cannot detect invalid IVs in a linear Gaussian model. This finding is consistent with the conclusions presented in Proposition \ref{proposition-linear-gaussian-model}. 
Tables \ref{Table-non-linear-constant-model-A2} $\sim$ \ref{Table-non-linear-non-constant-model-A2} further show that even a slight degree of nonlinearity facilitates the assessment of IV validity under the exogeneity condition in both nonlinear constant and non-constant effect models, as illustrated in Figure \ref{fig:SP_Violate_A2_assumption-nonlinear} and stated in Proposition \ref{Proposition-identifiable-IV-ANINCE}. 
Lastly, Tables \ref{Table-non-linear-constant-model-A3} $\sim$ \ref{Table-non-linear-non-constant-model-A3} reveal that in nonlinear models, when the direct causal effect of $Z \to Y$ does not follow a linear function of the effect of $Z \to X$, it becomes possible to assess the validity of IVs solely concerning the exclusion restriction condition, as highlighted in the negative of Proposition \ref{Pro-violate-algebraic} $(i)$ and elaborated in Proposition \ref{Proposition-identifiable-IV-ANINCE}.

\subsubsection{Comparison with IV-PIM in Continuous Treatment Setting}\label{Subsection-simulation-Continous}
In this section, we compared the proposed \emph{AIT Condition} with IV-PIM, as proposed by \citet{burauel2023evaluating}, in continuous treatment settings. Note that since IV-PIM requires covariates, we introduce covariates here for a fair comparison.

\textbf{Experimental Design:} The specific generation mechanism with covariates $\mathbf{W}$ in the linear model is defined as follows: $U =\varepsilon_U$, $\mathbf{W} = \boldsymbol{\varepsilon_{W}}$, $Z_1 =\mathcal{I}(U + \mathbf{W} + \varepsilon_{Z1})$, $Z_2= \mathcal{I}(\mathbf{W} + \varepsilon_{Z2})$, $X = 0.5Z_1 + 0.5Z_2 + \boldsymbol{\lambda W} + \delta$, and $Y = X + \mathbf{W} + \epsilon$, where $\varepsilon_{U} \sim T(5)$, $\varepsilon_{Z1} \sim Beta(0.5, 0.1)$, $\varepsilon_{Z2} \sim \mathcal{N}(0,1)$, and  $\delta, \epsilon \sim T(5)$. Here, $\mathcal{I}(*)$ is the indicator function such that $\mathcal{I}(*) > \text{mean}(*)$ equals 1; otherwise, it is 0. The coefficient $\boldsymbol{\lambda}$ is randomly drawn from a normalized standard normal distribution. The noise terms $\boldsymbol{\varepsilon_{W}}$ follow a multidimensional normal distribution and are consistent with IV-PIM, with the dimensionality of covariates $\mathbf{W}$ varying across $|\mathbf{W}|=\{2, 3, 5\}$. 
The remaining settings are the same as in Section \ref{Subsection-simulation-theory}.

\textbf{Results:} As shown in Table \ref{Table-IV-PIM-compared-linear-model}, our method outperforms experimental results of IV-PIM with covariates under both Valid MR and Invalid MR. Interestingly, IV-PIM’s performance improves as the dimensionality of covariates increases, consistent with findings in \cite{burauel2023evaluating}. Additionally, 
Table \ref{Table-IV-PIM-compared-linear-model} highlights the practicality of the AIT condition with covariates, as presented in Corollary \ref{Corollary-Necessary-Condition-IV-covariates}. 

\begin{table*}[hpt!]
\tworevise{
\centering
\small
\caption{Performance in Testing Instrumental Variables with Covariates.}
\resizebox{0.98\textwidth}{!}{
\begin{tabular}{c|c|cc|cc|cc}
\toprule
\multicolumn{1}{c}{}& \multicolumn{1}{c|}{} & \multicolumn{2}{|c|}{Size=3K} & \multicolumn{2}{|c|}{Size=5K} & \multicolumn{2}{c}{Size=7K}\\
\hline
$|\mathbf{W}|$& Condition & {\textbf{Valid MR}$\downarrow $}
&{\textbf{Invalid MR}$\downarrow$}  & {\textbf{Valid MR}}$\downarrow$ 
&{\textbf{Invalid MR}$\downarrow$} & {\textbf{Valid MR}}$\downarrow$   
&{\textbf{Invalid MR}$\downarrow$}\\
\hline
2&IV-PIM method~\citep{burauel2023evaluating}
& 0.22 & 0.25 
& 0.22 & 0.18  
& 0.32 & 0.22\\
&AIT condition
& 0.00 & 0.11 
& 0.00 & 0.01  
& 0.00 & 0.00\\
\hline
3&IV-PIM method~\citep{burauel2023evaluating}
& 0.15 & 0.27  
& 0.20 & 0.21   
& 0.15 & 0.16\\
&AIT condition
& 0.01 & 0.05 
& 0.00 & 0.00  
& 0.00 & 0.00\\
\hline
5&IV-PIM method~\citep{burauel2023evaluating}
& 0.11 & 0.16  
& 0.06 & 0.25   
& 0.08 & 0.23 \\
&AIT condition
& 0.00 & 0.02 
& 0.00 & 0.00  
& 0.00 & 0.00\\
\bottomrule
\end{tabular}
}
\label{Table-IV-PIM-compared-linear-model}
\begin{tablenotes}
\item \small{Note: $\downarrow $ means a lower value is better, and vice versa.}
\end{tablenotes}
}
\end{table*}

\subsubsection{Comparison with K-test in Discrete Treatment Setting}\label{Subsection-simulation-Discrete}
In this section, we compared the proposed instrument validity test with the K-test, which was proposed by \citet{kitagawa2015test} for discrete treatment settings without covariates. 
\sloppy
The source code for the K-test is available at 
\texttt{https://rdrr.io/github/CarrThomas/TestforInstrumentValidity/}.

\textbf{Experimental Design:} 
The discrete-treatment data were generated to simulate violations of the \emph{exogeneity} and \emph{exclusion restriction} conditions as follows:
    $U = \varepsilon_{U}$, $Z =  \mathcal{I} ( \varphi_{Z}(U) + \varepsilon_{Z})$, $X = \mathcal{I} (g_{X}(Z)+ \varphi_{X}(U) + \varepsilon_{X})$, $Y = \beta X + g_{Y}(Z) + \varphi_{Y}(U) + \varepsilon_{Y}$, and $\varepsilon_{*} \sim \mathcal{N}(0,1)$, where the causal effect $\beta = 1$, and $\mathcal{I}(*)$ is the indicator function such that $\mathcal{I}(*) > \text{mean}(*)$ equals 1; otherwise, it is 0. The functions $\varphi_{*}(U)$ and $g_{*}(Z)$ are nonlinear and randomly selected from the following set: {\emph{cos, sin, square, cubic(third-degree polynomials), logarithmic, exponential}}.
    \tworevise{The significance level is set to be the same for both methods at each sample size.}
    The remaining settings are the same as in Section \ref{Subsection-simulation-theory}.
\textbf{Results:} 
As shown in Table~\ref{Table-binary-treatment-model}, the proposed AIT condition exhibits competitive performance relative to the K-test in terms of Invalid MR. 
Specifically, AIT achieves lower Invalid MR at sample sizes of 3K and 7K, while yielding comparable performance at 5K. 
For the \emph{Valid MR} metric, both methods achieve a value of 0, indicating that neither method mistakenly identifies valid IVs as invalid.

\begin{table*}[hpt!]
\tworevise{
\centering
\small
\caption{Performance in Testing Instrumental Variables with Discrete Treatment.}
\resizebox{0.98\textwidth}{!}{
\begin{tabular}{c|cc|cc|cc}
\toprule
\multicolumn{1}{c}{} & \multicolumn{2}{|c|}{Size=3K} & \multicolumn{2}{|c|}{Size=5K} & \multicolumn{2}{c}{Size=7K}\\
\hline
Condition & {\textbf{Valid MR}$\downarrow$}
&{\textbf{Invalid MR}$\downarrow$}  & {\textbf{Valid MR}}$\downarrow$  
&{\textbf{Invalid MR}$\downarrow$} & {\textbf{Valid MR}}$\downarrow$  
&{\textbf{Invalid MR}$\downarrow$}\\
\hline
K-test method~\citep{kitagawa2015test}
& 0.00 & 0.20 
& 0.00 & 0.10  
& 0.00 & 0.08\\
AIT condition
& 0.00 & 0.15 
& 0.00 & 0.13  
& 0.00 & 0.07\\
\bottomrule
\end{tabular}
}
\label{Table-binary-treatment-model}
\begin{tablenotes}
\item \small{Note: $\downarrow$ means a lower value is better, and vice versa.}
\end{tablenotes}
}
\end{table*}

}

\subsubsection{Theoretical Validation with the Auxiliary Variable Correctly Specified}\label{Section-Supple-simulation} 
\revise{In this section, we present additional simulation experiments to verify the theoretical results with the auxiliary variable correctly specified. 
Specifically, we report results for both valid and invalid IV testing cases. 
\begin{itemize}
    \item For \textbf{valid IV testing}, we evaluate the Valid MR under three different model settings: the Linear Model, the Partial Non-Linear Model with Constant Causal Effect, and the ANINCE Model. 
    \item For \textbf{invalid IV testing}, we evaluate the Invalid MR under various scenarios where IV assumptions are violated. These include: the Linear Model with exogeneity violated; the Linear Model with both exogeneity and exclusion restriction violated; the Partial Non-Linear Model with Constant Causal Effect where exogeneity is violated; the Partial Non-Linear Model with Constant Causal Effect where exclusion restriction is violated; and the Partial Non-Linear Model with Constant Causal Effect where both exogeneity and exclusion restriction are violated.
\end{itemize}
It is worth noting that for the nonparametric invalid IV model, we are unable to derive the explicit closed-form solution for the auxiliary variable. As a result, we cannot report the Invalid MR for that setting. Besides, in each model, we include only a single candidate IV to focus on the core objective of the test—assessing the validity of that IV.
The remaining settings are the same as in Section \ref{Subsection-simulation-theory}.}

\revise{\textbf{Results.} \revise{As shown in Tables \ref{Table-type-I} and \ref{Table-type-II}, both Valid MR and Invalid MR are low across different sample sizes and model specifications. Specifically, 
\begin{itemize}
    \item The Valid MR is consistently low in both the constant and non-constant effect models, indicating that the test appropriately controls the false positive rate when the auxiliary variable is correctly specified. 
    \item The Invalid MR is also very low in the constant effect model with an invalid IV, suggesting that the test has strong power to detect violations of IV validity. 
\end{itemize}

In summary, these results provide empirical support for the validity of our test when the auxiliary variable is assumed to be correctly specified.

}}

\begin{table*}[hpt!]
\centering
\small
\caption{\revise{Performance in Testing Valid IV with the Auxiliary Variable Correctly Specified in Different Cases.}}
\begin{tabular}{c|c|c|c}
\toprule
\multicolumn{1}{c}{} & 
\multicolumn{3}{|c}{\textbf{Valid MR}}\\
\hline
\textbf{Cases} & {\textbf{Size=3K}}  
&{\textbf{Size=5K}} 
& {\textbf{Size=7K}} \\
\hline
Linear Model 
& 0.00  
& 0.00  
& 0.00 \\
Partial Non-Linear Model with 
Constant Causal Effect 
& 0.00 
& 0.00  
& 0.00 \\
ANINCE Model 
& 0.00 
& 0.00 
& 0.00 \\
\bottomrule
\end{tabular}
\label{Table-type-I}
\end{table*}

\begin{table*}[hpt!]
\centering
\small
\caption{\revise{Performance in Testing Invalid IV with the Auxiliary Variable Correctly Specified in Different Cases.}}
\resizebox{1\textwidth}{!}{
\begin{tabular}{c|c|c|c}
\toprule
\multicolumn{1}{c}{} & 
\multicolumn{3}{|c}{\textbf{Invalid MR}}\\
\hline
\textbf{Cases} & {\textbf{Size=3K}}
&{\textbf{Size=5K}} 
& {\textbf{Size=7K}} \\
\hline
Linear Model: Exogeneity Violated
& 0.01  
& 0.00  
& 0.00 \\
Linear Model: Exogeneity $\&$ Exclusion Restriction Violated
& 0.00 
& 0.00 
& 0.00  \\
Partial Non-Linear Model with 
Constant Causal Effect: Exogeneity Violated 
& 0.00 
& 0.00  
& 0.00 \\
Partial Non-Linear Model with 
Constant Causal Effect: Exclusion Restriction Violated 
& 0.02 
& 0.00 
& 0.00 \\
Partial Non-Linear Model with 
Constant Causal Effect: Exogeneity $\&$ 
Exclusion Restriction Violated 
& 0.01 
& 0.00 
& 0.00 \\
\bottomrule
\end{tabular}
}
\label{Table-type-II}

\end{table*}

\subsection{Real-World Datasets}\label{Subsection-Realdata}
In this section, we evaluated the effectiveness of the proposed method by applying it to three real-world datasets from different domains.

\subsubsection{Schooling-Returns Data}\label{Real:Schooling-return}
We consider the application of our method to the study by \citet{card1993using}. This study investigates the impact of education levels on earnings using data from the Young Men Cohort of the National Longitudinal Survey.

\textbf{Data Description:} 
The dataset is a sample of 3010 men taken from the US National Longitudinal Survey of Young Men (NLSY). It includes variables such as 
$\mathit{LivedNearCollege}$, $\mathit{Schooling}$, $\mathit{Returns}$, and a set of \emph{\textbf{Covariates}} including 
\emph{\{Experience, Experience square, Black, Smsa, Smsa66, Region information (reg662-reg669), South}\}, 
among others. The hypothesized model of \cite{card1993using} is presented in Figure \ref{Fig:real-data-example4}. The hypothesized data generation mechanism is described as follows:

\revise{
\begin{equation}
    \begin{aligned}
        \mathit{Schooling} = & \, \alpha_0 + \alpha_1 \cdot \mathit{LivedNearCollege} + \boldsymbol{\alpha}^{\top} \cdot \text{\emph{\textbf{Covariates}}} + \delta, \\
        \mathit{Returns} = & \, \beta_0 + \beta_1 \cdot \mathit{Schooling} + \boldsymbol{\beta}^{\top} \cdot \text{\emph{\textbf{Covariates}}} + \epsilon,
    \end{aligned}
\end{equation}
}
where $\delta$ and $\epsilon$ are dependent.

\begin{figure}[htp!]
    \centering
    \begin{tikzpicture}[line width=0.75pt]
    \draw (0, 0) node(ER) [rectangle, draw, minimum width=0.8cm, minimum height=0.8cm, , rounded corners=5pt] {\parbox{2.0cm}{\centering \emph{\footnotesize{
    \textbf{Covariates}
    \\(Experience \\$\dots$\\South)}}}};

    \draw (4, 0) node(IA) [rectangle, draw=black, dash pattern=on 2.5pt off 2pt, fill=gray!30, minimum width=1cm, minimum height=1cm, rounded corners=5pt] {\parbox{2cm}{\centering \emph{Individual\\Ability}}};

    \draw (0, -2) node(Lnc) [rectangle, draw=OrangeRed, minimum width=1cm, minimum height=1cm, rounded corners=5pt] {\parbox{2cm}{\centering \emph{\color{darkgray}{Lived near \\ the college}}}};
    
    \draw (3, -2) node(Sch) [rectangle, draw, minimum width=1cm, minimum height=1cm, rounded corners=5pt] {\parbox{2cm}{\centering \emph{Schooling}}};

    \draw (6, -2) node(Re) [rectangle, draw, minimum width=1cm, minimum height=1cm, rounded corners=5pt] {\parbox{2cm}{\centering \emph{Returns}}};

    \draw[-arcsq] (ER) -- (Lnc);
    \draw[-arcsq] (ER) -- (Sch);
    \draw[-arcsq] (ER) -- (Re);

    \draw[-arcsq] (Lnc) -- (Sch);
    \draw[-arcsq] (Sch) -- (Re);

    \draw[-arcsq] (IA) -- (Sch);
    \draw[-arcsq] (IA) -- (Re);

    \end{tikzpicture}
    \caption{Graphical illustration of an IV model for estimating the causal effect of \emph{Schooling} on \emph{Returns} \citep{card1993using}.}
    \label{Fig:real-data-example4}
\end{figure}
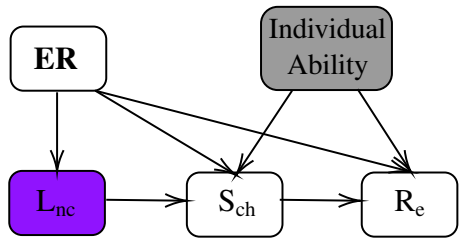

\textbf{Results:}
\cite{card1993using} demonstrated that \emph{LivedNearCollege} can serve as a valid IV for the causal relationship $\mathit{Schooling} \to \mathit{Returns}$, while controlling for the \emph{\textbf{Covariates}}. For consistency, we adopt the causal effect of \emph{Schooling} on \emph{Returns}, i.e., $\beta_1 = 0.1315$, as well as the coefficients $\boldsymbol{\beta}$ of the covariates from \cite{card1993using}, as the estimated parameter $\hat{\beta}$ under the AIT condition. We then obtain the residual $\widetilde{LivedNearCollege}$ by regressing \emph{LivedNearCollege} on the \emph{\textbf{Covariates}}. The $P$-value of the independence test between the auxiliary variable and the residual $\widetilde{LivedNearCollege}$ is 0.73, indicating that we cannot reject \emph{LivedNearCollege} as a valid IV. This result further supports the validity of using \emph{LivedNearCollege} as an IV, consistent with the findings in \cite{card1993using}.

\subsubsection{Colonial Origins data}\label{Real:colonial-or}
We apply our method to the study by \cite{acemoglu2001colonial}, which estimates the impact of colonial history on the economic development of different regions using the Colonial Origins of Comparative Development dataset.

\textbf{Data Description:} 
The dataset includes five key variables across 64 countries, after excluding samples with missing data. These variables are: $\mathit{Mortality}$, $\mathit{Euro1990}$, $\mathit{Latitude}$, $\mathit{Institutions}$, and $\mathit{Economic \ Development}$. The hypothesized model proposed by \cite{acemoglu2001colonial} is illustrated in Figure \ref{Fig:real-data-example2}, and the hypothesized data generation mechanism is described as follows:
\revise{
\begin{equation}
    \begin{aligned}
        Institutions &= \gamma + \begin{pmatrix}
        \gamma_1 & \gamma_2 & \gamma_3
        \end{pmatrix}
        \begin{pmatrix}
        \mathit{Mortality} \\
        \mathit{Latitude} \\
        \mathit{Euro1990}
    \end{pmatrix}
        + \delta,\\
        \mathit{Economic\ Development} &= \beta + 
        \begin{pmatrix}
        \beta_1 & \beta_2 
    \end{pmatrix}
    \begin{pmatrix}
        \mathit{Institutions} \\
        \mathit{Latitude} 
    \end{pmatrix}
    + \epsilon,
    \end{aligned}
\end{equation}
}
where $\delta$ and $\epsilon$ are dependent. Here, we adjust the parameter to $B = m/10$ to avoid overly small subsamples and improve the stability of estimation, due to the small sample size ($m=64$) in this dataset.

\begin{figure}[htp!]
    \centering
    \begin{tikzpicture}[line width=0.75pt]
    \draw (0, 0) node(Euro) [rectangle, draw, minimum width=1cm, minimum height=1cm, rounded corners=5pt] {\parbox{2cm}{\centering \emph{Euro1990}}};

    \draw (4, 0) node(Cd) [rectangle, draw=black, dash pattern=on 2.5pt off 2pt, fill=gray!30, minimum width=1cm, minimum height=1cm, rounded corners=5pt] {\parbox{2cm}{\centering \emph{Cultural\\difference}}};

    \draw (0, -2) node(Mor) [rectangle, draw=OrangeRed, minimum width=1cm, minimum height=1cm, rounded corners=5pt] {\parbox{2cm}{\centering \emph{\color{darkgray}{Mortality}}}};
    
    \draw (3, -2) node(Ins) [rectangle, draw, minimum width=1cm, minimum height=1cm, rounded corners=5pt] {\parbox{2cm}{\centering \emph{Institutions}}};

    \draw (6, -2) node(Ed) [rectangle, draw, minimum width=1cm, minimum height=1cm, rounded corners=5pt] {\parbox{2.2cm}{\centering \emph{Economic\\ Development}}};

    \draw (3, -4) node(Lat) [rectangle, draw, minimum width=1cm, minimum height=1cm, rounded corners=5pt] {\parbox{2cm}{\centering \emph{Latitude}}};

    \draw[-arcsq] (Euro) -- (Ins);
    \draw[-arcsq] (Ins) -- (Ed);
    \draw[-arcsq] (Mor) -- (Ins);

    \draw[-arcsq] (Lat) -- (Mor);
    \draw[-arcsq] (Lat) -- (Ins);
    \draw[-arcsq] (Lat) -- (Ed);
    \draw[-arcsq] (Lat) to[out=190, in=200] (Euro);

    \draw[-arcsq] (Cd) -- (Ins);
    \draw[-arcsq] (Cd) -- (Ed);
    \draw[-arcsq, dashed] (Cd) -- (Euro);

    \end{tikzpicture}
    \vspace{-3mm}
    \caption{Graphical illustration of an IV model for estimating the causal effect of \emph{Institutions} on \emph{Economic Development}  \citep{acemoglu2001colonial}.}
    \label{Fig:real-data-example2}
\end{figure}
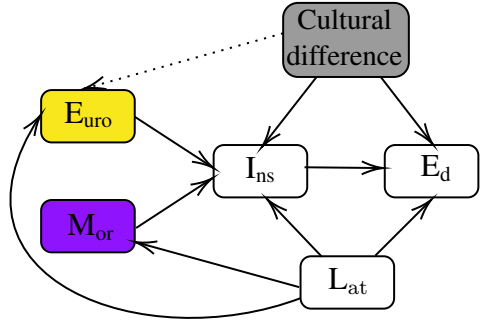

\textbf{Results:} 
\cite{acemoglu2001colonial} demonstrated that both $\mathit{Mortality}$ and $\mathit{Euro1990}$ can serve as valid IVs, conditional on $\mathit{Latitude}$, with respect to $\mathit{Institutions}$ and \textit{Economic Development}. To verify this, we test their validity using the AIT condition. For consistency, we adopt the causal effects of $\mathit{Institutions}$ on $\mathit{Economic\ Development}$ as reported by \cite{acemoglu2001colonial}, specifically $\beta_1 = 0.9458$ and $\beta_2 = -0.5971$, as the estimated parameters $\hat{\beta_i} (i=1,2)$ in the AIT condition. We next obtain the residuals $\widetilde{Mortality}$ and $\widetilde{Euro1990}$ by regressing $\mathit{Mortality}$ and $\mathit{Euro1990}$ on covariate $\mathit{Latitude}$, respectively. 
The validity test for $\mathit{Mortality}$ yields a $P$-value of $0.61$, whereas the test for $\mathit{Euro1990}$ yields a $P$-value of $0.25$. These results indicate that $\mathit{Euro1990}$ is more likely to be an invalid IV compared to $\mathit{Mortality}$, suggesting that the exogeneity of $\mathit{Euro1990}$ is weaker than that of $\mathit{Mortality}$. These findings are consistent with those of \cite{acemoglu2001colonial}, and we cannot reject the validity of $\mathit{Mortality}$ and $\mathit{Euro1990}$ as IVs, aligning with their conclusions. 

\subsubsection{Conflict and Time Preference Data}\label{Real:Conflict and Time Preference Data}
We consider the application of our method to the study by \citet{voors2012violent}. This study uses the Conflict and Time Preference Data to investigate the impact of violence on a person’s patience.

\textbf{Data Description:}
The dataset consists of 302 observations and fifteen variables, as described in \citet{voors2012violent}. Treatment variable, $\mathit{Violence}$, is measured by the percentage dead in attacks in the area the person lived, while the person’s $\mathit{Patience}$ (outcome variable) is assessed by a person’s discount rate for willingness to receive larger amounts of money in the future
compared to smaller amounts of money now. Other variables include $\mathit{Distance}$ (distance to Bujumbura) and $\mathit{Altitude}$, and \emph{\textbf{Covariates}} such as \{\emph{whether the
respondent is literate, the respondent’s age, the respondent’s sex, the total land holding per capita, land Gini coefficient, distance to market, conflict over land, ethnic homogeneity, 
socioeconomic homogeneity, population density, per capita total expenditure}\}. 
The \emph{\textbf{Covariates}} used in the study represent exogenous personal and geographical information variables. As discussed in \citet{voors2012violent}, violence may be targeted in a non-random way, potentially related to community patience, which makes violence endogenous. The hypothesized model from \citet{guo2016control} is illustrated in Figure \ref{Fig:real-data-conflict-Preference}, and the hypothesized generation mechanism is as follows: 

\revise{
\begin{equation}
    \begin{aligned}
        \mathit{Violence} &= 
        \alpha_0 +
        \begin{pmatrix}
             \alpha_1 & \alpha_2 & \alpha_3 & \alpha_4 & \alpha_5 & \boldsymbol{\alpha_6}^{\top}
        \end{pmatrix}
        \begin{pmatrix}
            \mathit{Distance} \\
            \mathit{Altitude} \\
            \mathit{Distance}^2 \\
            \mathit{Altitude}^2 \\
            Distance \cdot Altitude \\
            \text{\textbf{\emph{Covariates}}}
        \end{pmatrix}
        + \delta, \\
        \mathit{Patience} &= 
            \beta_0 + 
        \begin{pmatrix}
         \boldsymbol{\beta_1}^{\top} & \beta_2 & \beta_3
        \end{pmatrix}
        \begin{pmatrix}
            \text{\textbf{\emph{Covariates}}} \\
            \mathit{Violence} \\
            \mathit{Violence}^2
        \end{pmatrix}
        + \epsilon,
    \end{aligned}
\end{equation}
}
where $\delta$ and $\epsilon$ are dependent.

\begin{figure}[htp!]
    \centering
     \begin{tikzpicture}[line width=0.75pt]
    \draw (0, 0) node(dist) [rectangle, draw=Lavender, minimum width=1cm, minimum height=1cm, rounded corners=5pt] {\parbox{2cm}{\centering \emph{Distance to\\ Bujumbura}}};

    \draw (4, 0) node(U) [rectangle, draw=black, dash pattern=on 2.5pt off 2pt, fill=gray!30, minimum width=1cm, minimum height=1cm, rounded corners=5pt] {\parbox{4cm}{\centering \emph{Social and
\\Political confounder}}};

    \draw (0, -2) node(alti) [rectangle, draw=OrangeRed, minimum width=1cm, minimum height=1cm, rounded corners=5pt] {\parbox{2cm}{\centering \emph{\color{darkgray}{Altitude}}}};
    
    \draw (3, -2) node(Dper) [rectangle, draw, minimum width=1cm, minimum height=1cm, rounded corners=5pt] {\parbox{2cm}{\centering \emph{Violence}}};

    \draw (6, -2) node(Disc) [rectangle, draw, minimum width=1cm, minimum height=1cm, rounded corners=5pt] {\parbox{2.2cm}{\centering \emph{Patience}}};

    \draw (3, -4) node(Lat) [rectangle, draw, minimum width=0.8cm, minimum height=0.8cm, rounded corners=5pt] {\parbox{2.1cm}{\centering \footnotesize{\textbf{\emph{Covariates}}\\
    (\emph{Personal and Geographical information})}}};

    \draw[-arcsq] (dist) -- (Dper);
    \draw[-arcsq] (Dper) -- (Disc);
    \draw[-arcsq] (alti) -- (Dper);

    \draw[-arcsq] (Lat) -- (alti);
    \draw[-arcsq] (Lat) -- (Dper);
    \draw[-arcsq] (Lat) -- (Disc);
    \draw[-arcsq] (Lat) to[out=190, in=200] (dist);

    \draw[-arcsq] (U) -- (Dper);
    \draw[-arcsq] (U) -- (Disc);

    \end{tikzpicture}
    \vspace{-2mm}
    \caption{Graphical illustration of an IV model for estimating the causal effect of $\mathit{Violence}$  on a person's $\mathit{Patience}$ \citep{voors2012violent}.}
    \label{Fig:real-data-conflict-Preference}
\end{figure}
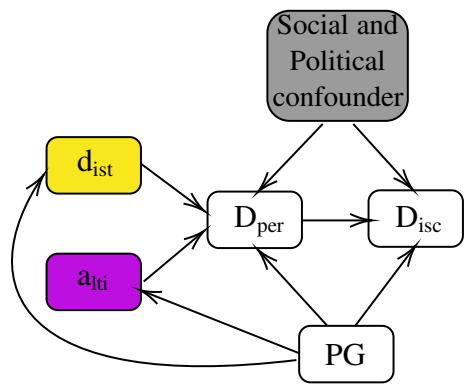

\vspace{2mm}

\textbf{Results:} 
\citet{voors2012violent} showed that both $\mathit{Distance}$ and $\mathit{Altitude}$ can serve as valid IVs, conditional on \emph{\textbf{Covariates}}, with respect to $\mathit{Violence}$ on $\mathit{Patience}$. To verify this, we test their validity using the AIT condition. For consistency, we adopt the causal effects of $\mathit{Violence}$ on $\mathit{Patience}$ as reported by \cite{guo2016control}, specifically $\beta_2 = 2.054$ and $\beta_3 = 0.049$, as the estimated parameters $\hat{\beta_i} (i=2,3)$ in the AIT condition. We obtain the residual $\widetilde{Distance}$ and $\widetilde{Altitude}$ by regressing $\mathit{Distance}$ and $\mathit{Altitude}$ on \emph{\textbf{Covariates}}, respectively. We first test the validity of $\mathit{Distance}$, which yields a $P$-value of $0.33$, suggesting that $\mathit{Distance}$ cannot be rejected as a valid instrumental variable. This is consistent with the findings of \cite{voors2012violent}. Similarly, testing $\mathit{Altitude}$ as a potential IV results in a $P$-value of $0.76$, indicating that $\mathit{Altitude}$ also cannot be rejected as a valid IV. These results align with the conclusions of \cite{voors2012violent}. 

\section{Conclusions}\label{Section-conclusions}
In this paper, we explored the testability of single IVs in the additive nonlinear, non-constant effects (ANINCE) model, where the treatment variable can be either discrete or continuous. To this end, we introduced a necessary condition, termed the AIT condition. \revise{We showed that, under the completeness condition, the AIT condition can detect whether a variable is a valid IV without requiring knowledge of whether any other variable serves as an instrument. 
Furthermore, we provided additional precision conditions for identifying all invalid IVs in linear and ANINCE models.} \tworevise{We then proposed the practical AIT condition test algorithm with covariates and finite data, and established its asymptotic level and power.} 
Experimental results using both simulation data and three real datasets have further validated the usefulness of our algorithm.
In the future, we plan to investigate whether the AIT condition could facilitate the testability implication of an invalid IV set.

\section*{Acknowledgements}
We are grateful to Kun Zhang for valuable discussions and helpful suggestions. 
We thank Patrick Burauel for kindly sharing his R implementation of IV-PIM. 
FX acknowledges support from the National Natural Science Foundation of China (Grant No.~62306019). 
XCG acknowledges the support of the Graduate Research Ability Enhancement Program Project Funding at Beijing Technology and Business University in 2024.

\appendix

\section{Proofs}\label{Appendix-proofs}
Before presenting the proofs, we introduce three important theorems since these are used to prove our results.
\revise{We begin by recalling the theorem on the properties of independent random variables from \cite{meester2008natural}. These results serve as the foundation for the proofs of Theorem \ref{Theorem-Necessary-Condition-IV} and Corollary \ref{Corollary-Necessary-Condition-IV-covariates}. 
\begin{Theorem}[Theorem 2.2.5 in \cite{meester2008natural}]\label{The_function_indep} 
Let $X_1, X_2, \ldots, X_m$ be \tworevise{mutually} independent random variables, and for $i = 1, \ldots, n$, $g_i$ be a function $g_i: \mathbb{R} \to \mathbb{R}$. Then the random variables $g_1(X_1), g_2(X_2), \ldots, g_m(X_m)$ are also \tworevise{mutually} independent. 
\end{Theorem}
\begin{Corollary}\label{The_sum_indep}
Let $X_1$, $X_2$ and $X_3$ be \tworevise{mutually} independent random variables; and let function $g: \mathbb{R}^2 \to \mathbb{R}$ and function $h: \mathbb{R} \to \mathbb{R}$. Then, $g(X_1, X_2)$ and $h(X_3)$ are independent random variables. 
\end{Corollary}}

\revise{
\begin{proof}
Let \( g: \mathbb{R}^2 \rightarrow \mathbb{R} \) and \( h: \mathbb{R} \rightarrow \mathbb{R} \). We aim to show that the random variables \( g(X_1, X_2) \) and \( h(X_3) \) are independent. We here provide the proof for the case where the random variables are discrete. The argument for continuous random variables is similar and can be obtained by replacing the summation with integration.

We compute:
\begin{align*}
& P(g(X_1, X_2) = a, h(X_3) = b)\\
&= \sum_{\substack{x_1, x_2, x_3:\\ g(x_1, x_2) = a,\, h(x_3) = b}} 
P(X_1 = x_1, X_2 = x_2, X_3 = x_3) \\
&= \sum_{\substack{x_1, x_2, x_3:\\ g(x_1, x_2) = a,\, h(x_3) = b}} 
P(X_1 = x_1) P(X_2 = x_2) P(X_3 = x_3) \\
&= \left( \sum_{\substack{x_1, x_2:\\ g(x_1, x_2) = a}} 
P(X_1 = x_1) P(X_2 = x_2) \right)
\left( \sum_{\substack{x_3:\\ h(x_3) = b}} 
P(X_3 = x_3) \right) \\
&= P(g(X_1, X_2) = a) \cdot P(h(X_3) = b)
\end{align*}

Therefore, \( g(X_1, X_2) \) and \( h(X_3) \) are independent.
\end{proof}
}

We then quote the Darmois–Skitovitch theorem that characterizes the independence of two linear statistics \citep{darmois1953analyse,skitovitch1953property}. This theorem provides the foundation for proving Propositions \ref{Proposition-linear-exogeneity} $\sim$ \ref{proposition-linear-model-exclusion} and Theorem \ref{Theorem-Necessary-Sufficient-Condition-IV-Linear-Models}. 

\begin{Theorem}[{\textbf{Darmois--Skitovitch Theorem}}]\label{Th-Darmois--Skitovitch}
Define two random variables ${V_1}$ and ${V_2}$ as linear combinations of independent random variables $\varepsilon_1,...,\varepsilon_p$:
\begin{align}
{V_1} = \sum\limits_{i = 1}^p {{\alpha _i}} {\varepsilon_i}, \quad \quad {V_2} = \sum\limits_{i = 1}^p {{\beta _i}} {\varepsilon_i},
\end{align}
where the $\alpha_{i}, \beta_{i}$ are constant coefficients. If ${V_1}$ and ${V_2}$ are independent, then the random variables ${\varepsilon_j}$ for which ${\alpha _j}{\beta _j} \ne 0$ are Gaussian. 
\end{Theorem}
The above theorem states that if there exists a non-Gaussian ${\varepsilon_j}$ for which ${\alpha _j}{\beta _j} \ne 0$, ${V_1}$ and ${V_2}$ are dependent. 

Next, we introduce a local geometric information theorem that characterizes the independence of two nonlinear statistics \citep{lin1997factorizing}. This result provides the foundation for proving Propositions \ref{Proposition-identifiable-IV-ANINCE} $\sim$ \ref{Pro-violate-algebraic}, and Theorem \ref{Theorem-Necessary-Sufficient-Condition-IV-ANINCE-Models}. 

\begin{Theorem}[Proposition in \cite{lin1997factorizing}]
\label{Theorem-lin}
    The Hessian $H_f$ of function $f$ is block diagonal everywhere, $\partial_i \partial_j f \big|_{\vec{s_0}} = 0$ for all points $\vec{s_0}$ and all $i \le  k$, $j > k$, if and only if f is separable into a sum $f(s_1,...,S_m) = g(s_1,...,s_k) + h(s_{k+1},...,S_m)$ for some functions $g$ and $h$. 
\end{Theorem}

Theorem \ref{Theorem-lin} states that a function $f$ is separable if and only if its mixed second-order partial derivative is zero.

\subsection{Proof of Proposition \ref{Proposition-linear-exogeneity}: Testability of Exogeneity in Linear Models}\label{Proof-proposition-exogeneity-linear-model} 

\begin{proof} 
    Under the assumption of a linear model, Equation \eqref{Eq-Main-Model} can be expressed as follows: 
    \begin{equation}\label{Appendix:Eq-proposition-linear-model-exogeneity}
    \begin{aligned}
    \mathbf{U} &= \boldsymbol{\varepsilon_U},  \quad 
    & \quad  
    Z &= \boldsymbol{\gamma}^{T}  \boldsymbol{U} + \varepsilon_{Z},\\ 
    X & = \tau Z + \boldsymbol{\rho}^{T} \boldsymbol{U} + \varepsilon_{X}, &\quad Y &= \beta X+ \nu Z + \boldsymbol{\kappa}^{T} \boldsymbol{U} + \varepsilon_{Y}.
    \end{aligned}
    \end{equation}
    \revise{Let $h(X) = \hat{\beta} X = \frac{\operatorname{Cov}(Y, Z)}{\operatorname{Cov}(X, Z)} X$. According to the AIT condition, in the linear model, the function $h(X)$ satisfies the conditional moment restriction, \ie, $\mathbb{E}[Y - h(X) \mid Z] = 0$. By Equation \eqref{Appendix:Eq-proposition-linear-model-exogeneity}, $\hat{\beta}$ can be expressed as follows:} 
    \begin{equation}
    \hat{\beta} = \beta + 
    \underbrace{
        \frac{ (\nu \boldsymbol{\gamma}^{T} + \boldsymbol{\kappa}^{T}) \mathrm{Cov}(\mathbf{U})\boldsymbol{\gamma} + \nu \, \mathrm{Var}(\varepsilon_Z)
        }{ (\tau \boldsymbol{\gamma}^{T} + \boldsymbol{\rho}^{T}) \mathrm{Cov}(\mathbf{U})\boldsymbol{\gamma} + \tau \, \mathrm{Var}(\varepsilon_Z) }}_{\beta_{\text{bias}}}.
\end{equation}
Hence, we have 
\begin{equation}
    \begin{aligned}
    \mathcal{A}_{X \to Y \mid \mid Z} 
    &= Y - \hat{\beta}X \\
    &= \nu Z + \boldsymbol{\kappa}^{T} \mathbf{U} + \varepsilon_Y - \beta_{\text{bias}} X \\
    &= \left[ \nu \boldsymbol{\gamma}^{T} + \boldsymbol{\kappa}^{T} - \left( \nu \tau + \boldsymbol{\rho}^{T} \right) \beta_{\text{bias}} \right] \boldsymbol{\varepsilon_U} + \left( \nu - \tau \beta_{\text{bias}} \right) \varepsilon_Z - \beta_{\text{bias}} \varepsilon_X + \varepsilon_Y.
    \end{aligned}
\end{equation}
Since $Z$ is an invalid IV violating exogeneity, without loss of generality, we assume that $U_i$ causes the candidate IV $Z$, i.e., ${\gamma_{i}} \neq 0$, then $\beta_{bias} \neq 0$. This will imply that $\left[ \nu \gamma_i + \kappa_i - \left( \nu \tau + \rho_i \right) \beta_{\text{bias}} \right] \neq 0$ for $\varepsilon_{U_i}$, and $\left(\nu - \tau\beta _{bias}\right) \neq 0$ for $\varepsilon_Z$. 
Furthermore, because of Assumption \ref{Ass-non-Gaussianity-exogeneity}, i.e., (i) there exists at least one variable $U_i \in \mathbf{U}$ whose noise term follows a non-Gaussian distribution and cause $Z$ or, (ii) the noise term of $Z$ follows a non-Gaussian distribution, then at least one of the non-Gaussian noise terms $\varepsilon_{U_i}$ or $\varepsilon_{Z}$, is common between $\mathcal{A}_{X \to Y \mid \mid Z}$ and $Z$. Due to the Darmois–Skitovitch Theorem, we have $\mathcal{A}_{X \to Y||Z}$ is dependent on $Z$. That is to say, $\{X, Y||Z\}$ \emph{violates} the AIT condition. 
\end{proof}

\subsection{Proof of Proposition \ref{proposition-linear-model-exclusion}: Non-Testability of Exclusion Restriction in Linear Models}
\begin{proof}\label{Proof-proposition-exclusion-linear-model}
Because candidate IV $Z$ satisfies exogeneity condition, the model of Equation \eqref{Eq-Main-Model} can be written as: 
\begin{equation}\label{Appendix:Eq-proposition-linear-model-exculsion}
    \begin{aligned}
    \mathbf{U} &= \boldsymbol{\varepsilon_U}, \quad
    & \quad  
    Z &= \varepsilon_{Z},\\ 
    X & = \tau Z + \boldsymbol{\rho}^{T} \boldsymbol{U} + \varepsilon_{X}, &\quad 
    Y &= \beta X + \nu Z + \boldsymbol{\kappa}^{T} \boldsymbol{U} + \varepsilon_{Y}.
    \end{aligned}
    \end{equation}
    \revise{Let $h(X) = \hat{\beta} X = \frac{\operatorname{Cov}(Y, Z)}{\operatorname{Cov}(X, Z)} X$. According to the AIT condition, in the linear model, the function $h(X)$ satisfies the conditional moment restriction, \ie, $\mathbb{E}[Y - h(X) \mid Z] = 0$. By Equation \eqref{Appendix:Eq-proposition-linear-model-exculsion}, $\hat{\beta}$ can be expressed as follows:} 
    $$\hat{\beta} = \frac{\operatorname{Cov}(Y, Z)}{\operatorname{Cov}(X, Z)}=\beta + \frac{\nu}{\tau}=\beta + \beta_{bias}. $$ 
    Hence, we obtain 
\begin{equation}
    \begin{aligned}
    \mathcal{A}_{X \to Y||Z}
    &= Y - \hat{\beta}X\\
    &= \nu Z + \boldsymbol{\kappa}^{T} \boldsymbol{U} + \varepsilon_{Y} - \beta_{bias} X\\
    & = \left (\boldsymbol{\kappa}^{T}- \frac{\boldsymbol{\rho}^{T} \nu}{\tau} \right)\boldsymbol{\varepsilon_U} - \frac{\nu}{\tau}\varepsilon _X +\varepsilon _Y.
    \end{aligned} 
    \end{equation}
When $\nu = 0$, Equation \eqref{Appendix:Eq-proposition-linear-model-exculsion} implies that $Z$ is a valid IV, we have $\mathcal{A}_{X \to Y||Z} = \boldsymbol{\kappa}^{T} \boldsymbol{\varepsilon_U} +\varepsilon _Y$. Since $Z$ only contains the noise term $\varepsilon_Z$, in both cases—whether $Z$ is an invalid IV that violates solely the exclusion restriction condition or a valid IV—there are no shared noise terms between $\mathcal{A}_{X \to Y||Z}$ and $Z$. By the Darmois–Skitovitch Theorem, we have $\mathcal{A}_{X \to Y||Z}$ is independent of $Z$. Thus, $\{X, Y||Z\}$ always \emph{satisfies} the AIT condition. 
\end{proof}

\subsection{Proof of Theorem \ref{Theorem-Necessary-Sufficient-Condition-IV-Linear-Models}: Necessary and Sufficient Conditions in Linear Models}
\begin{proof} \label{proof:Theorem-Necessary-Sufficient-Condition-IV-Linear-Models} 
Below, we prove the necessary and sufficient condition for identifying invalid IV in the linear model. 

($\boldsymbol{\Rightarrow}$): 
According to Theorem \ref{Theorem-Necessary-Condition-IV}, we know that if $Z$ is a valid IV relative to $X \to Y$, then $\{X, Y||Z\}$ always satisfies the AIT condition. This indicates that if $\{X, Y||Z\}$ violates the AIT condition, then $Z$ is an invalid IV. 
Below, we processed with a proof by contradiction to prove that candidate IV $Z$ violates the \emph{exogeneity} condition. If $Z$ solely violates the \emph{exclusion restriction} condition, then by Proposition \ref{proposition-linear-model-exclusion}, $\{X, Y||Z\}$ always satisfies the AIT condition. This contradicts the assumption that $\{X, Y||Z\}$ violates the AIT condition. As a result, candidate IV $Z$ violates the \emph{exogeneity} condition.

($\boldsymbol{\Leftarrow}$): 
To establish that Assumption \ref{Ass-non-Gaussianity-exogeneity} holds and that the candidate IV $Z$ is invalid due to a violation of the exogeneity condition, we need to show that \(\{X, Y||Z\}\) consequently violates the AIT condition. This conclusion is equivalent to Proposition \ref{Proposition-linear-exogeneity}, thereby proving the theorem.
\end{proof}

\subsection{Proof of Proposition \ref{Proposition-identifiable-IV-ANINCE}: Testability of IV in ANINCE Models}

\begin{proof}\label{Proof-Proposition-identifiable-IV-ANINCE}
If $Z$ violates the IV conditions, the generating mechanism can be described as follows:
\begin{equation}\label{Eq:ANINCE-exogeneity}
\begin{aligned}
    \mathbf{U} = \boldsymbol{\varepsilon_{U}}, & \quad &
    Z &= \varphi_{Z}(\mathbf{U}) + \varepsilon_{Z},\\
    X =  g({Z}) + \varphi_{X}(\mathbf{U}) + \varepsilon_{X}, & \quad & Y &= f(X,Z) + \varphi_Y(\mathbf{U}) + \varepsilon_{Y}.
    \end{aligned}
\end{equation}

\revise{Without loss of generality, let $h(\cdot)$ be a function satisfying $\mathbb{E}[ Y - h(X)|Z] = {0}$ and $h(\cdot) \neq {0}$. According to the definition of the AIT condition, we have
\begin{equation}\label{Eq:ANINCE-exogeneity-A}
    \begin{aligned}
    \mathcal{A}_{X \to Y||Z}
    &= Y - h(X) \\
    &= \underbrace{f(X,Z) - h(X) }_{\widetilde{f}_{bias}(X, Z)} + \varphi_Y(\mathbf{U}) + \varepsilon_{Y},
    \end{aligned} 
    \end{equation}
where $\widetilde{f}_{bias}(X, Z) = f(X,Z) - h(X)$. 

Combining Equations \eqref{Eq:ANINCE-exogeneity} and \eqref{Eq:ANINCE-exogeneity-A}, we have 
\begin{equation}\label{Eq-transform}
    \begin{aligned}
        \mathcal{A}_{X \to Y||Z} &= \widetilde{f}_{bias}(X, Z) + \varphi_Y(\mathbf{U}) + \varepsilon_Y,\\
        Z &= \varphi_{Z}(\mathbf{U}) + \varepsilon_{Z}.
    \end{aligned}
\end{equation}
Below, we prove this proposition using the linear separability of the logarithm of the joint density of independent variables, which states the fact that for a set of independent random variables whose joint density is twice differentiable, the Hessian of the logarithm of their density is diagonal everywhere from \cite{lin1997factorizing} (See Theorem \ref{Theorem-lin} in Appendix \ref{Appendix-proofs} for further details). 
\tworevise{
According to Assumption \ref{Ass-algebraic-condition}, we know that the second-order partial derivative $\frac{\partial^2 \operatorname{log}p(\mathcal{A}_{X \to Y||Z},Z)}{\partial \mathcal{A}_{X \to Y||Z}\partial Z} \neq 0$. Furthermore, according to the local geometric information theorem (Theorem \ref{Theorem-lin}), we have $\mathcal{A}_{X \to Y||Z} \nCI Z$. This implies that $\{X, Y||Z\}$ violates the AIT condition.} }
\end{proof}

\subsection{Proof of Proposition \ref{Pro-violate-algebraic}}\label{proof-Pro-violate-algebraic}
Below, we provide separate proofs for each of the two cases in Proposition \ref{Pro-violate-algebraic}. 

\subsubsection{Proof of Proposition \ref{Pro-violate-algebraic} (i)}
\begin{proof} 
Because candidate IV $Z$ satisfies exogeneity condition, the model of Equation \eqref{Eq-Main-Model} can be written as: 
\begin{equation}\label{Eq-CAM-invalid-IV-model}
    \begin{aligned} 
    \mathbf{U} &= \varepsilon_{\mathbf{U}}, \quad 
    & \quad
    Z &= \varepsilon_{Z},\\
    X &=  g_X(Z) + \varphi_{X}(\mathbf{U}) + \varepsilon_{X}, \quad 
    & \quad
    Y &= f(X) + g_Y(Z) + \varphi_Y(\mathbf{U}) + \varepsilon_{Y}. 
    \end{aligned}
    \end{equation}
If the direct causal effect of $Z \to Y$ is a linear function of the direct causal effect of $Z \to X$, i.e., $g_Y(Z) = a \cdot g_X(Z) + b$, then it is possible to construct a valid IV  model that shares the same distribution as the above invalid IV model \eqref{Eq-CAM-invalid-IV-model}. 
The model for this valid IV, which has an identical distribution, is as follows: 
\begin{equation}\label{Eq-CAM-valid-IV-model}
\begin{aligned} 
    \mathbf{U}^{\prime}  &= \mathbf{U}, \quad 
    & \quad
    Z^{\prime} &= Z, \\
    X^{\prime} &= X, \quad 
    & \quad
    Y^{\prime} &= f^{\prime}(X') + \varphi_Y(\mathbf{U}) + \varepsilon_{Y}=Y, 
\end{aligned}
\end{equation}
where $f^{\prime}(X') = f(X) + g_Y(Z)$. According to Equation \eqref{Eq-CAM-valid-IV-model}, we know that $Z^{\prime}$ is a valid IV relative to $X^{\prime} \to Y^{\prime}$. 

\revise{
Based on Theorem \ref{Theorem-Necessary-Condition-IV}, the auxiliary variable $\mathcal{A}_{X \to Y||Z}$ is independent of IV Z. As a result, the joint probability density of $(\mathcal{A}_{X \to Y||Z},Z)$ factorizes as the product of the marginal densities: $$p(\mathcal{A}_{X \to Y||Z},Z)= p(\mathcal{A}_{X \to Y||Z}) \cdot p(Z).$$ According to the local geometric information theorem (Theorem \ref{Theorem-lin}), taking the logarithm of both sides and computing the second-order partial derivative yields: 
$$\frac{\partial^2 \operatorname{log}p(\mathcal{A}_{X \to Y||Z}, Z)}{\partial \mathcal{A}_{X \to Y||Z}\partial Z} = \frac{\partial^2 [\operatorname{log}p(\mathcal{A}_{X \to Y||Z}) + \operatorname{log}p(Z)]}{\partial \mathcal{A}_{X \to Y||Z}\partial Z} = 0.$$
This result indicates that Assumption \ref{Ass-algebraic-condition} does not hold. }
\end{proof}

\subsubsection{Proof of Proposition \ref{Pro-violate-algebraic} (ii)}

\tworevise{
\begin{proof}
Due to the assumption of linearity in the model, the model in Equation \eqref{Eq-Main-Model} can be written as:
    \begin{equation}\label{Appendix:Eq-proposition-linear-model}
    \begin{aligned}
    \mathbf{U} &= \boldsymbol{\varepsilon_U},  \quad 
    & \quad  
    Z &= \boldsymbol{\gamma^{T}} \boldsymbol{U} + \varepsilon_{Z},\\ 
    X & = \tau Z + \boldsymbol{\rho}^{T} \boldsymbol{U} + \varepsilon_{X}, &\quad Y &= \beta X+ \nu Z + \boldsymbol{\kappa}^{T} \boldsymbol{U} + \varepsilon_{Y},
    \end{aligned}
    \end{equation} 
where all noise terms are Gaussian.

Let $h(X) = \hat{\beta} X = \frac{\operatorname{Cov}(Y, Z)}{\operatorname{Cov}(X, Z)} X$. According to the AIT condition, in the linear model, the function $h(X)$ satisfies the conditional moment restriction, \ie, $\mathbb{E}[Y - h(X) \mid Z] = 0$. 
Since all noise terms are Gaussian and the model is linear, the vector $(X,Y,Z)$ is jointly Gaussian. 
Because $\mathcal{A}_{X\to Y\|Z}=Y-\hat\beta X$ is a linear combination of $(X,Y)$, the pair $(\mathcal{A}_{X\to Y\|Z},Z)$ is also jointly Gaussian.

Hence, we have
\begin{equation}
\begin{aligned}
\mathrm{Cov}(\mathcal{A}_{X\to Y\|Z},Z)
=\mathrm{Cov}(Y-\hat\beta X,Z)
=\mathrm{Cov}(Y,Z)-\hat\beta\,\mathrm{Cov}(X,Z)=0.
\end{aligned}
\end{equation}
Therefore, $\mathcal{A}_{X\to Y\|Z}$ and $Z$ are uncorrelated. For jointly Gaussian variables, uncorrelatedness implies independence, and thus
\begin{equation}
\begin{aligned}
p(\mathcal{A}_{X\to Y\|Z},Z)=p(\mathcal{A}_{X\to Y\|Z})\,p(Z)
\;\Longrightarrow\;
\log p(\mathcal{A}_{X\to Y\|Z},Z)=\log p(\mathcal{A}_{X\to Y\|Z})+\log p(Z).
\end{aligned}
\end{equation}
It follows immediately that
\begin{equation}
\begin{aligned}
\frac{\partial^2}{\partial \mathcal{A}_{X\to Y\|Z}\,\partial Z}\log p(\mathcal{A}_{X\to Y\|Z},Z)\equiv 0,
\end{aligned}
\end{equation}
which implies that Assumption~\ref{Ass-algebraic-condition} does not hold.
\end{proof}
}

\subsection{Proof of Theorem \ref{Theorem-Necessary-Sufficient-Condition-IV-ANINCE-Models}: Necessary and Sufficient Conditions for IV in ANINCE Models} 

\begin{proof}\label{Proof-Theorem-Necessary-Sufficient-Condition-IV-ANINCE-Models}  
Below, we prove the necessary and sufficient conditions for identifying invalid IV in
the ANINCE model. 
\begin{itemize}
    \item ($\boldsymbol{\Rightarrow}$): According to Theorem \ref{Theorem-Necessary-Condition-IV}, we know that if $Z$ is a valid IV relative to $X \to Y$, then $\{X, Y||Z\}$ always satisfies the AIT condition. 
    \item ($\boldsymbol{\Leftarrow}$): To establish that Assumptions \ref{Ass-completeness} and \ref{Ass-algebraic-condition} hold and that the candidate IV $Z$ is invalid, we need to show that $\{X, Y||Z\}$ consequently violates the AIT condition. This conclusion is equivalent to Proposition \ref{Proposition-identifiable-IV-ANINCE}, thereby proving the theorem.
\end{itemize}
\end{proof}

\subsection{Proof of Corollary \ref{Corollary-Necessary-Condition-IV-covariates}: Necessary Condition for IV with Covariates}\label{Proof:Corollary-Necessary-Condition-IV-covariates} 
\begin{proof}
Below, we apply the same proof technique used in Theorem \ref{Theorem-Necessary-Condition-IV} to demonstrate Corollary \ref{Corollary-Necessary-Condition-IV-covariates}. 
\revise{Suppose $Z$ is a valid IV relative to $X \to Y$ given $\mathbf{W}$, the generating mechanism can be expressed as follows: 
\begin{equation}
\begin{aligned}
    \mathbf{U} &= \boldsymbol{\varepsilon_{U}}, \quad 
    \mathbf{W} = t_W(\mathbf{PA}_{\mathbf{W}}) +  \boldsymbol{\varepsilon_W}, \quad 
    Z = t_Z(\mathbf{W}) + \varepsilon_Z,\\
    X &= g(\mathbf{W}, Z) + \varphi_{X}(\mathbf{U}) + \varepsilon_{X},  \quad  Y = f(X, \mathbf{W}) + \varphi_Y(\mathbf{U}) + \varepsilon_{Y},
    \end{aligned}
\end{equation}
where functions $g(\cdot)$, $f(\cdot)$ and $\varphi_{*}(\cdot)$ are smooth functions, $\mathbf{PA}_{\mathbf{W}}$ denotes the set of parent variables for each variable in $\mathbf{W}$, and $\mathbf{PA}_{\mathbf{W}} \subseteq \mathbf{W}$. 
According to \citet{newey2003instrumental}, \citet{singh2019kernel}, \citet{bennett2019deep}, under the completeness condition and given that $Z$ is a valid IV under the covarites $\mathbf{W}$, the function $h(\cdot)$ that satisfies the conditional moment restriction $\mathbb{E}[Y-h(X)|Z, \mathbf{W}]=0$ is uniquely identified, and coincides with the true causal effect function $f(\cdot)$ of $(X, \mathbf{W})$ on $Y$, that is, $h(\cdot) = f(\cdot)$.}  

Thus, we have 
\begin{align}\label{Eq:Theorem-Necessary-Condition-IV-A-covariates}
    \mathcal{A}_{X \to Y||(Z, \mathbf{W})} = Y - h(X, \mathbf{W}) = Y - f(X, \mathbf{W}) = \varphi_{Y}(\mathbf{U}) + \varepsilon_{Y}.
\end{align}
Let $\mathcal{Z}$ denote the residual from the regression $Z$ on $\mathbf{W}$. Thus, $\mathcal{Z} = \varepsilon_Z$. 

By Theorem 2.2.5 and its extension in \cite{meester2008natural}, if random variables are mutually independent, then any measurable functions applied to disjoint subsets of them yield independent random variables (see Theorem \ref{The_function_indep} and Corollary~\ref{The_sum_indep} in Appendix~\ref{Appendix-proofs} for further details). Based on this result, we next show that the auxiliary variable $\mathcal{A}_{X \to Y||(Z, \mathbf{W})}$ and $\mathcal{Z}$ are statistically independent. Specifically, since the noise terms $\varepsilon_Z$, $\varepsilon_Y$, and $\boldsymbol{\varepsilon_U}$ are mutually independent, we can obtain that $\varepsilon_Z$ is also independent of $\varphi_Y(\boldsymbol{\varepsilon_U}) + \varepsilon_Y$. Furthermore, combining the equations $\boldsymbol{U}=\boldsymbol{\varepsilon_U}$ and $\mathcal{Z} = \varepsilon_Z$ , we conclude that $\mathcal{Z}$ is independent of $\varphi_Y(\boldsymbol{U}) + \varepsilon_Y$. Hence, the auxiliary variable $\mathcal{A}_{X \to Y||(Z, \mathbf{W})}$ and $\mathcal{Z}$ are statistically independent, \ie, $\mathcal{A}_{X \to Y||(Z, \mathbf{W})} \CI \mathcal{Z} $. 
Consequently, $\{X, Y||(Z, \mathbf{W})\}$ satisfies the AIT condition. 
\end{proof}

\subsection{Proof of Theorem \ref{The:HSIC}: Asymptotic Level and Power of the AIT Test}\label{App:proof of HSIC}
\begin{proof}
Let $h(\cdot)$ denote the oracle function and $\hat{h}(\cdot)$ its estimator trained on the dataset $\mathcal{D}_1 \subset \mathcal{D}$. 
We establish our theoretical results using an independent dataset $\mathcal{D}_2 = \{(X_i, Y_i, Z_i)\}_{i=1}^m \subset \mathcal{D}$.
For notational simplicity, we below use $(X, Y, Z)$ to denote a generic observation drawn from $\mathcal{D}_2$. Recall that $\mathcal{A} = Y - h(X)$ and $\hat{\mathcal{A}} = Y - \hat{h}(X)$. 
In the absence of covariates, we simply have $\mathcal{Z} = Z$. To establish the asymptotic level and power of the AIT test, we analyze the difference between 
$\widehat{\mathrm{HSIC}}(\hat{\mathcal{A}}, \mathcal{Z})$ 
and its oracle counterpart 
$\widehat{\mathrm{HSIC}}(\mathcal{A}, \mathcal{Z})$, where $\widehat{\mathrm{HSIC}}(\cdot,\cdot)$ denotes the block-based estimator of \citet{zhang2018large}. Specifically, 
\begin{equation}
    \begin{aligned}\label{Eq: test statistic}
        &\Big|\sqrt{mB} \widehat{\mathrm{HSIC}}(\hat{\mathcal{A}}, \mathcal{Z}) - \sqrt{mB}\widehat{\mathrm{HSIC}}(\mathcal{A}, \mathcal{Z}) \Big| \\
        =&\Big|\sqrt{mB} \cdot \frac{B}{m}\sum_{b=1}^{m/B} \hat{\eta}_{b}(\hat{\mathcal{A}}) - \sqrt{mB} \cdot \frac{B}{m}\sum_{b=1}^{m/B} \hat{\eta}_{b}(\mathcal{A}) \Big| \\
        =& \sqrt{mB} \cdot \frac{B}{m} \left|\sum_{b=1}^{m/B}\left[\hat{\eta}_{b}(\hat{\mathcal{A}}) - \hat{\eta}_{b}(\mathcal{A}) \right] \right| \quad \text{(by the triangle inequality)} \\
        \le & \sqrt{mB} \cdot \frac{B}{m} \sum_{b=1}^{m/B} \left|\left[\hat{\eta}_{b}(\hat{\mathcal{A}}) - \hat{\eta}_{b}(\mathcal{A}) \right] \right| \\
        =& \sqrt{mB} \cdot \frac{B}{m} \sum_{b=1}^{m/B} \cdot \left.\left|\left[\underbrace{\frac{1}{B(B-3)}\left[\operatorname{tr}\left(\tilde{\hat{K}}^{(b)} \tilde{L}^{(b)}\right)- \operatorname{tr}\left(\tilde{K}^{(b)} \tilde{L}^{(b)}\right)\right]}_{\text{Part I}} \right.\right.\right.\\
        +& \left. \underbrace{\frac{1}{B(B-3)}\left[\frac{\mathbbm{1}^{T} \tilde{\hat{K}}^{(b)} \mathbbm{1}\mathbbm{1}^{T} \tilde{L}^{(b)} \mathbbm{1}}{(B-1)(B-2)} -\frac{\mathbbm{1}^{T} \tilde{K}^{(b)} \mathbbm{1}\mathbbm{1}^{T} \tilde{L}^{(b)} \mathbbm{1}}{(B-1)(B-2)} \right]}_{\text{Part II}} \right. \\ 
        -& \left. \left. \underbrace{\frac{1}{B(B-3)}\left[\frac{2}{B-2} \mathbbm{1}^{T} \tilde{\hat{K}}^{(b)} \tilde{L}^{(b)} \mathbbm{1} - \frac{2}{B-2} \mathbbm{1}^{T} \tilde{K}^{(b)} \tilde{L}^{(b)} \mathbbm{1}\right]}_{\text{Part III}} \right] \right|,
    \end{aligned}
\end{equation}
where $\tilde{\hat{K}}^{(b)}$ denotes the kernel matrix for the estimated auxiliary variables $\hat{\mathcal{A}}$ in the $b$-th block. 

Below, we prove the asymptotic level and power of the AIT condition test, respectively. 

\noindent\textbf{1. Proof of the asymptotic level of the AIT condition. } 

To establish the asymptotic level of the AIT test under the null hypothesis, we proceed in three steps. First, we derive a bound on the estimation error within a single block. 
Second, we extend this bound to the full HSIC estimator. Finally, we study the asymptotic behavior of the resulting test statistic.

\begin{itemize}
    \item \textbf{Step 1: Single-block Error Bound for $|\hat{\eta}_b(\hat{\mathcal{A}})-\hat{\eta}_b(\mathcal{A})|$.} For a fixed block $b$ (omitting the superscript $(b)$ for brevity), let $l_{ij} = l(\mathcal{Z}_i, \mathcal{Z}_j)$. We decompose the difference into the three components appearing in \eqref{Eq: test statistic} (Parts I–III) and bound them separately. 
    \begin{enumerate}
    \item[(1)] \textbf{Part I in Equation \eqref{Eq: test statistic}:}
    \begin{equation}
        \begin{aligned}\nonumber
            &\frac{1}{B(B-3)}\operatorname{tr}\left(\widetilde{\hat{K}} \tilde{L}\right) - \frac{1}{B(B-3)}\operatorname{tr}\left(\tilde{K} \tilde{L}\right) \\
            =& \frac{1}{B(B-3)}\left[\sum_{i,j: i \neq j}\left(k\left(\hat{\mathcal{A}}_{i}, \hat{\mathcal{A}}_{j}\right)-k\left(\mathcal{A}_{i}, \mathcal{A}_{j}\right)\right) l_{i j}\right]
        \end{aligned}
    \end{equation} 
    Taking absolute value of the above expression, and by the boundedness and Lipschitz continuity (Assumption \ref{Ass: Bounded Lipschitz Kernels}), we have 
    \begin{equation}\label{Eq: first term}
        \begin{aligned}
                &\left|\frac{1}{B(B-3)}\left[\sum_{i,j:i \neq j}\left(k\left(\hat{\mathcal{A}}_{i}, \hat{\mathcal{A}}_{j}\right)-k\left(\mathcal{A}_{i}, \mathcal{A}_{j}\right)\right) l_{i j}\right]\right|\\
                \le& \frac{M}{B(B-3)}\sum_{i,j: i \neq j}\left|\left(k\left(\hat{\mathcal{A}}_{i}, \hat{\mathcal{A}}_{j}\right)-k\left(\mathcal{A}_{i}, \mathcal{A}_{j}\right)\right) \right| \qquad \text{(by Assumption \ref{Ass: Bounded Lipschitz Kernels} (1))} \\
                \le& \frac{MQ}{B(B-3)}\sum_{i,j: i \neq j}\left(\left\|\hat{\mathcal{A}}_{i}-\mathcal{A}_{i}\right\|+\left\|\hat{\mathcal{A}}_{j}-\mathcal{A}_{j}\right\|\right) \qquad \text{(by Assumption \ref{Ass: Bounded Lipschitz Kernels} (2))}\\
                \le& \frac{MQ}{B(B-3)}\left[(B-1) \cdot \sum_{i} \left\|\hat{\mathcal{A}}_{i}-\mathcal{A}_{i}\right\|+(B-1) \cdot \sum_{j} \left\|\hat{\mathcal{A}}_{j}-\mathcal{A}_{j}\right\|\right]\\
                \le& \frac{2MQ(B-1)}{B(B-3)} \sum_{i} \left\|\hat{\mathcal{A}}_{i}-\mathcal{A}_{i}\right\|.
        \end{aligned}
    \end{equation}
    \item[(2)] \textbf{Part II in Equation \eqref{Eq: test statistic}:} 
    \begin{equation}\nonumber
        \begin{aligned}
         &\frac{1}{B(B-3)}\frac{\mathbbm{1}^{T} \widetilde{\hat{K}} \mathbbm{1}\mathbbm{1}^{T} \tilde{L} \mathbbm{1}}{(B-1)(B-2)} - \frac{1}{B(B-3)}\frac{\mathbbm{1}^{T} \tilde{K} \mathbbm{1}\mathbbm{1}^{T} \tilde{L} \mathbbm{1}}{(B-1)(B-2)} \\
         =& \frac{1}{B(B-3)(B-1)(B-2)}\sum_{i,j: i\neq j}\left[k\left(\hat{\mathcal{A}}_i,\hat{\mathcal{A}}_j\right)-k\left(\mathcal{A}_{i}, \mathcal{A}_{j}\right)\right] \cdot \sum_{q,r: q\neq r} l_{qr}.
        \end{aligned}
    \end{equation}
    Taking absolute value of the above expression, and by the boundedness and Lipschitz continuity (Assumption \ref{Ass: Bounded Lipschitz Kernels}), we have 
    \begin{equation}\label{Eq: second term}
        \begin{aligned}
            &\left|\frac{1}{B(B-3)(B-1)(B-2)}\sum_{i,j: i\neq j}\left[k\left(\hat{\mathcal{A}}_i,\hat{\mathcal{A}}_j\right)-k\left(\mathcal{A}_{i}, \mathcal{A}_{j}\right)\right]  \sum_{q,r: q\neq r} l_{qr}\right|\\
            \le& \frac{MB(B-1)}{B(B-3)(B-1)(B-2)} \sum_{i, j:i\neq j}\left|k\left(\hat{\mathcal{A}}_{i}, \hat{\mathcal{A}}_{j}\right)-k\left(\mathcal{A}_{i}, \mathcal{A}_{j}\right)\right| \qquad \text{(by Assumption \ref{Ass: Bounded Lipschitz Kernels} (1))} \\
            \leq & \frac{MQ(B-1)}{(B-3)(B-2)} \left(\sum_{i}\left\|\hat{\mathcal{A}}_{i}-\mathcal{A}_{i}\right\|+\sum_{j}\left\|\hat{\mathcal{A}}_{j}-\mathcal{A}_{j}\right\|\right) \qquad \text{(by Assumption \ref{Ass: Bounded Lipschitz Kernels} (2))} \\
            = &  \frac{2MQ(B-1)}{(B-3)(B-2)} \sum_{i} \left\|\hat{\mathcal{A}}_{i}-\mathcal{A}_{i}\right\|.
        \end{aligned}
    \end{equation}
    \item[(3)] \textbf{Part III in Equation \eqref{Eq: test statistic}:} 
    \begin{equation}
        \begin{aligned}\nonumber
            &\frac{-1}{B(B-3)}\frac{2}{B-2} \mathbbm{1}^{T} \widetilde{\hat{K}} \tilde{L} \mathbbm{1}+ \frac{1}{B(B-3)}\frac{2}{B-2} \mathbbm{1}^{T} \tilde{K} \tilde{L} \mathbbm{1}\\
            =&\frac{-2}{B(B-3)(B-2)}\sum_{i,j,r:r\neq i,j \neq i}\left[k\left(\hat{\mathcal{A}}_i,\hat{\mathcal{A}}_j\right)-k(\mathcal{A}_{i},\mathcal{A}_{j})\right]l_{ir}.
        \end{aligned}
    \end{equation}
    Taking absolute value of the above expression, and by the boundedness and Lipschitz continuity (Assumption \ref{Ass: Bounded Lipschitz Kernels}), we have 
    \begin{equation}\label{Eq: third term}
    \begin{aligned}
        &\left|\frac{-2}{B(B-3)(B-2)}\sum_{i,j,r:r\neq i,j \neq i}\left[k\left(\hat{\mathcal{A}}_i,\hat{\mathcal{A}}_j\right)-k(\mathcal{A}_{i},\mathcal{A}_{j})\right]l_{ir}\right|\\
        \le &\frac{2}{B(B-3)(B-2)}\sum_{i} \sum_{r:r\neq i}\left|k(\hat{\mathcal{A}}_i,\hat{\mathcal{A}}_r)-k(\mathcal{A}_{i},\mathcal{A}_{r})\right| \cdot \sum_{j:j\neq i}|l_{ji}| \\
        \le &\frac{2M(B-1)}{B(B-3)(B-2)}\sum_{i} \sum_{r:r\neq i}\left|k(\hat{\mathcal{A}}_i,\hat{\mathcal{A}}_r)-k(\mathcal{A}_{i},\mathcal{A}_{r})\right| \qquad \text{(by Assumption \ref{Ass: Bounded Lipschitz Kernels} (1))} \\
        \le & \frac{2MQ(B-1)}{B(B-3)(B-2)} \sum_{i} \left(\sum_{r:r\neq i}  \left\|\hat{\mathcal{A}}_{r}-\mathcal{A}_{r}\right\| + (B-1)\left\|\hat{\mathcal{A}}_{i}-\mathcal{A}_{i}\right\| \right) \qquad \text{(by Assumption \ref{Ass: Bounded Lipschitz Kernels} (2))}\\ 
        \le & \frac{4MQ(B-1)^2}{B(B-3)(B-2)} \sum_{i} \left\|\hat{\mathcal{A}}_{i}-\mathcal{A}_{i}\right\|. 
    \end{aligned}
    \end{equation}
\end{enumerate}
    Combining the bounds derived in \eqref{Eq: first term}--\eqref{Eq: third term}, we obtain  
    \begin{equation}
        \begin{aligned}\nonumber
            |\hat{\eta}_b(\hat{\mathcal{A}})-\hat{\eta}_b(\mathcal{A})|= \frac{C(B)}{B}\sum_{i\in \mathcal{I}_b}\left\|\hat{\mathcal{A}}_{i}-\mathcal{A}_{i}\right\|,
        \end{aligned}
    \end{equation}
    where the constant $C(B):=\frac{2MQ(B-1)}{(B-3)}\left[1+\frac{B}{B-2}+\frac{2(B-1)}{B-2}\right]$, and $\mathcal{I}_b$ denotes the index set of the $b$-th block with $|\mathcal{I}_b|=B$. 

    \item \textbf{Step 2: Overall HSIC Error Bound $|\widehat{\mathrm{HSIC}}(\hat{\mathcal{A}},\mathcal{Z})-\widehat{\mathrm{HSIC}}(\mathcal{A},\mathcal{Z})|$.} 
    Recall that the block-based HSIC estimator is defined as $\widehat{\mathrm{HSIC}}(\mathcal{A},\mathcal{Z})=\frac{B}{m}\sum_{b=1}^{m/B}\hat{\eta}_b(\mathcal{A})$. Hence, summing over all $m/B$ blocks: 
        \begin{equation}
            \begin{aligned}\nonumber
            |\widehat{\mathrm{HSIC}}(\hat{\mathcal{A}},\mathcal{Z})-\widehat{\mathrm{HSIC}}(\mathcal{A},\mathcal{Z})|&=\frac{B}{m} \sum_{b=1}^{m/B}|\hat{\eta}_b(\hat{\mathcal{A}})-\hat{\eta}_b(\mathcal{A})|\\
            & \le\frac{B}{m} \sum_{b=1}^{m/B} \frac{C(B)}{B}\sum_{i\in\mathcal{I}_b}\left\|\hat{\mathcal{A}}_{i}-\mathcal{A}_{i}\right\|
            \le \frac{C(B)}{m}\sum_{i=1}^m\left\|\hat{\mathcal{A}}_{i}-\mathcal{A}_{i}\right\| 
            \end{aligned}
        \end{equation}
        Substituting the auxiliary variable $\mathcal{A}=Y-h(X)$, we obtain $\hat{\mathcal{A}}_i-\mathcal{A}_i=
        h(X_i)-\hat{h}(X_i)$. Thus,  
        $$|\widehat{\mathrm{HSIC}}(\hat{\mathcal{A}},\mathcal{Z})-\widehat{\mathrm{HSIC}}(\mathcal{A},\mathcal{Z})|
        \le\frac{C(B)}{m}\sum_{i=1}^m\bigl\|
        h(X_i)-\hat{h}(X_i)\bigr\|.$$
        By the law of large numbers, $\frac{1}{m}\sum_{i=1}^m\|h(X_i) - \hat{h}(X_i)\|
        \xrightarrow{p}
        \mathbb{E}[\|h(X_i) - \hat{h}(X_i)\|]$. Moreover, since $\hat{h}$ is estimated using another independent dataset $\mathcal{D}_1$ of sample size $n$, and is $\sqrt{n}$-consistent by Assumption \ref{Ass:sqrt-n-consistency}, we have $\|h-\hat{h}\|_{L_2({P_X})}=O_p(n^{-1/2})$. 
        Combining the above results yields 
        \begin{equation}\label{Eq:the difference of HSIC}
            |\widehat{\mathrm{HSIC}}(\hat{\mathcal{A}},\mathcal{Z})-\widehat{\mathrm{HSIC}}(\mathcal{A},\mathcal{Z})|=O_p\!\left(n^{-1/2}\right).
        \end{equation}

\item \textbf{Step 3: Asymptotic Validity $\sqrt{mB}\,\widehat{\mathrm{HSIC}}$.}
By the assumption that the sample sizes satisfy $mB = o(n)$ as $n \to \infty$, together with Equation~\eqref{Eq:the difference of HSIC}, Equation~\eqref{Eq: test statistic} can be expressed as follows: 
\[
\bigl|\sqrt{mB}\,\widehat{\mathrm{HSIC}}(\hat{\mathcal{A}},\mathcal{Z})
- \sqrt{mB}\,\widehat{\mathrm{HSIC}}(\mathcal{A},\mathcal{Z})\bigr|
= \sqrt{mB}\cdot O_p(n^{-1/2}) = o_p(1).
\]

Hence, 
\[
\sqrt{mB}\,\widehat{\mathrm{HSIC}}(\hat{\mathcal{A}},\mathcal{Z})
=\sqrt{mB}\,\widehat{\mathrm{HSIC}}(\mathcal{A},\mathcal{Z})+o_p(1).
\]

Under $H_0$, let $c_{m,\alpha}$ be a critical value such that the oracle statistic satisfies
\[
P_{H_0}\!\left(\sqrt{mB}\,\widehat{\mathrm{HSIC}}(\mathcal{A},\mathcal{Z})> c_{m,\alpha}\right)
\le \alpha + o(1).
\]

By Slutsky's theorem, 
\[
P_{H_0}\!\left(\sqrt{mB}\,\widehat{\mathrm{HSIC}}(\hat{\mathcal{A}},\mathcal{Z})> c_{m,\alpha}
\right)= {P}_{H_0}\bigl(\sqrt{mB}\,\widehat{\mathrm{HSIC}}({\mathcal{A}},\mathcal{Z}) + o_p(1) > c_{m,\alpha}\bigr)\le \alpha + o(1).
\]
Consequently, 
\[
P_{H_0}(\text{Type I error} )=P_{H_0}\!\left(\sqrt{mB}\,\widehat{\mathrm{HSIC}}(\hat{\mathcal{A}},\mathcal{Z})> c_{m,\alpha}
\right)\le \alpha + o(1).
\]

\end{itemize}

\noindent \textbf{2. Proof of the asymptotic power of the AIT condition.} 

Under the fixed alternative $H_1$, assume that 
$\mathrm{HSIC}(\mathcal{A}, \mathcal{Z}) = \eta > 0$. Recall that the block-based HSIC estimator can be written as
\[
\widehat{\mathrm{HSIC}}(\mathcal{A}, \mathcal{Z})
=
\frac{1}{M} \sum_{b=1}^{M} \hat{\eta}_b,
\quad
M = m/B.
\]
Under $H_1$, the block statistics $\{\hat{\eta}_b\}_{b=1}^{M}$ are i.i.d. with 
$\mathbb{E}[\hat{\eta}_b] = \eta > 0$ and finite variance $\mathrm{Var}(\hat{\eta}_b)=\tau^2<\infty$.
By the classical central limit theorem,
\[
\sqrt{M}\left(
\widehat{\mathrm{HSIC}}(\mathcal{A}, \mathcal{Z}) - \eta
\right)
\xrightarrow{D}
\mathcal{N}(0, \tau^2).
\]
Since $M = m/B$, multiplying both sides by $B$, this is equivalent to
\[
\sqrt{mB}\left(
\widehat{\mathrm{HSIC}}(\mathcal{A}, \mathcal{Z}) - \eta
\right)
\xrightarrow{D}
\mathcal{N}(0, \sigma_{H_1}^2),
\]
where $\sigma_{H_1}^2 = \tau^2 B^2$.

By the plug-in error bound established in Equation \eqref{Eq:the difference of HSIC}, $\widehat{\mathrm{HSIC}}(\hat{\mathcal{A}}, \mathcal{Z})=
\widehat{\mathrm{HSIC}}(\mathcal{A}, \mathcal{Z})
+O_p(n^{-1/2})$, and under $mB=o(n)$,
\[
\sqrt{mB}\left(
\widehat{\mathrm{HSIC}}(\hat{\mathcal{A}}, \mathcal{Z})
-
\widehat{\mathrm{HSIC}}(\mathcal{A}, \mathcal{Z})
\right)
=
o_p(1).
\]
By Slutsky’s theorem,
\[
\sqrt{mB}\left(
\widehat{\mathrm{HSIC}}(\hat{\mathcal{A}}, \mathcal{Z})
-
\eta
\right)
\xrightarrow{D}
\mathcal{N}(0, \sigma_{H_1}^2).
\]

The Type II error probability is
\[
P_{H_1}(\text{Type II error})=P_{H_1}\!\left(
\sqrt{mB}\,\widehat{\mathrm{HSIC}}(\hat{\mathcal{A}}, \mathcal{Z})
\le c_{m,\alpha}
\right),
\]
where $c_{m,\alpha}$ is defined as the $(1-\alpha)$-quantile of the null limiting distribution. Subtracting $\sqrt{mB}\eta$ and standardizing yields
\[
P_{H_1}\left(\frac{\sqrt{mB}\left(
\widehat{\mathrm{HSIC}}(\hat{\mathcal{A}}, \mathcal{Z}) - \eta
\right)}{\sigma_{H_1}}\le \frac{c_{m,\alpha} - \sqrt{mB}\eta}{\sigma_{H_1}}\right) := 
P_{H_1}\!\left(T_{m,B}\le \Delta_{m,B}\right),
\]
where $T_{m,B}\xrightarrow{D}\mathcal{N}(0,1)$. Since the null distribution is non-degenerate, its quantiles are finite, and hence $c_{m,\alpha} = O_p(1)$. Moreover, because $\eta>0$,
\[
\Delta_{m,B}
=
-
\frac{\eta}{\sigma_{H_1}}\sqrt{mB}
+
O(1)
\to -\infty.
\]
Hence,
\[
P_{H_1}(\text{Type II error})
=
\Phi(\Delta_{m,B}) + o(1), \text{ as } mB \to \infty.
\]

Applying Mills' ratio for $\Delta_{m,B}<0$ gives
\[
\Phi(\Delta_{m,B})
\le
\frac{1}{\sqrt{2\pi}(-\Delta_{m,B})}
\exp\!\left(-\frac{\Delta_{m,B}^2}{2}\right)
=
O\!\left(
\frac{1}{\sqrt{mB}}
e^{-c mB}
\right),
\]
where $c=\eta^2/(2\sigma_{H_1}^2)>0$.
Therefore,
\[
P_{H_1}(\text{Type II error})
=
O\!\left(
\frac{1}{\sqrt{mB}} e^{-c mB}
\right).
\]
\end{proof}

\tworevise{\subsection{Proof of Corollary \ref{The-HSIC-covariate}: Asymptotic Level and Power of the AIT Test with Covariates}\label{App:proof of HSIC of covariates}
\begin{proof}
Let $h(\cdot)$ and $\pi(\cdot)$ denote the oracle functions, and let $\hat{h}(\cdot)$ and $\hat{\pi}(\cdot)$ denote their estimators trained on the dataset $\mathcal{D}_1 \subset \mathcal{D}$. 
We establish our theoretical results using an independent dataset $\mathcal{D}_2 = \{(X_i, Y_i, \mathbf{W}_i, Z_i)\}_{i=1}^m \subset \mathcal{D}$. 
For notational simplicity, below we use $(X, Y, \mathbf{W}, Z)$ to denote a generic observation drawn from $\mathcal{D}_2$. Recall that $\mathcal{A} = Y - h(X)$ and $\hat{\mathcal{A}} = Y - \hat{h}(X)$, and that $\mathcal{Z} = Z - \pi(\mathbf{W})$ and $\hat{\mathcal{Z}} = Z - \hat{\pi}(\mathbf{W})$. 
To establish the asymptotic level and power of the AIT test, we analyze the difference between 
$\widehat{\mathrm{HSIC}}(\hat{\mathcal{A}}, \hat{\mathcal{Z}})$ 
and its oracle counterpart 
$\widehat{\mathrm{HSIC}}(\mathcal{A}, \mathcal{Z})$, where $\widehat{\mathrm{HSIC}}(\cdot,\cdot)$ denotes the block-based estimator of \citet{zhang2018large}, \ie, 
$\bigl| \widehat{\mathrm{HSIC}}(\hat{\mathcal{A}}, \hat{\mathcal{Z}})-\widehat{\mathrm{HSIC}}(\mathcal{A}, \mathcal{Z})\bigr|$. 

Applying the triangle inequality, we have 
\begin{equation}
\label{Eq:HSIC-difference-covariates}
\begin{aligned}
&\bigl|
\widehat{\mathrm{HSIC}}(\hat{\mathcal{A}}, \hat{\mathcal{Z}})
-
\widehat{\mathrm{HSIC}}(\mathcal{A}, \mathcal{Z})
\bigr| \\
\le\;&
\underbrace{
\bigl|
\widehat{\mathrm{HSIC}}(\hat{\mathcal{A}}, \hat{\mathcal{Z}})
-
\widehat{\mathrm{HSIC}}(\mathcal{A}, \hat{\mathcal{Z}})
\bigr|
}_{\text{Part I: error induced by }\hat{\mathcal{A}}}
+
\underbrace{
\bigl|
\widehat{\mathrm{HSIC}}(\mathcal{A}, \hat{\mathcal{Z}})
-
\widehat{\mathrm{HSIC}}(\mathcal{A}, \mathcal{Z})
\bigr|
}_{\text{Part II: error induced by }\hat{\mathcal{Z}}}.
\end{aligned}
\end{equation}

We next show that, by arguments analogous to those in the proof of Theorem~\ref{The:HSIC}, the two terms can be bounded separately. Specifically, by Assumptions~\ref{Ass: Bounded Lipschitz Kernels with covariates} and \ref{Ass:sqrt-n-consistency-covariates}, we have
\[
\bigl|
\widehat{\mathrm{HSIC}}(\hat{\mathcal{A}}, \hat{\mathcal{Z}})
-
\widehat{\mathrm{HSIC}}(\mathcal{A}, \hat{\mathcal{Z}})
\bigr|
= O_p(n^{-1/2}).
\]
Moreover, by Assumptions~\ref{Ass: Bounded Lipschitz Kernels with covariates} and \ref{Ass:rate}, it holds that
\[
\bigl|
\widehat{\mathrm{HSIC}}(\mathcal{A}, \hat{\mathcal{Z}})
-
\widehat{\mathrm{HSIC}}(\mathcal{A}, \mathcal{Z})
\bigr|
= O_p(n^{-q}), \quad q \in (0,1/2).
\]

Combining the above two results, we obtain 
\[
\bigl|
\widehat{\mathrm{HSIC}}(\hat{\mathcal{A}}, \hat{\mathcal{Z}})
-
\widehat{\mathrm{HSIC}}(\mathcal{A}, \mathcal{Z})
\bigr|
=
O_p(n^{-1/2}) + O_p(n^{-q}),
\qquad q \in (0,1/2).
\]
Since $n^{-q}$ dominates $n^{-1/2}$ for $q < 1/2$, the overall convergence rate
is $O_p(n^{-q})$, \ie, $\bigl|
\widehat{\mathrm{HSIC}}(\hat{\mathcal{A}}, \hat{\mathcal{Z}})
-\widehat{\mathrm{HSIC}}(\mathcal{A}, \mathcal{Z})
\bigr|=O_p(n^{-q})$. 

By the assumption that the sample sizes satisfy $mB = o(n^{2q})$, which ensures that: 
\begin{equation}\label{Eq:difference of HSIC with covariates}
    \begin{aligned}
    \bigl|\sqrt{mB}\,\widehat{\mathrm{HSIC}}(\hat{\mathcal{A}}, \hat{\mathcal{Z}})
    -\sqrt{mB}\,\widehat{\mathrm{HSIC}}(\mathcal{A}, \mathcal{Z})\bigr|= o_p(1).
    \end{aligned}
\end{equation}

Below, we prove the asymptotic level and power of the AIT condition test, respectively. 

\noindent \textbf{1. Proof of the asymptotic level of the AIT condition with covariates.}

Under $H_0$, let $c_{m,\alpha}$ be a critical value such that the oracle statistic satisfies 
\[
P_{H_0}\!\left(\sqrt{mB}\,\widehat{\mathrm{HSIC}}(\mathcal{A},\mathcal{Z})> c_{m,\alpha}\right)\le
\alpha + o(1).
\]
By Slutsky's theorem, we further have 
\[
{P}_{H_0}(\sqrt{mB}\,\widehat{\mathrm{HSIC}}(\hat{\mathcal{A}},\hat{\mathcal{Z}}) > c_{m,\alpha}) 
= {P}_{H_0}(\sqrt{mB}\,\widehat{\mathrm{HSIC}}(\mathcal{A},\mathcal{Z}) + o_p(1) > c_{m,\alpha})
\le \alpha + o(1).
\]
Consequently, 
$${P}_{H_0}(\text{Type I error}) = {P}_{H_0}(\sqrt{mB}\,\widehat{\mathrm{HSIC}}(\hat{\mathcal{A}},\hat{\mathcal{Z}}) > c_{m,\alpha}) \le \alpha + o(1).$$

\noindent \textbf{2. Proof of the asymptotic power of the AIT condition with covariates.}

We proceed analogously to the proof of Theorem~\ref{The:HSIC}. 
In particular, applying the same block-averaging and a central limit theorem argument to the oracle variables $(\mathcal A,\mathcal Z)$ yields
\[
\sqrt{mB}\left(
\widehat{\mathrm{HSIC}}(\mathcal{A}, \mathcal{Z}) - \eta
\right)
\xrightarrow{D}
\mathcal{N}(0,\sigma_{H_1}^2),
\qquad
\sigma_{H_1}^2=\tau^2 B^2,
\]
where $\eta=\mathrm{HSIC}(\mathcal A,\mathcal Z)>0$ and $\tau^2=\mathrm{Var}(\hat{\eta}_b)<\infty$.

According to Equation \eqref{Eq:difference of HSIC with covariates}, we have 
\[
\sqrt{mB}\left(\widehat{\mathrm{HSIC}}(\hat{\mathcal{A}}, \hat{\mathcal{Z}}) - \eta\right)
=\sqrt{mB}\left(\widehat{\mathrm{HSIC}}(\mathcal{A}, \mathcal{Z}) - \eta\right)+o_p(1).
\]
By Slutsky’s theorem, 
\[
\sqrt{mB}\left(
\widehat{\mathrm{HSIC}}(\hat{\mathcal{A}}, \hat{\mathcal{Z}}) - \eta
\right)
\xrightarrow{D}
\mathcal{N}(0,\sigma_{H_1}^2).
\]

The Type II error probability is
\[
P_{H_1}(\text{Type II error})
=
P_{H_1}\!\left(
\sqrt{mB}\,\widehat{\mathrm{HSIC}}(\hat{\mathcal{A}}, \hat{\mathcal{Z}})
\le c_{m,\alpha}
\right),
\]
where $c_{m,\alpha}$ is defined as the $(1-\alpha)$-quantile of the null limiting distribution. Subtracting $\sqrt{mB}\eta$ and standardizing yields
\[
P_{H_1}\left(\frac{\sqrt{mB}\left(
\widehat{\mathrm{HSIC}}(\hat{\mathcal{A}}, \hat{\mathcal{Z}}) - \eta
\right)}{\sigma_{H_1}} \le \frac{c_{m,\alpha}-\sqrt{mB}\eta}{\sigma_{H_1}} \right):=P_{H_1}\!\left(T_{m,B}\le \Delta_{m,B}\right),
\]
where $T_{m,B}\xrightarrow{D}\mathcal{N}(0,1)$. Since the null distribution is non-degenerate, its quantiles are finite, and hence $c_{m,\alpha}=O_p(1)$. Moreover, because $\eta>0$, we have
\[
\Delta_{m,B}=-\frac{\eta}{\sigma_{H_1}}\sqrt{mB}+O_p(1)\to -\infty.
\]

Hence,
\[
P_{H_1}(\text{Type II error})
=
\Phi(\Delta_{m,B})+o(1),
\text{ as } mB\to\infty.
\]
Applying Mills' ratio for $\Delta_{m,B}<0$ gives
\[
\Phi(\Delta_{m,B})
\le
\frac{1}{\sqrt{2\pi}(-\Delta_{m,B})}
\exp\!\left(-\frac{\Delta_{m,B}^2}{2}\right)
=
O\!\left(\frac{1}{\sqrt{mB}}e^{-c mB}\right),
\]
where $c=\eta^2/(2\sigma_{H_1}^2)>0$.
Therefore,
\[
P_{H_1}(\text{Type II error})
=
O\!\left(\frac{1}{\sqrt{mB}}e^{-c mB}\right).
\]
\end{proof}

}

\tworevise{
\section{Two Analytic Examples Satisfying Assumption~\ref{Ass-algebraic-condition}}
\label{App-example of assumption 3}

In this section, we provide two fully analytic data-generating processes and verify 
Assumption~\ref{Ass-algebraic-condition} in a straightforward manner. Our verification is based on the following observation: since $p(\mathcal A,Z)=p(\mathcal A\mid Z)p(Z)$, we have
\[
\log p(\mathcal A,Z)=\log p(\mathcal A\mid Z)+\log p(Z),
\quad \Rightarrow \quad
\frac{\partial^2}{\partial \mathcal A\,\partial Z}\log p(\mathcal A,Z)
=
\frac{\partial^2}{\partial \mathcal A\,\partial Z}\log p(\mathcal A\mid Z),
\]
because $\log p(Z)$ does not depend on $\mathcal A$.
Thus, it suffices to show that the mixed derivative of $\log p(\mathcal A\mid Z)$ is nonzero on a set with non-zero Lebesgue measure.

Throughout, we define the auxiliary variable for the constant effect model as
\[
\mathcal A_{X\to Y\|Z}:=Y-\hat\beta X,
\qquad
\hat\beta:=\frac{\Cov(Y,Z)}{\Cov(X,Z)}.
\]

\subsection{Example \ref{Example:violation of exogeneity}: Violation of the Exogeneity Condition (\texorpdfstring{$Z \nCI U$)})}
\label{App:Ex1-exogeneity}

Consider the nonlinear Gaussian noise model
\begin{equation}\label{Eq:App-ex1}
\begin{aligned}
U &= \varepsilon_U, \qquad
Z = \gamma U+\varepsilon_Z,\\
X &= \exp(Z)+\rho U+\varepsilon_X,\qquad
Y = \beta X+\kappa U+\varepsilon_Y,
\end{aligned}
\end{equation}
where $(\varepsilon_U,\varepsilon_Z,\varepsilon_X,\varepsilon_Y)\stackrel{\mathrm{ind}}{\sim}\mathcal N(0,1)$.
Since $Z=\gamma U+\varepsilon_Z$, we have $Z \nCI U$ whenever $\gamma\neq 0$, i.e., exogeneity is violated.

\paragraph{Step 1: Compute $\hat\beta=\Cov(Y,Z)/\Cov(X,Z)$.} Let $\sigma_Z^2:=\Var(Z)=\gamma^2+1$.
For $Z\sim \mathcal N(0,\sigma_Z^2)$, the moment identity gives
\begin{equation}\label{Eq:App-cov-expZ}
\Cov(\exp(Z),Z)=\mathbb E[Z e^{Z}]
=\left.\frac{d}{dt}\mathbb E[e^{tZ}]\right|_{t=1}
=\left.\frac{d}{dt}\exp\!\left(\tfrac12\sigma_Z^2 t^2\right)\right|_{t=1}
=\sigma_Z^2\exp\!\left(\tfrac12\sigma_Z^2\right).
\end{equation}
Moreover, $\Cov(U,Z)=\Cov(U,\gamma U+\varepsilon_Z)=\gamma$.
Hence
\[
\Cov(X,Z)=\Cov(\exp(Z),Z)+\rho\,\Cov(U,Z)
=\sigma_Z^2 \exp({\sigma_Z^2/2})+\rho\gamma.
\]
Since $Y=\beta X+\kappa U+\varepsilon_Y$ and $\varepsilon_Y \CI Z$,
\[
\Cov(Y,Z)=\beta\,\Cov(X,Z)+\kappa\,\Cov(U,Z)
=\beta\big(\sigma_Z^2 \exp({\sigma_Z^2/2})+\rho\gamma\big)+\kappa\gamma.
\]
Therefore,
\begin{equation}\label{Eq:App-ex1-betahat}
\hat\beta
=\frac{\Cov(Y,Z)}{\Cov(X,Z)}
=\beta+\underbrace{\frac{\kappa\gamma}{\sigma_Z^2 \exp({\sigma_Z^2/2})+\rho\gamma}}_{\beta_{bias}}.
\end{equation}

\paragraph{Step 2: Derive $\mathcal A$ and $p(\mathcal A\mid Z)$.}

Using $Y=\beta X+\kappa U+\varepsilon_Y$,
\begin{align*}
\mathcal A
&=Y-\hat\beta X \\
&=(\beta-\hat\beta)X+\kappa U+\varepsilon_Y\\
&=(\beta-\hat\beta)\exp(Z)+\big((\beta-\hat\beta)\rho+\kappa\big)U+(\beta-\hat\beta)\varepsilon_X+\varepsilon_Y.
\end{align*}
Conditional on $Z=z$, the random variable $U\mid Z=z$ is Gaussian since $(U,Z)$ is jointly Gaussian. 
Specifically, 
\[
\mathbb E[U\mid Z=z]=\frac{\Cov(U,Z)}{\Var(Z)}z=\frac{\gamma}{\sigma_Z^2}z,
\qquad
\Var(U\mid Z=z)=1-\frac{\Cov(U,Z)^2}{\Var(Z)}=\frac{1}{\sigma_Z^2}.
\]
Since $(\varepsilon_X,\varepsilon_Y)\CI Z$ and are Gaussian, it follows that 
$\mathcal A\mid Z=z$ is Gaussian with mean $m(z)$ and variance $v$ given by 
\begin{equation}\label{Eq:App-ex1-meanvar}
\begin{aligned}
m(z)
&=(\beta-\hat\beta)\exp({z})+\big((\beta-\hat\beta)\rho+\kappa\big)\frac{\gamma}{\sigma_Z^2}z,\\
v
&=\big((\beta-\hat\beta)\rho+\kappa\big)^2\frac{1}{\sigma_Z^2}+(\beta-\hat\beta)^2+1.
\end{aligned}
\end{equation}
In particular, $v>0$ is constant (independent of $z$).

\paragraph{Step 3: Verify Assumption~\ref{Ass-algebraic-condition}.}

Since $\mathcal A\mid Z=z\sim \mathcal N(m(z),v)$,
\[
\log p(\mathcal A\mid Z=z)
=-\frac{(\mathcal A-m(z))^2}{2v}-\frac12\log(2\pi v).
\]
Hence
\[
\frac{\partial}{\partial \mathcal A}\log p(\mathcal A\mid Z=z)
=-\frac{\mathcal A-m(z)}{v},
\qquad
\frac{\partial^2}{\partial \mathcal A\,\partial z}\log p(\mathcal A\mid Z=z)
=\frac{m'(z)}{v}.
\]
Using Equation~\eqref{Eq:App-ex1-meanvar}, we have the derivative 
\begin{equation}\label{Eq:App-ex1-mprime}
m'(z)
=(\beta-\hat\beta)\exp({z})+\big((\beta-\hat\beta)\rho+\kappa\big)\frac{\gamma}{\sigma_Z^2}.
\end{equation}
Therefore, combining Equations~\eqref{Eq:App-ex1-meanvar} and \eqref{Eq:App-ex1-mprime} yields 
\begin{equation}\label{Eq:App-ex1-mixed}
\frac{\partial^2}{\partial \mathcal A\,\partial Z}\log p(\mathcal A,Z)
=
\frac{\partial^2}{\partial \mathcal A\,\partial Z}\log p(\mathcal A\mid Z)
=\frac{m'(Z)}{v} = \frac{\kappa\gamma-\beta_{bias}[\rho\gamma+\exp(Z)(\gamma^2+1)]}{(\kappa-\beta_{bias}\rho)^2 + (\gamma^2+1)(\beta_{bias}^2+1)}.
\end{equation}
Since $m'(z)$ is not identically zero (which holds generically; e.g., if $\kappa\gamma-\beta_{bias}[\rho\gamma+\exp(z)(\gamma^2+1)] \neq 0$), 
the set $\{z: m'(z)\neq 0\}$ has non-zero Lebesgue measure. Consequently, Equation~\eqref{Eq:App-ex1-mixed} is nonzero on a set with non-zero Lebesgue measure, and Assumption~\ref{Ass-algebraic-condition} holds.

\subsection{Example \ref{Example:violation of exclusion}: Violation of the Exclusion Restriction Condition ($Z\to Y$ Direct Effect)}
\label{App:Ex2-exclusion}

Consider the model
\begin{equation}\label{Eq:App-ex2}
\begin{aligned}
U &= \varepsilon_U,\qquad
Z = \varepsilon_Z,\\
X &= \exp(Z)+\rho U+\varepsilon_X,\qquad
Y = \beta X+\nu Z+\kappa U+\varepsilon_Y,
\end{aligned}
\end{equation}
where $(\varepsilon_U,\varepsilon_Z,\varepsilon_X,\varepsilon_Y)\stackrel{\mathrm{ind}}{\sim}\mathcal N(0,1)$. 
Here $Z\CI U$ (exogeneity holds), but exclusion is violated when $\nu\neq 0$ due to the direct effect $\nu Z$.

\paragraph{Step 1: Compute $\hat\beta=\Cov(Y,Z)/\Cov(X,Z)$.}
Since $Z\sim\mathcal N(0,1)$, by Equation~\eqref{Eq:App-cov-expZ} with $\sigma_Z^2=1$,
\[
\Cov(\exp(Z),Z)=\exp(1/2).
\]
Moreover, $U \CI Z$ and $\varepsilon_X \CI Z$, hence
\[
\Cov(X,Z)=\Cov(\exp(Z),Z)=\exp({1/2}).
\]
Since $Y=\beta X+\nu Z+\kappa U+\varepsilon_Y$ and $(U,\varepsilon_Y) \CI Z$,
\[
\Cov(Y,Z)=\beta\,\Cov(X,Z)+\nu\,\Var(Z)=\beta \exp({1/2})+\nu.
\]
Therefore, 
\begin{equation}\label{Eq:App-ex2-betahat}
\hat\beta=\frac{\Cov(Y,Z)}{\Cov(X,Z)}=\beta+ \underbrace{\nu \exp({-1/2})}_{\beta_{bias}}.
\end{equation}

\paragraph{Step 2: Derive $\mathcal A$ and $p(\mathcal A\mid Z)$.}
We have
\begin{align*}
\mathcal A
&=Y-\hat\beta X\\
&=(\beta-\hat\beta)X+\nu Z+\kappa U+\varepsilon_Y\\
&=(\beta-\hat\beta)\exp({Z})+\nu Z+\big((\beta-\hat\beta)\rho+\kappa\big)U+(\beta-\hat\beta)\varepsilon_X+\varepsilon_Y.
\end{align*}
Conditional on $Z=z$, $U$ remains $\mathcal N(0,1)$ and is independent of $Z$. Thus $\mathcal A\mid Z=z$ is Gaussian with mean and variance
\begin{equation}\label{Eq:App-ex2-meanvar}
\begin{aligned}
m(z)
&=(\beta-\hat\beta)\exp({z})+\nu z,\\
v
&=\big((\beta-\hat\beta)\rho+\kappa\big)^2+(\beta-\hat\beta)^2+1,
\end{aligned}
\end{equation}
where $v>0$ is constant.

\paragraph{Step 3: Verify Assumption~\ref{Ass-algebraic-condition}.}
As in Example~\ref{App:Ex1-exogeneity}, using Equations~\eqref{Eq:App-ex2-betahat} and \eqref{Eq:App-ex2-meanvar}, we obtain $m'(z)=(\beta-\hat\beta)\exp({z})+\nu$, which implies 
\begin{equation}\label{Eq:second-order}
\frac{\partial^2}{\partial \mathcal A\,\partial Z}\log p(\mathcal A,Z)=
\frac{\partial^2}{\partial \mathcal A\,\partial Z}\log p(\mathcal A\mid Z)
=\frac{m'(Z)}{v}=\frac{\nu [1-\exp(Z-1/2)]}{[\kappa-\nu \exp{(\frac{-1}{2})}]^2 + \nu^2\exp{(-1)+1}}.
\end{equation}
Since $\nu\neq 0$ (i.e., exclusion is violated), and $m'(z)\neq 0$ for all $z\neq 1/2$, $\{z:m'(z)\neq 0\}$ has non-zero Lebesgue measure. Therefore, the mixed second derivative (Equation \eqref{Eq:second-order}) is nonzero on a set with non-zero Lebesgue measure, and Assumption~\ref{Ass-algebraic-condition} holds.

}

\tworevise{
\section{More Details on Simulation Experiments in Section \ref{Section:experiments}}\label{Appendix-Synthetic}
In this section, we provide details of the simulation experiments corresponding to Tables \ref{Table-linear-model-A2} $\sim$ \ref{Table-type-II}. Specifically, the generation mechanism for each table is as follows: 
\begin{itemize}
    \item Table \ref{Table-linear-model-A2}. The model setup is as follows: $U=\varepsilon_U$, $Z_1 = U + \varepsilon_{Z1}$, $Z_2 = \varepsilon_{Z2}$, $X = \tau_1 Z_1 + \tau_2 Z_2 + \rho U + \varepsilon_{X}$, $Y = X + \kappa U + \varepsilon_{Y}$, where all constant coefficients are randomly selected from a uniform distribution with parameters min = 0.5 and max = 1.5. The noise terms $\varepsilon_U, \varepsilon_{Z1}, \varepsilon_{Z2}, \varepsilon_{X}$, and $ \varepsilon_{Y}$ follow the specific distributions listed in each row. The final row indicates that all noise terms are randomly drawn from one of six distributions \footnote{These six distributions include Uniform, Beta, T, Gamma, Lognormal, and Gaussian.}.
    
    \item Table \ref{Table-non-linear-constant-model-A2}. The model setup is as follows: $U=\varepsilon_U$, $Z_1 = \varphi_{Z1}(U) + \varepsilon_{Z1}$, $Z_2 = \varepsilon_{Z2}$, $X = \tau_1 Z_1 + \tau_2 Z_2 + \rho U + \varepsilon_{X}$, $Y = \beta X + \kappa U + \varepsilon_{Y}$, where all constant coefficients are set to 1, and all noise terms follow the Gaussian distribution with mean 0 and standard deviation 1. The nonlinear function $\varphi_{Z1}(U)$ matches the corresponding function provided in each row. The details of the nonlinear function are as follows:
    \begin{equation}
    \begin{aligned}
    \text{Log:}\quad Y &= \log_e(\tworevise{0.2 |X- 1|}); \\
    \text{Quadratic polynomial:}\quad Y &= X^2 - 2 \cdot X + 1;\\
    \text{Cubic polynomial:}\quad Y &= X^3 - 0.5 \cdot X^2 + 0.2 \cdot X; \\
    \text{Log (quadratic):}\quad Y &= \log_e(|0.5 \cdot X^2 + X|); \\
   \text{Exp (quadratic):}\quad Y &= e^{0.3 \cdot X^2 + \tworevise{1}}; \\
    \end{aligned}
    \end{equation}
    
    \item Table \ref{Table-non-linear-non-constant-model-A2}. The model setup is as follows: $U=\varepsilon_U$, $Z_1 = \gamma U + \varepsilon_{Z1}$, $Z_2 = \varepsilon_{Z2}$, $X = \tau_1 Z_1 + \tau_2 Z_2 + \rho U + \varepsilon_{X}$, $Y = f(X) + \kappa U + \varepsilon_{Y}$, where all constant coefficients are set to 1, and all noise terms follow the Uniform distribution with parametric min = -2 and max = 2. The nonlinear function $f(X)$ corresponds to the specific function listed in each row.  
    The details of the nonlinear function are as follows:
    \begin{equation}
    \begin{aligned}
    \text{Log:}\quad Y &= \tworevise{8}\cdot\log_e(|X|); \\
    \text{Quadratic polynomial:}\quad Y &= \tworevise{(5\cdot X + 2)^2 - 10};\\
    \text{Cubic polynomial:}\quad Y &= \tworevise{5\cdot X^3 + 2 \cdot X^2 + 2 \cdot X - 3}; \\
    \text{Log (quadratic):}\quad Y &= \tworevise{8\cdot\log_e(2 \cdot X^2 + X + 1|)}; \\
    \text{Exp (quadratic):}\quad Y &= \tworevise{3\cdot e^{1.5 \cdot X^2 + 0.5\cdot X + 2}}; \\
\end{aligned}
    \end{equation}

    \item Table \ref{Table-non-linear-constant-model-A3}. The model setup is as follows: $U=\varepsilon_U$, $Z_1 = \varepsilon_{Z1}$, $Z_2 = \varepsilon_{Z2}$, $X = sign(Z_1) + g_X(Z_2) + \rho U + \varepsilon_{X}$, $Y = X + g_Y(Z_1) + \kappa U + \varepsilon_{Y}$, \tworevise{where all constant coefficients are set to 1}, and all noise terms follow a Beta distribution with parameters alpha = 0.5 and beta = 0.1. The sign($*$) denotes sign function, where $* > 0$ equals 1, $* = 0$ equals 0, and otherwise, it equals -1. The nonlinear functions $g_X(Z_2), g_Y(Z_1)$ are defined by the specific functions provided in each row. The specific nonlinear functions are as follows:
    \begin{equation}
   \begin{aligned}
    \text{Log:}\quad Y &= \tworevise{4.5\cdot\log_e(2 \cdot|X| + 0.3)}; \\
    \text{Quadratic polynomial:}\quad Y &= \tworevise{(1.5\cdot X + 2)^2 - 8};\\
    \text{Cubic polynomial:}\quad Y &= \tworevise{3\cdot X^3 + 2\cdot X^2 + 5.5\cdot X + 5}; \\
    \text{Log (quadratic):}\quad Y &= \tworevise{\log_e(|4.5 \cdot X^2 - 0.1\cdot X - 0.1|)}; \\
    \text{Exp (quadratic):}\quad Y &= \tworevise{e^{2.5 \cdot X^2 + 1.9\cdot X}}; \\
\end{aligned}
    \end{equation}
    \item Table \ref{Table-non-linear-non-constant-model-A3}. The model setup is as follows: $U=\varepsilon_U$, $Z_1 = \varepsilon_{Z1}$, $Z_2 = \varepsilon_{Z2}$, $X = sign(Z_1) + g_X(Z_2) + \varphi_X(U) + \varepsilon_{X}$, $Y = f(X) + g_Y(Z_1) + \varphi_Y(U) + \varepsilon_{Y}$, where sign($*$) denotes sign function, \tworevise{and all noise terms follow the Uniform distribution with parametric min = -2 and max = 2}. The nonlinear functions correspond to the specific functions provided in each row. The detailed forms of these nonlinear functions are as follows:
    
    \begin{equation}
    \begin{aligned}
    \text{Log:}\quad Y &= \tworevise{\log_e(|X -0.01|) - 0.01}; \\
    \text{Quadratic polynomial:}\quad Y &= \tworevise{0.1 \cdot (0.5\cdot X -0.1)^2 - 10};\\
    \text{Cubic polynomial:}\quad Y &= \tworevise{ X^3 + 0.5\cdot X^2 + 0.2\cdot X - 0.05}; \\
    \text{Log (quadratic):}\quad Y &= \tworevise{\log_e(|0.5 \cdot X^2 + X| + 1)}; \\
    \text{Exp (quadratic):}\quad Y &= \tworevise{e^{0.2 \cdot X^2}}; \\
\end{aligned}
    \end{equation}
    \item Table \ref{Table-IV-PIM-compared-linear-model}. The specific generation mechanism for the linear model with covariates $\mathbf{W}$ is defined as follows: $U =\varepsilon_U$, $\mathbf{W} = \boldsymbol{\varepsilon_{W}}$, $Z_1 =\mathcal{I}(U + \mathbf{W} + \varepsilon_{Z1})$, $Z_2= \mathcal{I}(\mathbf{W} + \varepsilon_{Z2})$, $X = 0.5Z_1 + 0.5Z_2 + \boldsymbol{\lambda W} + \delta$, and $Y = X + \mathbf{W} + \epsilon$, where $\varepsilon_{U} \sim T(5)$, $\varepsilon_{Z1} \sim Beta(0.5, 0.1)$, $\varepsilon_{Z2} \sim \mathcal{N}(0,1)$, and $\delta, \epsilon \sim T(5)$. Here, $\mathcal{I}(*)$ is the indicator function such that $\mathcal{I}(*) > \text{mean}(*)$ equals 1; otherwise, it is 0. The coefficient $\boldsymbol{\lambda}$ is randomly drawn from a normalized standard normal distribution. The noise terms $\boldsymbol{\varepsilon_{W}}$ follow a multidimensional normal distribution, consistent with IV-PIM, with the dimensionality of covariates $\mathbf{W}$ varying across $|\mathbf{W}|=\{2, 3, 5\}$. 

    In the IV-PIM method, the parameters are set as follows: the number of bootstrap samples $B_{bootstrap}=5$, the kappa method is specified as \emph{spectral}, and the synthetic treatment variable method is set to \emph{knockoff}.
    
    \item Table \ref{Table-binary-treatment-model}. The discrete treatment data that simulates violations of the \emph{exogeneity} and \emph{exclusion restriction} conditions as follows:
    $U = \varepsilon_{U}$, $Z =  \mathcal{I} ( \varphi_{Z}(U) + \varepsilon_{Z})$, $X = \mathcal{I} (g_{X}(Z)+ \varphi_{X}(U) + \varepsilon_{X})$, $Y = \beta X + g_{Y}(Z) + \varphi_{Y}(U) + \varepsilon_{Y}$, and $\varepsilon_{*} \sim \mathcal{N}(0,1)$, where $\beta=1$, and $\mathcal{I}(*)$ is the indicator function such that $\mathcal{I}(*) > \text{mean}(*)$ equals 1; otherwise, it is 0. The functions $\varphi_{*}(U)$ and $g_{*}(Z)$ are nonlinear and randomly selected from the following: {\emph{cos, sin, square, cubic(third-degree polynomials), logarithmic, exponential function}}. 
    \item {Table \ref{Table-type-I}. The model settings are as follows:}
    \begin{itemize}
        \item {Linear Model: $U = \varepsilon_{U}$, $Z = \varepsilon_{Z}$, $X = Z+ U + \varepsilon_{X}$, $Y = X + U + \varepsilon_{Y}$, where $\varepsilon_U$, $\varepsilon_Z$ $\sim$ $\mathcal{U}[-1,0)\cup(0,1]$, and $\varepsilon_X, \varepsilon_Y \sim \mathcal{N}(0,1)$.}
        \item {Partial Non-Linear Model with Constant Causal Effect: $U = \varepsilon_{U}$, $Z = \varepsilon_{Z}$, $X = sin(Z)+ U + \varepsilon_{X}$, $Y = X + U^2 + \varepsilon_{Y}$, where all independent noise follow the standard Gaussian distribution $\varepsilon_* \sim \mathcal{N}(0,1)$.} 
        \item {ANINCE Model: $U = \varepsilon_{U}$, $Z = \varepsilon_{Z}$, $X = Z^2+ U + \varepsilon_{X}$, $Y = e^X + U + \varepsilon_{Y}$, where all independent noise follow the standard Gaussian distribution $\varepsilon_* \sim \mathcal{N}(0,1)$.}
    \end{itemize}
    
    \item {Table \ref{Table-type-II}. The model settings are as follows:} 
    \begin{itemize}
        \item {Linear Model and Exogeneity Violated: $U = \varepsilon_{U}$, $Z = 2U + \varepsilon_{Z}$, $X = Z+ U + \varepsilon_{X}$, $Y = X + U + \varepsilon_{Y}$, where $\varepsilon_U$, $\varepsilon_Z$ $\sim$ $Gamma(1,2)$, and $\varepsilon_X, \varepsilon_Y \sim Beta(1,2)$.}
        \item {Linear Model and Exogeneity $\&$ Exclusion Restriction Violated: $U = \varepsilon_{U}$, $Z = 3U + \varepsilon_{Z}$, $X = Z+ U + \varepsilon_{X}$, $Y = X + 2Z + U + \varepsilon_{Y}$, where $\varepsilon_U$, $\varepsilon_Z$ $\sim$ $Gamma(1,2)$, and $\varepsilon_X, \varepsilon_Y \sim \mathcal{N}(0,1)$.}
        \item {Partial Non-Linear Model with Constant Causal Effect and Exogeneity Violated: $U = \varepsilon_{U}$, $Z = U^3 + \varepsilon_{Z}$, $X = Z + U + \varepsilon_{X}$, $Y = X + U + \varepsilon_{Y}$, where $\varepsilon_U$, $\varepsilon_Z$ $\sim$ $\mathcal{N}(0,1)$, and $\varepsilon_X, \varepsilon_Y \sim \mathcal{U}[-1,0)\cup(0,1]$.}
        \item {Partial Non-Linear Model with Constant Causal Effect and Exclusion Restriction Violated: $U = \varepsilon_{U}$, $Z = \varepsilon_{Z}$, $X = Z + U + \varepsilon_{X}$, $Y = X + sin(Z) + U^2 + \varepsilon_{Y}$, where $\varepsilon_U$, $\varepsilon_Z$ $\sim$ $\mathcal{N}(0,1)$, and $\varepsilon_X, \varepsilon_Y \sim \mathcal{U}[-1,0)\cup(0,1]$.} 
        \item {Partial Non-Linear Model with Constant Causal Effect and  Exogeneity $\&$ Exclusion Restriction Violated: $U = \varepsilon_{U}$, $Z = U + \varepsilon_{Z}$, $X = Z + U + \varepsilon_{X}$, $Y = X + e^Z + 2U + \varepsilon_{Y}$, where $\varepsilon_U$, $\varepsilon_Z$ $\sim$ $\mathcal{U}[-1,0)\cup(0,1]$, and $\varepsilon_X, \varepsilon_Y \sim \mathcal{N}(0,1)$.}
    \end{itemize}
\end{itemize}
}

\bibliography{reference}

@String{Computing = "Computing" }

@String{Springer = "Springer-Verlag" }

@ArtifactSoftware{R,
    title = {R: A Language and Environment for Statistical Computing},
    author = {{R Core Team}},
    organization = {R Foundation for Statistical Computing},
    address = {Vienna, Austria},
    year = {2019},
    url = {https://www.R-project.org/},
}

@article{cai2019triad,
  title={Triad constraints for learning causal structure of latent variables},
  author={Cai, Ruichu and Xie, Feng and Glymour, Clark and Hao, Zhifeng and Zhang, Kun},
  journal={Advances in neural information processing systems},
  volume={32},
  year={2019}
}

@book{spirtes2000causation,
	title={Causation, Prediction, and Search},
	author={Spirtes, Peter and Glymour, Clark and Scheines, Richard},
	year={2000},
	publisher={MIT press}
}

@book{pearl2009causality,
  title={Causality: Models, Reasoning, and Inference},
  author={ Pearl, Judea},
  edition={2nd},
  publisher={Cambridge University Press},
address={New York},
  year={2009}
}

@article{Sullivant-T-separation,
	title={Trek separation for Gaussian graphical models},
	author={Sullivant, Seth and Talaska, Kelli and Draisma, Jan and others},
	journal={The Annals of Statistics},
	volume={38},
	number={3},
	pages={1665--1685},
	year={2010},
	publisher={Institute of Mathematical Statistics}
}

@inproceedings{spirtes2016causal,
  title={Causal discovery and inference: concepts and recent methodological advances},
  author={Spirtes, Peter and Zhang, Kun},
  booktitle={Applied informatics},
  volume={3},
  pages={1--28},
  year={2016},
  organization={SpringerOpen}
}

@article{shimizu2006linear,
	title   = "A linear non-{G}aussian acyclic model for causal discovery",
	journal = "{Journal of Machine Learning Research}",
	volume  = "7",
	number  = "Oct",
	pages   = "2003--2030",
	year    = "2006",
	doi     = "",
	author  = "Shimizu, Shohei and Hoyer, Patrik O and Hyv{\"a}rinen, Aapo and Kerminen, Antti"
}

@inproceedings{gretton2008kernel,
	title={A kernel statistical test of independence},
	author={Gretton, Arthur and Fukumizu, Kenji and Teo, Choon H and Song, Le and Sch{\"o}lkopf, Bernhard and Smola, Alex J},
	booktitle={Advances in neural information processing systems},
	pages={585--592},
	year={2008}
}

@article{skitovitch1953property,
	title={On a property of the normal distribution},
	author={Skitovitch, VP},
	journal={DAN SSSR},
	volume={89},
	pages={217--219},
	year={1953}
}

@article{darmois1953analyse,
	title={Analyse g{\'e}n{\'e}rale des liaisons stochastiques: etude particuli{\`e}re de l'analyse factorielle lin{\'e}aire},
	author={Darmois, George},
	journal={Revue de l'Institut international de statistique},
	pages={2--8},
	year={1953},
	publisher={JSTOR}
}

@inproceedings{zhang2011kernel,
  title={Kernel-based conditional independence test and application in causal discovery},
  author={Zhang, Kun and Peters, Jonas and Janzing, Dominik and Sch{\"o}lkopf, Bernhard},
  booktitle={Proceedings of the Twenty-Seventh Conference on Uncertainty in Artificial Intelligence},
  pages={804--813},
  year={2011},
  organization={AUAI Press}
}

@inproceedings{chen2017identification,
  title={Identification and model testing in linear structural equation models using auxiliary variables},
  author={Chen, Bryant and Kumor, Daniel and Bareinboim, Elias},
  booktitle={International Conference on Machine Learning},
  pages={757--766},
  year={2017},
  organization={PMLR}
}

@BOOK{Cramer62,
  AUTHOR =       "H. Cram\'{e}r",
  TITLE =        "Random variables and probability distributions",
  PUBLISHER =    "Cambridge University Press",
  YEAR =         "1962",
  address =      "Cambridge",
  edition =      "2nd"
}

@inproceedings{Spirtes:2013:Extended-Trek-Theorem,
 author = {Spirtes, Peter},
 title = {Calculation of Entailed Rank Constraints in Partially Non-linear and Cyclic Models},
 booktitle = {Proceedings of the Twenty-Ninth Conference on Uncertainty in Artificial Intelligence},
 year = {2013},
 location = {Bellevue, WA},
 pages = {606--615},
 numpages = {10},
 publisher = {AUAI Press}
}

@article{salehkaleybar2020learning,
  title={Learning Linear Non-Gaussian Causal Models in the Presence of Latent Variables.},
  author={Salehkaleybar, Saber and Ghassami, AmirEmad and Kiyavash, Negar and Zhang, Kun},
  journal={Journal of Machine Learning Research},
  volume={21},
  number={39},
  pages={1--24},
  year={2020}
}

@article{peters2014identifiability,
  title={Identifiability of Gaussian structural equation models with equal error variances},
  author={Peters, Jonas and B{\"u}hlmann, Peter},
  journal={Biometrika},
  volume={101},
  number={1},
  pages={219--228},
  year={2014},
  publisher={Oxford University Press}
}

@book{bowden1990instrumental,
  title={Instrumental variables},
  author={Bowden, Roger J and Turkington, Darrell A},
  year={1990},
  publisher={Cambridge university press}
}

@article{silva2017learning,
  title={Learning instrumental variables with structural and non-gaussianity assumptions},
  author={Silva, Ricardo and Shimizu, Shohei},
  journal={Journal of Machine Learning Research},
  volume={18},
  number={120},
  pages={1--49},
  year={2017},
  publisher={Microtome Publishing}
}

@article{imbens2015causal,
  title={Causal Inference for Statistics, Social, and Biomedical Sciences: An Introduction.},
  author={Imbens, Guido W and Rubin, Donald B},
  journal={Cambridge University Press},
  year={2015},
  publisher={ERIC}
}

@article{Imbens2014IV,
 author = {Guido W. Imbens},
 journal = {Statistical Science},
 number = {3},
 pages = {323--358},
 publisher = {Institute of Mathematical Statistics},
 title = {Instrumental Variables: An Econometrician's Perspective},
 volume = {29},
 year = {2014}
}

@article{gunsilius2021nontestability,
  title={Nontestability of instrument validity under continuous treatments},
  author={Gunsilius, Florian F},
  journal={Biometrika},
  volume={108},
  number={4},
  pages={989--995},
  year={2021},
  publisher={Oxford University Press}
}

@inproceedings{pearl1995testability,
  title={On the testability of causal models with latent and instrumental variables},
  author={Pearl, Judea},
  booktitle={Proceedings of the Eleventh conference on Uncertainty in artificial intelligence},
  pages={435--443},
  year={1995}
}

@article{kedagni2020generalized,
  title={Generalized instrumental inequalities: testing the instrumental variable independence assumption},
  author={K{\'e}dagni, D{\'e}sir{\'e} and Mourifi{\'e}, Ismael},
  journal={Biometrika},
  volume={107},
  number={3},
  pages={661--675},
  year={2020},
  publisher={Oxford University Press}
}

@inproceedings{chu2001semi,
  title={Semi-instrumental variables: a test for instrument admissibility},
  author={Chu, Tianjiao and Scheines, Richard and Spirtes, Peter},
  booktitle={Proceedings of the Seventeenth conference on Uncertainty in artificial intelligence},
  pages={83--90},
  year={2001}
}

@article{hernan2006instruments,
  title={Instruments for causal inference: an epidemiologist's dream?},
  author={Hern{\'a}n, Miguel A and Robins, James M},
  journal={Epidemiology},
  pages={360--372},
  year={2006},
  publisher={JSTOR}
}

@article{baiocchi2014instrumental,
  title={Instrumental variable methods for causal inference},
  author={Baiocchi, Michael and Cheng, Jing and Small, Dylan S},
  journal={Statistics in medicine},
  volume={33},
  number={13},
  pages={2297--2340},
  year={2014},
  publisher={Wiley Online Library}
}

@book{manski2003partial,
  title={Partial identification of probability distributions},
  author={Manski, Charles F},
  year={2003},
  publisher={Springer Science \& Business Media}
}

@article{kitagawa2015test,
  title={A test for instrument validity},
  author={Kitagawa, Toru},
  journal={Econometrica},
  volume={83},
  number={5},
  pages={2043--2063},
  year={2015},
  publisher={Wiley Online Library}
}

@article{wang2017falsification,
  title={On falsification of the binary instrumental variable model},
  author={Wang, Linbo and Robins, James M and Richardson, Thomas S},
  journal={Biometrika},
  volume={104},
  number={1},
  pages={229--236},
  year={2017},
  publisher={Oxford University Press}
}

@article{palmer2011nonparametric,
  title={Nonparametric bounds for the causal effect in a binary instrumental-variable model},
  author={Palmer, Tom M and Ramsahai, Roland R and Didelez, Vanessa and Sheehan, Nuala A},
  journal={The Stata Journal},
  volume={11},
  number={3},
  pages={345--367},
  year={2011},
  publisher={SAGE Publications Sage CA: Los Angeles, CA}
}

@article{kang2016instrumental,
  title={Instrumental variables estimation with some invalid instruments and its application to Mendelian randomization},
  author={Kang, Hyunseung and Zhang, Anru and Cai, T Tony and Small, Dylan S},
  journal={Journal of the American statistical Association},
  volume={111},
  number={513},
  pages={132--144},
  year={2016},
  publisher={Taylor \& Francis}
}

@inproceedings{drton2004iterative,
  title={Iterative conditional fitting for Gaussian ancestral graph models},
  author={Drton, Mathias and Richardson, Thomas S},
  booktitle={Proceedings of the 20th conference on Uncertainty in artificial intelligence},
  pages={130--137},
  year={2004}
}

@article{zhang2018large,
  title={Large-scale kernel methods for independence testing},
  author={Zhang, Qinyi and Filippi, Sarah and Gretton, Arthur and Sejdinovic, Dino},
  journal={Statistics and Computing},
  volume={28},
  number={1},
  pages={113--130},
  year={2018},
  publisher={Springer}
}

@article{peters2014causal,
  title={Causal Discovery with Continuous Additive Noise Models},
  author={Peters, Jonas and Mooij, Joris M and Janzing, Dominik and Sch{\"o}lkopf, Bernhard},
  journal={Journal of Machine Learning Research},
  volume={15},
  pages={2009--2053},
  year={2014}
}

@article{skaaby2013vitamin,
  title={Vitamin D status, filaggrin genotype, and cardiovascular risk factors: a Mendelian randomization approach},
  author={Skaaby, Tea and Husemoen, Lise Lotte Nystrup and Martinussen, Torben and Thyssen, Jacob P and Melgaard, Michael and Thuesen, Betina Heinsb{\ae}k and Pisinger, Charlotta and J{\o}rgensen, Torben and Johansen, Jeanne D and Menn{\'e}, Torkil and others},
  journal={PloS one},
  volume={8},
  number={2},
  pages={e57647},
  year={2013},
  publisher={Public Library of Science San Francisco, USA}
}

@article{martinussen2019instrumental,
  title={Instrumental variables estimation under a structural Cox model},
  author={Martinussen, Torben and N{\o}rbo S{\o}rensen, Ditte and Vansteelandt, Stijn},
  journal={Biostatistics},
  volume={20},
  number={1},
  pages={65--79},
  year={2019},
  publisher={Oxford University Press}
}

@article{chen2022instrumental,
  title={On instrumental variable regression for deep offline policy evaluation},
  author={Chen, Yutian and Xu, Liyuan and Gulcehre, Caglar and Paine, Tom Le and Gretton, Arthur and De Freitas, Nando and Doucet, Arnaud},
  journal={The Journal of Machine Learning Research},
  volume={23},
  number={1},
  pages={13635--13674},
  year={2022},
  publisher={JMLRORG}
}

@inproceedings{wu2022instrumental,
  title={Instrumental variable regression with confounder balancing},
  author={Wu, Anpeng and Kuang, Kun and Li, Bo and Wu, Fei},
  booktitle={International Conference on Machine Learning},
  pages={24056--24075},
  year={2022},
  organization={PMLR}
}

@article{bowden2016consistent,
  title={Consistent estimation in Mendelian randomization with some invalid instruments using a weighted median estimator},
  author={Bowden, Jack and Davey Smith, George and Haycock, Philip C and Burgess, Stephen},
  journal={Genetic epidemiology},
  volume={40},
  number={4},
  pages={304--314},
  year={2016},
  publisher={Wiley Online Library}
}

@article{windmeijer2019use,
  title={On the use of the lasso for instrumental variables estimation with some invalid instruments},
  author={Windmeijer, Frank and Farbmacher, Helmut and Davies, Neil and Davey Smith, George},
  journal={Journal of the American Statistical Association},
  volume={114},
  number={527},
  pages={1339--1350},
  year={2019},
  publisher={Taylor \& Francis}
}

@article{guo2018confidence,
  title={Confidence intervals for causal effects with invalid instruments by using two-stage hard thresholding with voting},
  author={Guo, Zijian and Kang, Hyunseung and Tony Cai, T and Small, Dylan S},
  journal={Journal of the Royal Statistical Society: Series B (Statistical Methodology)},
  volume={80},
  number={4},
  pages={793--815},
  year={2018},
  publisher={Wiley Online Library}
}

@article{windmeijer2021confidence,
  title={The confidence interval method for selecting valid instrumental variables},
  author={Windmeijer, Frank and Liang, Xiaoran and Hartwig, Fernando P and Bowden, Jack},
  journal={Journal of the Royal Statistical Society: Series B (Statistical Methodology)},
  volume={83},
  number={4},
  pages={752--776},
  year={2021},
  publisher={Wiley Online Library}
}

@article{hartwig2017robust,
  title={Robust inference in summary data Mendelian randomization via the zero modal pleiotropy assumption},
  author={Hartwig, Fernando Pires and Davey Smith, George and Bowden, Jack},
  journal={International journal of epidemiology},
  volume={46},
  number={6},
  pages={1985--1998},
  year={2017},
  publisher={Oxford University Press}
}

@article{bowden2015mendelian,
  title={Mendelian randomization with invalid instruments: effect estimation and bias detection through Egger regression},
  author={Bowden, Jack and Davey Smith, George and Burgess, Stephen},
  journal={International journal of epidemiology},
  volume={44},
  number={2},
  pages={512--525},
  year={2015},
  publisher={Oxford University Press}
}

@article{nakamura1981relationships,
  title={On the relationships among several specification error tests presented by Durbin, Wu, and Hausman},
  author={Nakamura, Alice and Nakamura, Masao},
  journal={Econometrica: journal of the Econometric Society},
  pages={1583--1588},
  year={1981},
  publisher={JSTOR}
}

@book{wooldridge2010econometric,
  title={Econometric analysis of cross section and panel data},
  author={Wooldridge, Jeffrey M},
  year={2010}
}

@article{buhlmann2014cam,
  title={CAM: CAUSAL ADDITIVE MODELS, HIGH-DIMENSIONAL ORDER SEARCH AND PENALIZED REGRESSION},
  author={B{\"u}hlmann, Peter and Peters, Jonas and Ernest, Jan},
  journal={The Annals of Statistics},
  volume={42},
  number={6},
  pages={2526--2556},
  year={2014}
}

@article{acemoglu2001colonial,
  title={The colonial origins of comparative development: An empirical investigation},
  author={Acemoglu, Daron and Johnson, Simon and Robinson, James A},
  journal={American economic review},
  volume={91},
  number={5},
  pages={1369--1401},
  year={2001},
  publisher={American Economic Association}
}

@article{newey2003instrumental,
  title={Instrumental variable estimation of nonparametric models},
  author={Newey, Whitney K and Powell, James L},
  journal={Econometrica},
  volume={71},
  number={5},
  pages={1565--1578},
  year={2003},
  publisher={Wiley Online Library}
}

@misc{card1993using,
  title={Using geographic variation in college proximity to estimate the return to schooling},
  author={Card, David},
  year={1993},
  publisher={National Bureau of Economic Research Cambridge, Mass., USA}
}

@article{zaremba2013b,
  title={B-test: A non-parametric, low variance kernel two-sample test},
  author={Zaremba, Wojciech and Gretton, Arthur and Blaschko, Matthew},
  journal={Advances in neural information processing systems},
  volume={26},
  year={2013}
}

@article{polo2023conditional,
  title={Conditional independence testing under misspecified inductive biases},
  author={Maia Polo, Felipe and Sun, Yuekai and Banerjee, Moulinath},
  journal={Advances in Neural Information Processing Systems},
  volume={36},
  pages={58577--58612},
  year={2023}
}

@inproceedings{wang2025practical,
  title={Practical Kernel Selection for Kernel-based Conditional Independence Test},
  author={Wang, Wenjie and Gong, Mingming and Huang, Biwei and Bailey, James and Han, Bo and Zhang, Kun and Liu, Feng},
  booktitle={The Thirty-ninth Annual Conference on Neural Information Processing Systems},
  year={2025}
}

@inproceedings{doran2014permutation,
  title={A Permutation-Based Kernel Conditional Independence Test.},
  author={Doran, Gary and Muandet, Krikamol and Zhang, Kun and Sch{\"o}lkopf, Bernhard},
  booktitle={UAI},
  pages={132--141},
  year={2014}
}

@article{scheidegger2025residual,
  title={A residual prediction test for the well-specification of linear instrumental variable models},
  author={Scheidegger, Cyrill and Londschien, Malte and B{\"u}hlmann, Peter},
  journal={arXiv preprint arXiv:2506.12771},
  year={2025}
}

@article{breiman2001random,
  title={Random forests},
  author={Breiman, Leo},
  journal={Machine learning},
  volume={45},
  number={1},
  pages={5--32},
  year={2001},
  publisher={Springer}
}

@article{pedregosa2011scikit,
  title={Scikit-learn: Machine learning in Python},
  author={Pedregosa, Fabian and Varoquaux, Ga{\"e}l and Gramfort, Alexandre and Michel, Vincent and Thirion, Bertrand and Grisel, Olivier and Blondel, Mathieu and Prettenhofer, Peter and Weiss, Ron and Dubourg, Vincent and others},
  journal={the Journal of machine Learning research},
  volume={12},
  pages={2825--2830},
  year={2011},
  publisher={JMLR. org}
}

@article{horowitz2012specification,
  title={Specification testing in nonparametric instrumental variable estimation},
  author={Horowitz, Joel L},
  journal={Journal of Econometrics},
  volume={167},
  number={2},
  pages={383--396},
  year={2012},
  publisher={Elsevier}
}

@article{chen2011rate,
  title={On rate optimality for ill-posed inverse problems in econometrics},
  author={Chen, Xiaohong and Reiss, Markus},
  journal={Econometric Theory},
  volume={27},
  number={3},
  pages={497--521},
  year={2011},
  publisher={Cambridge University Press}
}

@book{bain1992introduction,
  title={Introduction to probability and mathematical statistics},
  author={Bain, Lee J and Engelhardt, Max},
  volume={4},
  year={1992},
  publisher={Duxbury Press Belmont, CA}
}

@article{lin1997factorizing,
  title={Factorizing multivariate function classes},
  author={Lin, Juan},
  journal={Advances in neural information processing systems},
  volume={10},
  year={1997}
}

@article{burauel2023evaluating,
  title={Evaluating Instrument Validity using the Principle of Independent Mechanisms},
  author={Burauel, Patrick F},
  journal={Journal of Machine Learning Research},
  volume={24},
  number={176},
  pages={1--56},
  year={2023}
}

@article{voors2012violent,
  title={Violent conflict and behavior: a field experiment in Burundi},
  author={Voors, Maarten J and Nillesen, Eleonora E M and Verwimp, Philip and Bulte, Erwin H and Lensink, Robert and Soest, Daan P Van},
  journal={American Economic Review},
  volume={102},
  number={2},
  pages={941--964},
  year={2012},
  publisher={American Economic Association}
}

@article{guo2016control,
  title={Control function instrumental variable estimation of nonlinear causal effect models},
  author={Guo, Zijian and Small, Dylan S},
  journal={Journal of Machine Learning Research},
  volume={17},
  number={100},
  pages={1--35},
  year={2016}
}

@article{henckel2024graphical,
  title={Graphical tools for selecting conditional instrumental sets},
  author={Henckel, Leonard and Buttenschoen, Martin and Maathuis, Marloes H},
  journal={Biometrika},
  volume={111},
  number={3},
  pages={771--788},
  year={2024},
  publisher={Oxford University Press}
}

@article{blundell2007semi,
  title={Semi-nonparametric IV estimation of shape-invariant Engel curves},
  author={Blundell, Richard and Chen, Xiaohong and Kristensen, Dennis},
  journal={Econometrica},
  volume={75},
  number={6},
  pages={1613--1669},
  year={2007},
  publisher={Wiley Online Library}
}

@article{basmann1957generalized,
  title={A generalized classical method of linear estimation of coefficients in a structural equation},
  author={Basmann, Robert L},
  journal={Econometrica: Journal of the Econometric Society},
  pages={77--83},
  year={1957},
  publisher={JSTOR}
}

@book{shimizu2022statistical,
  title={Statistical Causal Discovery: LiNGAM Approach},
  author={Shimizu, Sh{o}hei},
  year={2022},
  publisher={Springer}
}

@article{bound1995problems,
  title={Problems with instrumental variables estimation when the correlation between the instruments and the endogenous explanatory variable is weak},
  author={Bound, John and Jaeger, David A and Baker, Regina M},
  journal={Journal of the American statistical association},
  volume={90},
  number={430},
  pages={443--450},
  year={1995},
  publisher={Taylor \& Francis}
}

@article{andrews2017examples,
  title={Examples of L2-complete and boundedly-complete distributions},
  author={Andrews, Donald WK},
  journal={Journal of econometrics},
  volume={199},
  number={2},
  pages={213--220},
  year={2017},
  publisher={Elsevier}
}

@article{newey2013nonparametric,
  title={Nonparametric instrumental variables estimation},
  author={Newey, Whitney K},
  journal={American Economic Review},
  volume={103},
  number={3},
  pages={550--556},
  year={2013},
  publisher={American Economic Association}
}

@article{angrist1996identification,
  title={Identification of causal effects using instrumental variables},
  author={Angrist, Joshua D and Imbens, Guido W and Rubin, Donald B},
  journal={Journal of the American statistical Association},
  volume={91},
  number={434},
  pages={444--455},
  year={1996},
  publisher={Taylor \& Francis}
}

@article{hoyer2008nonlinear,
  title={Nonlinear causal discovery with additive noise models},
  author={Hoyer, Patrik and Janzing, Dominik and Mooij, Joris M and Peters, Jonas and Sch{\"o}lkopf, Bernhard},
  journal={Advances in neural information processing systems},
  volume={21},
  year={2008}
}

@article{chernozhukov2007instrumental,
  title={Instrumental variable estimation of nonseparable models},
  author={Chernozhukov, Victor and Imbens, Guido W and Newey, Whitney K},
  journal={Journal of Econometrics},
  volume={139},
  number={1},
  pages={4--14},
  year={2007},
  publisher={Elsevier}
}

@article{d2011completeness,
  title={On the completeness condition in nonparametric instrumental problems},
  author={D’Haultfoeuille, Xavier},
  journal={Econometric Theory},
  volume={27},
  number={3},
  pages={460--471},
  year={2011},
  publisher={Cambridge University Press}
}

@article{darolles2011nonparametric,
  title={Nonparametric instrumental regression},
  author={Darolles, Serge and Fan, Yanqin and Florens, Jean-Pierre and Renault, Eric},
  journal={Econometrica},
  volume={79},
  number={5},
  pages={1541--1565},
  year={2011},
  publisher={Wiley Online Library}
}

@book{peters2017elements,
  title={Elements of causal inference: foundations and learning algorithms},
  author={Peters, Jonas and Janzing, Dominik and Sch{\"o}lkopf, Bernhard},
  year={2017},
  publisher={The MIT press}
}

@article{ren2025regression,
  title={Regression-based Conditional Independence Test with Adaptive Kernels},
  author={Ren, Yixin and Zhang, Juncai and Xia, Yewei and Wang, Ruxin and Xie, Feng and Guan, Jihong and Zhang, Hao and Zhou, Shuigeng},
  journal={Artificial Intelligence},
  pages={104391},
  year={2025},
  publisher={Elsevier}
}

@article{chernozhukov2018double,
  title={Double/debiased machine learning for treatment and structural parameters},
  author={Chernozhukov, Victor and Chetverikov, Denis and Demirer, Mert and Duflo, Esther and Hansen, Christian and Newey, Whitney and Robins, James},
  journal={The Econometrics Journal},
  pages={C1--C68},
  year={2018},
  publisher={JSTOR}
}

@inproceedings{ramdas2015decreasing,
  title={On the decreasing power of kernel and distance based nonparametric hypothesis tests in high dimensions},
  author={Ramdas, Aaditya and Reddi, Sashank Jakkam and P{\'o}czos, Barnab{\'a}s and Singh, Aarti and Wasserman, Larry},
  booktitle={Proceedings of the AAAI Conference on Artificial Intelligence},
  volume={29},
  number={1},
  year={2015}
}

@book{kress2013linear,
  title={Linear Integral Equations},
  author={Kress, Rainer},
  volume={82},
  year={2013},
  publisher={Springer Science \& Business Media}
}

@book{kress1989linear,
  title={Linear integral equations},
  author={Kress, Rainer},
  volume={82},
  year={1989},
  publisher={Springer}
}

@article{hu2018nonparametric,
  title={Nonparametric identification using instrumental variables: sufficient conditions for completeness},
  author={Hu, Yingyao and Shiu, Ji-Liang},
  journal={Econometric Theory},
  volume={34},
  number={3},
  pages={659--693},
  year={2018},
  publisher={Cambridge University Press}
}

@article{chen2014local,
  title={Local identification of nonparametric and semiparametric models},
  author={Chen, Xiaohong and Chernozhukov, Victor and Lee, Sokbae and Newey, Whitney K},
  journal={Econometrica},
  volume={82},
  number={2},
  pages={785--809},
  year={2014},
  publisher={Wiley Online Library}
}

@article{canay2013testability,
  title={On the testability of identification in some nonparametric models with endogeneity},
  author={Canay, Ivan A and Santos, Andres and Shaikh, Azeem M},
  journal={Econometrica},
  volume={81},
  number={6},
  pages={2535--2559},
  year={2013},
  publisher={Wiley Online Library}
}

@article{ai2003efficient,
  title={Efficient estimation of models with conditional moment restrictions containing unknown functions},
  author={Ai, Chunrong and Chen, Xiaohong},
  journal={Econometrica},
  volume={71},
  number={6},
  pages={1795--1843},
  year={2003},
  publisher={Wiley Online Library}
}

@article{hu2008instrumental,
  title={Instrumental variable treatment of nonclassical measurement error models},
  author={Hu, Yingyao and Schennach, Susanne M},
  journal={Econometrica},
  volume={76},
  number={1},
  pages={195--216},
  year={2008},
  publisher={Wiley Online Library}
}

@article{an2012well,
  title={Well-posedness of measurement error models for self-reported data},
  author={An, Yonghong and Hu, Yingyao},
  journal={Journal of Econometrics},
  volume={168},
  number={2},
  pages={259--269},
  year={2012},
  publisher={Elsevier}
}

@article{chen2006identification,
  title={Identification and inference of nonlinear models using two samples with arbitrary measurement errors},
  author={Chen, Xiaohong and Hu, Yingyao},
journal = {Cowles Foundation Discussion Papers},
  year={2006}
}

@article{florens2011identification,
  title={Identification and estimation by penalization in nonparametric instrumental regression},
  author={Florens, Jean-Pierre and Johannes, Jan and Van Bellegem, S{\'e}bastien},
  journal={Econometric Theory},
  volume={27},
  number={3},
  pages={472--496},
  year={2011},
  publisher={Cambridge University Press}
}

@article{carroll2010identification,
  title={Identification and estimation of nonlinear models using two samples with nonclassical measurement errors},
  author={Carroll, Raymond J and Chen, Xiaohong and Hu, Yingyao},
  journal={Journal of nonparametric statistics},
  volume={22},
  number={4},
  pages={379--399},
  year={2010},
  publisher={Taylor \& Francis}
}

@article{shiu2013identification,
  title={Identification and estimation of nonlinear dynamic panel data models with unobserved covariates},
  author={Shiu, Ji-Liang and Hu, Yingyao},
  journal={Journal of Econometrics},
  volume={175},
  number={2},
  pages={116--131},
  year={2013},
  publisher={Elsevier}
}

@article{feve2014non,
  title={Non parametric analysis of panel data models with endogenous variables},
  author={F{\`e}ve, Fr{\'e}d{\'e}rique and Florens, Jean-Pierre},
  journal={Journal of Econometrics},
  volume={181},
  number={2},
  pages={151--164},
  year={2014},
  publisher={Elsevier}
}

@book{meester2008natural,
  title={A natural introduction to probability theory},
  author={Meester, Ronald},
  year={2008},
  publisher={Springer}
}

@article{singh2019kernel,
  title={Kernel instrumental variable regression},
  author={Singh, Rahul and Sahani, Maneesh and Gretton, Arthur},
  journal={Advances in Neural Information Processing Systems},
  volume={32},
  year={2019}
}

@article{bennett2019deep,
  title={Deep generalized method of moments for instrumental variable analysis},
  author={Bennett, Andrew and Kallus, Nathan and Schnabel, Tobias},
  journal={Advances in neural information processing systems},
  volume={32},
  year={2019}
}

@article{xie2022testability,
  title={Testability of instrumental variables in linear non-Gaussian acyclic causal models},
  author={Xie, Feng and He, Yangbo and Geng, Zhi and Chen, Zhengming and Hou, Ru and Zhang, Kun},
  journal={Entropy},
  volume={24},
  number={4},
  pages={512},
  year={2022},
  publisher={MDPI}
}

@article{carrasco2007linear,
  title={Linear inverse problems in structural econometrics estimation based on spectral decomposition and regularization},
  author={Carrasco, Marine and Florens, Jean-Pierre and Renault, Eric},
  journal={Handbook of econometrics},
  volume={6},
  pages={5633--5751},
  year={2007},
  publisher={Elsevier}
}

@article{hall2005nonparametric,
  title={Nonparametric methods for inference in the presence of instrumental variables},
  author={Hall, Peter and Horowitz, Joel L},
 journal = {The Annals of Statistics},
 number = {6},
 volume = {33},
 pages = {2904--2929},
  year={2005}
}

@article{murphy2002estimation,
  title={Estimation and inference in two-step econometric models},
  author={Murphy, Kevin M and Topel, Robert H},
  journal={Journal of Business \& Economic Statistics},
  volume={20},
  number={1},
  pages={88--97},
  year={2002},
  publisher={Taylor \& Francis}
}

@inproceedings{gretton2005measuring,
  title={Measuring statistical dependence with Hilbert-Schmidt norms},
  author={Gretton, Arthur and Bousquet, Olivier and Smola, Alex and Sch{\"o}lkopf, Bernhard},
  booktitle={International conference on algorithmic learning theory},
  pages={63--77},
  year={2005},
  organization={Springer}
}

@article{hu2022simple,
  title={A simple test of completeness in a class of nonparametric specification},
  author={Hu, Yingyao and Shiu, Ji-Liang},
  journal={Econometric Reviews},
  volume={41},
  number={4},
  pages={373--399},
  year={2022},
  publisher={Taylor \& Francis}
}

@article{freyberger2017completeness,
  title={On completeness and consistency in nonparametric instrumental variable models},
  author={Freyberger, Joachim},
  journal={Econometrica},
  volume={85},
  number={5},
  pages={1629--1644},
  year={2017},
  publisher={Wiley Online Library}
}

@article{robin2000tests,
  title={Tests of rank},
  author={Robin, Jean-Marc and Smith, Richard J},
  journal={Econometric Theory},
  volume={16},
  number={2},
  pages={151--175},
  year={2000},
  publisher={Cambridge University Press}
}

@article{glymour2019review,
  title={Review of causal discovery methods based on graphical models},
  author={Glymour, Clark and Zhang, Kun and Spirtes, Peter},
  journal={Frontiers in genetics},
  volume={10},
  pages={524},
  year={2019},
  publisher={Frontiers Media SA}
}

@article{scornet2015consistency,
  title={Consistency of random forests},
  author={Scornet, Erwan and Biau, G{\'e}rard and Vert, Jean-Philippe},
  journal={The Annals of Statistics},
  volume={43},
  number={4},
  pages={1716--1741},
  year={2015}
}

@article{stone1982optimal,
  title={Optimal global rates of convergence for nonparametric regression},
  author={Stone, Charles J},
  journal={The annals of statistics},
  pages={1040--1053},
  year={1982},
  publisher={JSTOR}
}

@article{yang2015minimax,
  title={MINIMAX-OPTIMAL NONPARAMETRIC REGRESSION IN HIGH DIMENSIONS},
  author={Yang, Yun and Tokdar, Surya T},
  journal={The Annals of Statistics},
  pages={652--674},
  year={2015},
  publisher={JSTOR}
}

\end{document}